\documentclass[aps,prd,notitlepage,nofootinbib,superscriptaddress,groupedaddress]{revtex4-2}

\usepackage{mathrsfs}
\usepackage{bm}
\usepackage{amsmath}
\usepackage{amsthm}
\usepackage{amssymb}
\usepackage{graphicx}
\usepackage{esint}
\usepackage{hyperref}
\usepackage{tabularx}

\theoremstyle{plain}

\theoremstyle{plain}

\ifx\proof\undefined
\newenvironment{proof}[1][\protect\proofname]{\par
	\normalfont\topsep6\p@\@plus6\p@\relax
	\trivlist
	\itemindent\parindent
	\item[\hskip\labelsep\scshape #1]\ignorespaces
}{%
	\endtrivlist\@endpefalse
}
\providecommand{\proofname}{Proof}
\fi

\usepackage[table ]{ xcolor}

\usepackage{amssymb}
\usepackage{graphicx,color}

\usepackage{amsmath,amsthm}

\usepackage{mathrsfs}

\usepackage{bm}

\def\beq{\begin{equation}}
\def\eeq{\end{equation}}
\def\bi{\begin{itemize}}
\def\ei{\end{itemize}}
	\def\ba{\begin{array}}
	\def\ea{\end{array}}
	\def\bfig{\begin{figure}}
	\def\efig{\end{figure}}

	\def\C{\mathbb{C}}
	\def\R{\mathbb{R}}
	\def\Z{\mathbb{Z}}

	\newtheorem{theorem}{Theorem}[section]

	\newtheorem{lemma}[theorem]{Lemma}

	\newcommand{\bA}{{\bar{A}}}

	\newcommand{\Slc}{\mathrm{SL}(2,\mathbb{C})}

	\newcommand{\Su}{\mathrm{SU}(2)}

	\def\be{\begin{eqnarray}}
	\def\ee{\end{eqnarray}}

	\newcommand{\cb}{\mathcal B}
	\newcommand{\cc}{\mathcal C}

	\newcommand{\cf}{\mathcal F}
	\newcommand{\cg}{\mathcal G}
	\newcommand{\ch}{\mathcal H}
	
	\newcommand{\cj}{\mathcal J}
	\newcommand{\ck}{\mathcal K}
	
	\newcommand{\cm}{\mathcal M}
	
	\newcommand{\co}{\mathcal O}
	\newcommand{\calp}{\mathcal P}

	\newcommand{\cu}{\mathcal U}

	\newcommand{\cz}{\mathcal Z}

	\newcommand{\sa}{\mathscr{A}}
	\newcommand{\sm}{\mathscr{M}}

	\newcommand{\ff}{\mathfrak{f}}

	\newcommand{\fl}{\mathfrak{l}}  \newcommand{\Fl}{\mathfrak{L}}
	  
	\newcommand{\fn}{\mathfrak{n}}

	  \newcommand{\Fr}{\mathfrak{R}}
	  \newcommand{\Fs}{\mathfrak{S}}

	\renewcommand{\a}{\alpha}
	\renewcommand{\b}{\beta}
	\newcommand{\g}{\gamma}
	\newcommand{\G}{\Gamma}

	\newcommand{\Sig}{\Sigma}
	\renewcommand{\l}{\lambda}
	
	\renewcommand{\o}{\omega}
	\renewcommand{\O}{\Omega}
	\renewcommand{\t}{\tau}

	\newcommand{\rmd}{\mathrm d}

	\newcommand{\lt}{\left}
	\newcommand{\rt}{\right}

	\newcommand{\lag}{\left\langle}
	\newcommand{\rag}{\right\rangle}

	\newcommand{\Ar}{\mathrm{Ar}}

	\newcommand{\re}{\mathrm{Re}}
	\newcommand{\im}{\mathrm{Im}}

	\newcommand{\tr}{\mathrm{Tr}}

\newcommand{\bR}{\overline{R}}

\begin{document}

\title{Lorentzian spinfoam gravity path integral and geometrical area-law entanglement entropy}

\author{Muxin Han}
\email{hanm(At)fau.edu}
\affiliation{Department of Physics, Florida Atlantic University, 777 Glades Road, Boca Raton, FL 33431, USA}
\affiliation{Institut f\"ur Quantengravitation, Universit\"at Erlangen-N\"urnberg, Staudtstr. 7/B2, 91058 Erlangen, Germany}

\begin{abstract}
This paper investigates entanglement entropy in 3+1 dimensional Lorentzian covariant Loop Quantum Gravity (LQG). We compute the entanglement entropy for a spatial region from states dynamically generated by a spinfoam path integral that sums over a family of 2-complexes. The resulting entropy exhibits a geometric area law, $S \simeq \beta a$, where the area $a$ of the entangling surface is determined by the LQG area spectrum and the leading coefficient $\beta>0$ is independent of the underlying 2-complexes. By relating the coupling constant of the sum over 2-complexes to the Barbero-Immirzi parameter $\gamma$, we reproduce the Bekenstein-Hawking formula for the range $0 < \gamma \lesssim 1/2$. This work provides a Lorentzian path integral approach to gravitational entropy without the need for contour prescriptions.

\end{abstract}

\maketitle

\tableofcontents

\section{Introduction}

The intersection of general relativity, quantum mechanics, and thermodynamics has revealed deep connections between the geometry of spacetime and the nature of information. A cornerstone of this relationship is the discovery that black holes possess thermodynamic properties, notably an entropy proportional to their horizon area---the celebrated Bekenstein-Hawking formula. This holographic principle, suggesting that the information content of a gravitational system is encoded on its boundary, was spectacularly generalized by the Ryu-Takayanagi formula, which equates the entanglement entropy of a boundary subregion with the area of a minimal surface in the bulk \cite{Ryu:2006bv}. These developments strongly imply that spacetime geometry itself may emerge from the entanglement structure of underlying quantum degrees of freedom, a viewpoint compellingly articulated by Jacobson's derivation of Einstein's equations from entanglement \cite{Jacobson:2015hqa}.

A challenge in quantum gravity is to derive these profound geometric and entanglement properties from a fundamental, microscopic theory of spacetime. Loop Quantum Gravity (LQG), as a background-independent and non-perturbative quantum gravity theory \cite{ thiemann2008modern,rovelli2014covariant}, offers a promising framework for this task. In LQG, the fundamental quantum excitations of geometry are described by spin-network states: graphs whose links are decorated by representations of SU(2) and whose nodes are decorated by intertwiners. These states form an orthonormal basis for the kinematical Hilbert space and are eigenstates of geometric operators, such as area and volume, which consequently feature discrete spectra \cite{rovelli1995discreteness,ashtekar1997quantum, ashtekar1997quantumII}. The covariant quantum dynamics are encoded in spinfoam amplitudes, which describe the evolution of these spin-network states and form a path integral for quantum gravity \cite{rovelli2014covariant,Perez2012,hanPI,Conrady:2008ea}.

One of the early successes of LQG was the microscopic derivation of the black hole entropy area law by counting the number of spin-network states that endow a horizon with a given area \cite{Rovelli:1996dv, Ashtekar:1997yu, Domagala:2004jt,Agullo:2010zz,ENP,BarberoG:2015xcq}. The work in \cite{lqgee1} derives the area-law entanglement entropy in LQG by relating the entanglement entropy to a similar counting of states. In these results, the area law is \emph{geometrical} in the sense that the entropy is proportional to the surface area given by the area spectrum of LQG quantum geometry, in contrast to the ``topological'' area law: the entropy is proportional to the number of links in the network intersecting the surface, in e.g. \cite{Bianchi:2018fmq,Pastawski:2015qua,Qi1}.

While a significant achievement, these results was largely kinematical, focusing on the properties of states on a fixed boundary without fully incorporating the bulk gravitational dynamics. A more complete and physical picture requires understanding how gravitational entropy, particularly entanglement entropy, arises from states that are solutions to the theory's dynamics. This leads to a crucial question: can we compute the entanglement entropy of a spatial region from a state generated by Lorentzian spinfoam path integral, and does it reproduce the expected area law?

Path integral formulations of gravitational entropy involves the replica trick \cite{Calabrese:2004eu, LM2013,Dong:2016fnf}, which requires calculating the partition function on a replicated manifold with branch-cuts. In the context of gravity, this calculation is typically performed in the Euclidean signature via analytic continuation. However, the Euclidean gravity path integral has various issues, such as it can be ill-defined due to the unboundedness of the Einstein-Hilbert action from below, and the formal Wick rotation from a Lorentzian to a Euclidean signature is not well-understood in a background-independent theory, as well as puzzles in physical interpretation (see \cite{Banihashemi:2024weu} and references therein). 

Lorentzian path integral formulations of gravitational entropy has recently gained prominence. It offers a conceptually sound alternative to the Euclidean path integral by retaining the Lorentzian signature of physical spacetime. Conventionally, these approaches rely on an analytic continuation of the gravitational action, with the entropy emerging from its imaginary part (see e.g. \cite{dche,Colin-Ellerin:2020mva,Dittrich:2024awu,Banihashemi:2024weu}). However, a significant challenge arises within the replica trick: the saddle-point geometry of the replicated manifold is singular at the branching surface. The singular geometry corresponds to a singularity on the integration cycle, which introduces a fundamental sign ambiguity $\pm$ to the entropy. The choice of sign depends on how the path integral contour is prescribed to circumvent this singularity. Although the entropy must be positive, obtaining the correct sign rests on an ad-hoc choice of contour that lacks a clear physical justification.

This paper study the Lorentzian gravitational entropy by using the spinfoam formulation. We compute the entanglement entropy of a spatial region $R$ from a state $\psi$ dynamically generated by Lorentzian spinfoam path integral. We denote by $\Ar(\Fs)=4\pi \g\ell_P^2 a$ the area of $\Fs=\partial R$ given by the area spectrum. For large $a$, the entanglement entropy behaves as the area law at the leading order and accompanied by the logarithmic correction:
\be
S={\b} a+c\log (a)+O(1),\label{arealawentropyS=ba000}
\ee
The coefficient $\b$ of the area law is positive definite, so the entropy does not have the $\pm$ sign ambiguity. Moreover, $\b$ does not depend on the choice of 2-complexes in defining spinfoams, although the coefficient $c$ of logarithmic correction may depend on the boundary graph of the complex. 

The state $\psi$, which is central to our derivation of the entanglement entropy, is a superposition of spin-network states over a family of graphs sharing a common set of nodes. This family is generated from a root graph, $\G$, by varying the multiplicity of links between nodes, where $\G$ itself has an link multiplicity of one. Correspondingly, the spinfoam path integral for $\psi$ is a sum over amplitudes on a family of 2-complexes. These complexes share the same underlying vertices and edges and are generated from a root complex, $\ck$, by varying the number of faces for any given boundary loop, where $\ck$ contains a single face for each such boundary.  We call the state $\psi$ a spin-network stack and the spinfoam path integral a spinfoam stack, since they are obtained by stacking links and faces upon the root graph and 2-complex. Since spin-network stacks are well-defined states in the LQG Hilbert space and spinfoam stacks represent their covariant histories, these constructions are, by definition, physically relevant and well-motivated subjects for analysis. In addition, these constructions can be understood as a coarse-graining process where macroscopic geometry emerges from coarse-grained observables, analogous to how a macrostate is defined in statistical mechanics \cite{Bodendorfer:2018csn,Han:2016xmb,Han:2019emp}. The entropy formula \eqref{arealawentropyS=ba000} provides a connection between classical and quantum theory that does not involve the large spin approximation, upon which the semiclassical analysis of LQG frequently relies.

In our construction, the sum of amplitudes over different 2-complexes is weighted by some ``coupling constants''. When these constants are suitably related to the Barbero-Immirzi parameter $\g$, the entanglement entropy \eqref{arealawentropyS=ba000} reproduces the Bekenstein-Hawking formula
\be
S\simeq \frac{\Ar(\Fs)}{4\ell_P^2},
\ee
This result holds for any $0<\g\lesssim 1/2$.

Our computation of entanglement entropy applies the replica trick to spinfoams. The entropy is computed from the replica partition function which is a sum of spinfoam amplitudes on a family of replicated 2-complexes. There is no singularity in the integration cycles of the spinfoam amplitudes. The sum of these spinfoam amplitudes can be formulated a state-sum that has the same scheme as the state-counting in the computation LQG black hole entropy. The entropy formula \eqref{arealawentropyS=ba000} is derived from this state-sum.

The state $\psi$ evolves from an initial state $\psi_0$ by spinfoam dynamics, but the entanglement entropy \eqref{arealawentropyS=ba000} is remarkably insensitive to the choice of the initial state. This suggests that the area-law gravitational entanglement are a generic outcome of the dynamics. This also seems to suggest an automatic time-averaging mechanism toward the equilibrium inherent in gravity path integral, consistent with the theory's diffeomorphism invariance. 

The paper is organized as follows. Section \ref{Stacking spin-networks} introduces the concept of a spin-network stack as a coarse-graining of spatial geometry. Section \ref{Stacking spinfoams} extends this idea to the spacetime picture, defining the spinfoam stack and its associated amplitude. In Section \ref{Laplace transform}, we discuss the Laplace transform and stationary phase approximation needed to analyze the state sum. Section \ref{Modular entanglement entropy} introduces the formalism of modular entanglement entropy in the context of LQG, building on the algebraic approach suitable for gauge theories. In section \ref{Spinfoam entanglement entropy}, we apply the replica trick to the spinfoam state $\psi$ and derive the area law for entanglement entropy. Section \ref{Spinfoam state with fix areas} discusses an alternative state that fixes macroscopic areas, and the area-law entanglement entropy. In Section \ref{An analog of black hole entropy}, we adapt our formalism to a scenario analogous to a black hole, where the total area of a single entangling surface is fixed, and derive the area-law entanglement entropy, including logarithmic corrections that depend on the discretization of the entangling surface.

\section{Stacking spin-networks}\label{Stacking spin-networks}

Given a closed graph $\G$ embedded in the spatial slice $\Sig$, a spin-network state based on $\G$ is defined by assigning a spin $j=k/2$ ($k\in\Z_+$) to every oriented link $\fl$ and a normalized interwiner $I_\fn$ to every node $\fn$. We assume that in $\G$, any two nodes are connected at most by a single link if they are connected. Focusing on a link $\fl$ connecting a source node $\fn_1=s(\fl)$ to a target node $\fn_2=t(\fl)$, the state can be written as:
\be
\cdots\lt(I_{\fn_2}\rt)^{k;\cdots}_{m;\cdots} \Pi^{k}_{m,n}\lt(H_\fl\rt)\lt(I_{\fn_1}\rt)^{k;\cdots}_{n;\cdots}\cdots,\label{spinnetworkstate}
\ee
where $\Pi^{k}_{m,n}(H_\fl)=\sqrt{d_k} D^{k}_{m,n}(H_\fl)$ with $d_k=k+1$ is the normalized Wigner $D$-function of the SU(2) holonomy $H_\fl$, and $\cdots$ stands for the quantities relating to links other than $\fl$. The magnetic indices $m,n$  are contracted between the intertwiners and the Wigner $D$-function.

We can construct a more general set of spin-network states from $\G$ by simply increasing the number of links connecting $\fn_1,\fn_2$, a generic one among these states is
\be
\cdots\lt(I_{\fn_2}\rt)^{k_1\cdots k_p;\cdots}_{m_1\cdots m_p;\cdots}\prod_{i=1}^p \Pi^{k_i}_{m_i,n_i}\lt(H_{\fl(i)}\rt)\lt(I_{\fn_1}\rt)^{k_1\cdots k_p;\cdots}_{n_1\cdots n_p;\cdots}\cdots.
\ee
There are in total $p$ links $\fl(i)$, $i=1,\cdots,p$, connecting $\fn_1,\fn_2$. Each link carries the SU(2) holonomy $H_{\fl(i)}$ and the spin $k_i/2$. Consequently, the intertwiners $I_{\fn_{1}}$ and $I_{\fn_{2}}$ at the nodes become higher-valent to accommodate the additional links. 

We form the superposition of these spin-networks over the number of links $p$, the spins $\vec{k}=(k_1,\cdots,k_p)$, and the intertwiners (see FIG.\ref{sn_stack} for illustration). To ensure the state is normalizable, the sum is truncated by imposing a constraint: the total LQG area associated with the $p$ links must not exceed a cutoff value, $4\pi\g\ell_P^2A_\fl$ for some arbitrary $A_\fl>0$. This constraint is imposed to every pair of neighboring nodes and to each state in the superposition. The resulting state is a finite linear combination of spin-network states with coefficients $C_{\vec{\mu},\{I_\fn\}_{\fn}}$:
\be
\Psi_{\G,\vec{A}}\lt(\vec{H}\rt)&=&\sum_{\vec \mu}\sum_{\{I_\fn\}_\fn} C_{\vec \mu,\{I_\fn\}_\fn}\prod_{\fl\subset\G} \Theta\left(A_{\fl}-\alpha_{p_\fl,\vec{k}(\fl)}\right) \tr\lt( \bigotimes_{\fn\in\G} I_\fn \cdot \bigotimes_{\fl\subset\G}\lt[ \bigotimes_{i=1}^{p_\fl} \Pi^{k_i(\fl)}\lt(H_{\fl(i)}\rt)\rt]\rt),\label{spinnetworkstack}\\
\vec\mu&=&\lt(\lt\{p_\fl,\vec{k}(\fl)\rt\}\rt)_{\fl},\qquad 
\alpha_{p,\vec{k}}=\sum_{i=1}^{p}\sqrt{k_{i}(k_{i}+2)}.
\ee
In this expression, $p_\fl=\Z_+$ is the number of links replacing the original link $\fl$, and $k_i(\fl)\in\Z_+$ are their spins. The function $\Theta(x)$ is the Heaviside step function: $\Theta(x)=1$ for $x\geq 0$ and $\Theta(x)=0$ for $x< 0$ The trace $\tr$ denotes the complete contraction of all magnetic indices. In a representation-independent form, the state is written as
 \be
|\Psi_{\G,\vec{A}}\rangle =\sum_{\vec{\mu}}\sum_{\{I_\fn\}_\fn}C_{\vec{\mu},\{I_\fn\}_{\fn}}\prod_{\fl\subset\G} \Theta\left(A_{\fl}-\alpha_{p_\fl,\vec{k}(\fl)}\right)\bigotimes_{\fn\in\G} |I_{\fn}\rangle. \label{stackcoarsegr}
\ee
We call $\Psi_{\G,\vec{A}}$ a \emph{spin-network stack}, and the underlying graph $\G$ is referred to as the \emph{root graph}. A spin-network state is a special case of spin-network stack, by setting all $C_{\vec{\mu},\{I_\fn\}_{\fn}}$ to vanish except one. In the LQG Hilbert space that includes all graphs, densely many state can be represented as a linear combination of spin-network stacks (with some $C_{\vec{\mu},\{I_\fn\}_{\fn}}$ and $\vec A$) based on different root graphs.

\begin{figure}[t]
\centering
\includegraphics[width=0.5\textwidth]{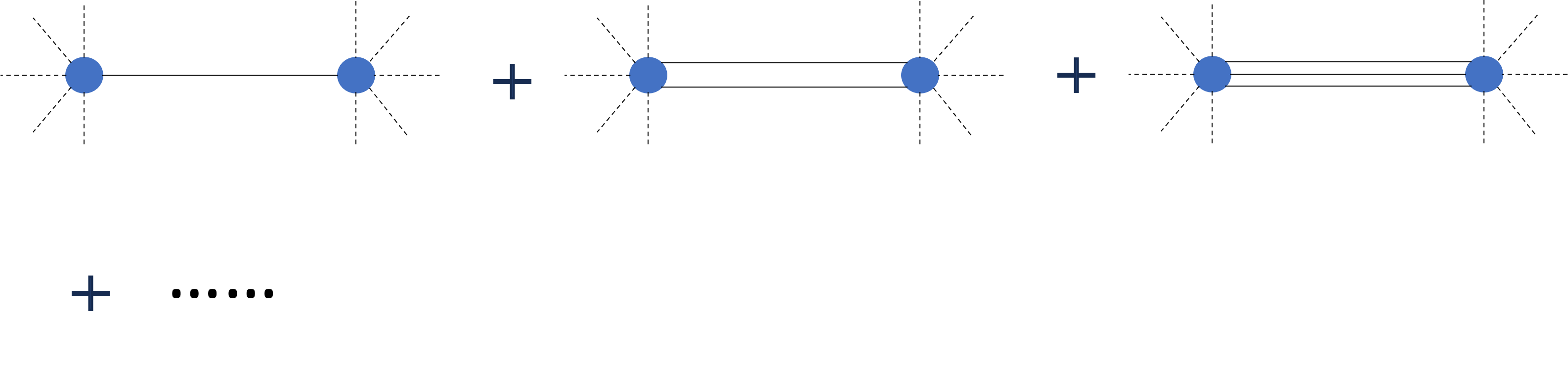}
\caption{The spin-network stack.}
\label{sn_stack}
\end{figure}

The concept of a spin-network stack can be understood in the context of coarse-graining a spatial region. Consider a compact spatial slice $\Sig$ partitioned into a set of closed, 3d subregions $\{R_1, \dots, R_n\}$ such that $\cup_{i=1}^n R_i = \Sig$. Any two regions $R_i$ and $R_j$ intersect only at their boundaries, if at all.

This partition has a dual graph $\G$ embedded in $\Sig$: each region $R_i$ contains a corresponding node $\fn_i$, and for any two adjacent regions, the interface $\Fs_{ij} = R_i \cap R_j$ is dual to a single link $\fl_{ij}$ connecting nodes $\fn_i$ and $\fn_j$. We identify this dual graph $\G$ as the root graph for our construction. We then generalize this duality by allowing the interface $\Fs_{ij}$ between two regions to be dual not just to a single link, but to a collection of arbitrary $p$ links connecting the nodes $\fn_i$ and $\fn_j$.

The collection of spins $\vec{k}$ on these $p$ links endows the interface $\Fs_{ij}$ with a total quantum area given by $\alpha_{p,\vec{k}}$. We now introduce a set of coarse-grained observables by assigning a macroscopic area $A_{\fl_{ij}} \gg 1$ to each interface. A state is constructed by assigning spins and intertwiners to the graph, subject to the constraint that the total quantum area of each interface does not exceed its assigned macroscopic area, i.e., $\alpha_{p,\vec{k}} \leq A_{\fl_{ij}}$. This area cut-off ensures that the number of links with non-trivial spins penetrating any interface is finite.

Let $\G_{\rm max}$ be the maximal graph, defined as the union of all possible graphs that satisfy these area constraints for a given root graph $\G$. This graph has $n$ nodes (one for each region $R_i$) and the maximum number of links allowed by the area cut-offs between each pair of nodes. The Hilbert space $\ch_{\vec A}$ is then constructed upon this maximal graph. It contains all spin-network stacks corresponding to the macrostate defined by the set of areas $\vec A=\{A_{\fl_{ij}}\}$.

This Hilbert space has the following decomposition:
\be
\ch_{\vec A}=\bigoplus_{\{\vec{k}\}}\lt[\bigotimes_{i=1}^n \ch_{R_i,\{\vec k\}}\rt],\label{decompA1An}
\ee
Here, $\{\vec{k}\}$ denotes a specific assignment of spins to all links in the maximal graph $\G_{\rm max}$. The space $\ch_{R_i,\{\vec k\}}$ is the space of allowed intertwiners at the node $\fn_i$ (corresponding to region $R_i$) for that given spin configuration.

In this framework, the assigned areas $\{A_{\fl_{ij}}\}$ are the macroscopic, coarse-grained observables that define a macrostate. The Hilbert space $\ch_{\vec A}$ is then the space containing all the microscopic quantum states consistent with that macrostate.

Beyond surface areas, other quantities can serve as macroscopic observables. For instance, one could use the volumes of the subregions $R_i$, given by the eigenvalues of the corresponding volume operators. Since the volume and area operators commute, it is possible to impose simultaneous constraints on both. This would involve restricting each intertwiner space $\ch_{R_i,\{\vec{k}\}}$ to states with a specific macroscopic volume or an upper volume bound. However, for the purposes of the present discussion, we will not consider volume.

It is also possible to consider more complex coarse-graining schemes. For example, each subregion $R_i$ could contain multiple graph nodes rather than a single one. Nevertheless, the current model, with one node per subregion, represents a case that provides essential insights.

The structure of the Hilbert space reflects the principle of macroscopic locality. The factorization (under direct sum) of $\ch_\Sig$ into a tensor product of spaces associated with individual subregions, as shown in \eqref{decompA1An}, reflects this locality with respect to the partition $\{R_i\}_{i=1}^n$.

In statistical mechanics, coarse-graining often involves fixing a macroscopic observable like energy, where the corresponding microstates of a fixed energy lie within a narrow energy shell, approximating energy eigenstates. By analogy, we can define an alternative spin-network stack that approximates eigenstates of the area operators. This is achieved by replacing the Heaviside step function $\Theta(A_{\fl}-\alpha_{p_\fl,\vec{k}(\fl)})$ in \eqref{stackcoarsegr} with a characteristic function $\chi_\delta(A_{\fl}-\alpha_{p_\fl,\vec{k}(\fl)})$. The new state, denoted $|\Psi'_{\G,\vec{A}}\rangle$, is defined as:
\be
|\Psi'_{\G,\vec{A}}\rangle =\sum_{\vec{\mu}}C_{\vec{\mu},\{I_\fn\}_\fn}\prod_{\fl\subset\G} \chi_\delta\left(A_{\fl}-\alpha_{p_\fl,\vec{k}(\fl)}\right)\bigotimes_{\fn\in\G} |I_{\fn}\rangle. \label{stackcoarsegr2}
\ee
Here, $\chi_\delta(x)$ is the characteristic function for the interval $[-\delta, \delta]$, meaning $\chi_\delta(x)=1$ if $|x| \le \delta$ and is zero otherwise, for some small constant $\delta \sim O(1)$. This modification means the state $|\Psi'_{\G,\vec{A}}\rangle$ is a superposition of microstates whose total interface area $\alpha_{p_\fl,\vec{k}(\fl)}$ falls within a small window $\delta$ of the macroscopic value $A_\fl \gg 1$. Thus, $|\Psi'_{\G,\vec{A}}\rangle$ approximates an eigenstate of the area operators with eigenvalues $\{A_\fl\}_\fl$. This contrasts with the original state $|\Psi_{\G,\vec{A}}\rangle$, where the values $\{A_\fl\}_\fl$ merely served as upper bounds for the areas.

The Hilbert space spanned by these new states, $|\Psi'_{\G,\vec{A}}\rangle$, retains the same decomposition shown in \eqref{decompA1An}; the only difference is that the constraint on the spin configurations $\vec{k}$ is now given by $\chi_\delta(\cdots)$ instead of $\Theta(\cdots)$. In the following analysis, we will consider both types of spin-network stack states, $|\Psi_{\G,\vec{A}}\rangle$ and $|\Psi'_{\G,\vec{A}}\rangle$.

\begin{figure}[t]
\centering
\includegraphics[width=0.25\textwidth]{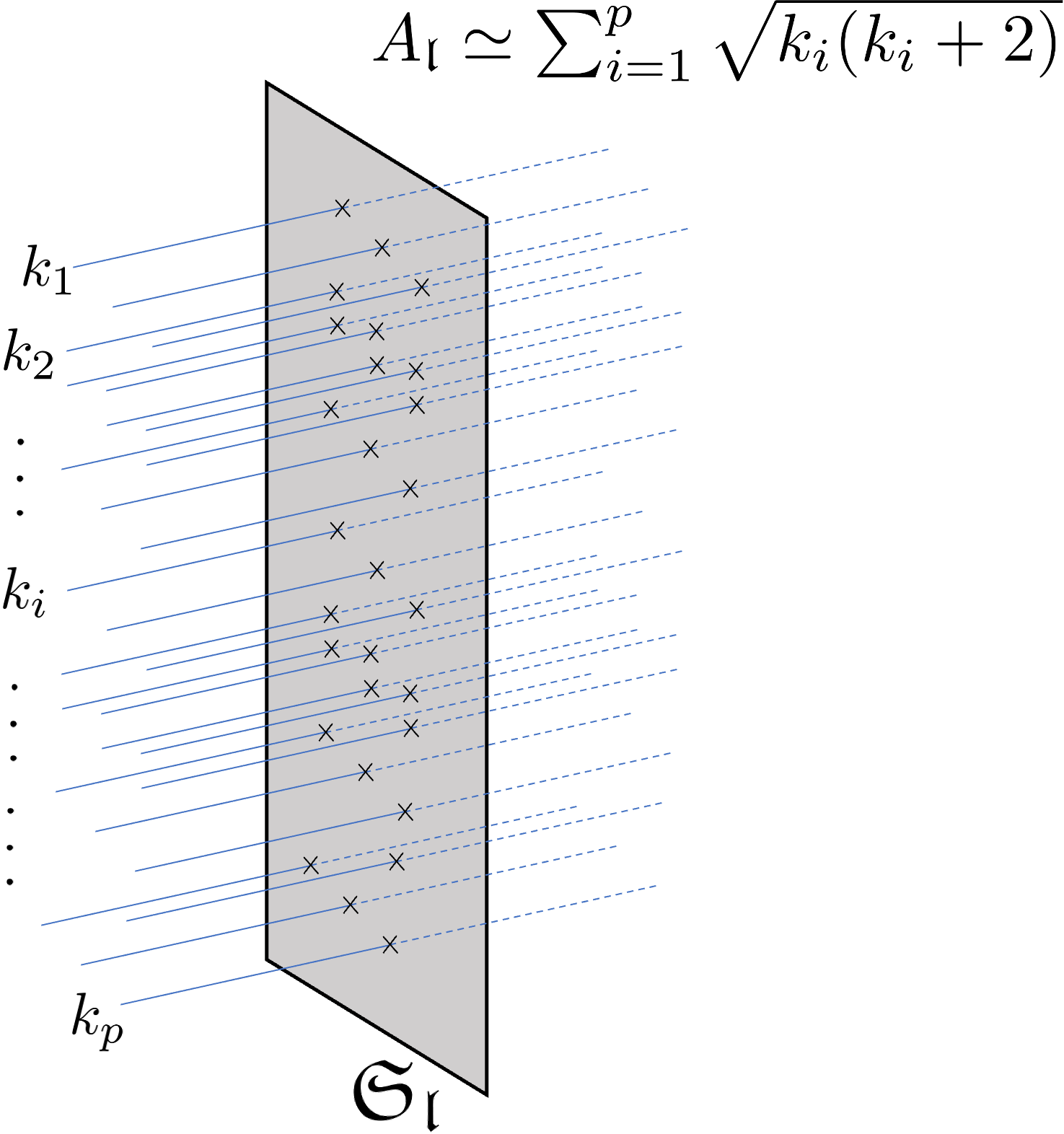}
\caption{Coarse-graining $p$ stacked links with spins $k_1/2,\cdots,k_p/2$ to the macroscopic area $A_\fl$ associated to the surface $\Fs_\fl$.}
\label{grain}
\end{figure}

The entanglement entropy of a subregion $R_i$ provides a measure of its quantum entanglement with the rest of the system. It is calculated by tracing out the microscopic degrees of freedom in the complementary region. This partial trace operation effectively averages over external degrees of freedom, resulting in a loss of information about the precise global state. This act of ignoring fine-grained details to arrive at a coarse-grained quantity--the entanglement entropy itself--is a form of coarse-graining.

This procedure is analogous to how the macrostate is defined: by restricting observations to subsystems. It therefore establishes a possible connection to the coarse-graining framework discussed above, where the macrostate was specified by observables associated with (the boundaries of) the subregions.

The expression of this connection is expected to be the area law for entanglement entropy. The area law posits that the entanglement entropy of a region $R_i$ is proportional to the macroscopic area of its boundary $\partial R_i$. The boundary area is defined by the set of observables $\{A_\fl\}_\fl$ that specifies the macrostate.

\section{Stacking spinfoams}\label{Stacking spinfoams}

Our primary objective is to identify states within the LQG Hilbert space that satisfy the area law for entanglement entropy. We propose that a wide class of spinfoam states, which are generated by the dynamics of LQG, fulfill this condition.

A spinfoam state $\psi$ is constructed by applying a spinfoam evolution operator $\widehat\sa_\ck$ to an initial state $\psi_0$ 
\be
|\psi\rangle = \widehat\sa_\ck \mid \psi_0\rangle,\label{psi8}
\ee
The matrix elements of $\widehat\sa_\ck$ are the spinfoam amplitudes, which define the path integral for quantum gravity in the LQG framework. Thus, $\psi$ is a state that encodes the dynamics of quantum gravity. We will show in Section \ref{Spinfoam entanglement entropy} that the entanglement entropy of such a state indeed satisfies the desired area law.

In the spinfoam formulation of LQG, a spinfoam represents the history of spin-network. In this (3+1)-dimensional picture, the links $\fl$ and nodes $\fn$ of a 3d spin-network evolve and become the faces $f$ and edges $e$ of a 2-complex, respectively. The faces and edges are assigned respectively spins $j_f = k_f/2$ and intertwiners $I_e$, which are the same as the ones on the initial spin-network links and nodes. Conversely, given a spinfoam, a spin-network state on a spatial slice can be viewed as the boundary of the spinfoam, and any spatial cross-section of a spinfoam is a spin-network.

We now extend the concept of a stack to the spacetime picture. A spinfoam stack is defined as a superposition of spinfoams (summed over different 2-complexes) with the defining property that its intersection with any spatial slice yields a spin-network stack. The amplitude for a spinfoam stack, which we term the \emph{stack amplitude}, is the sum of the individual amplitudes of all spinfoams in the superposition. Since the constituent spinfoams are generally built on non-simplicial 2-complexes, their amplitudes are constructed using the KKL formalism \cite{KKL,generalize}. Given that spin-network stacks are well-defined states in the LQG Hilbert space, spinfoam stacks representing the covariant history of spin-network stacks must be taken into account in spinfoam theory.

The construction is analogous to the spatial case: just as a spin-network stack is built by ``stacking'' links upon a root graph $\G$, a spinfoam stack is built by ``stacking'' faces upon a foundational root 2-complex, denoted $\ck$ (see FIG. \ref{sf_stack}). The spacetime manifold partitioned by this root 2-complex is denoted by $\sm_4$. The stack amplitude $\sa_{{\cal K}}$ depends on the root complex $\ck$, the area cut-offs at each faces $\vec{A}=\{A_f\}_f$, the boundary states $\vec \ff$, and some ``coupling constant'' $\vec\l=\{\l_f\}_f$ weighting the sum over complexes. The amplitude has the following integral formula:
\be
\sa_{{\cal K}}\left(\vec{A},\vec{\ff},\vec{\l}\right)&=&\int\prod_{(v,e)}\rmd g_{ve}\prod_{f=h,b}\o_f\left(A_f;\{g_{ve}\},\vec{\ff},\l_f\right)\label{saKampli}\\
\o_f&=&\sum_{p_f=1}^{\infty}\prod_{i=1}^{p_f}\sum_{k_{i}(f)=1}^{\infty}\zeta^{(f)}_{k_{i}(f)}\left(\{g_{ve}\},\ff^{(i)},\l_f\right)\Theta\left(A_{f}-\alpha_{p_f,\vec{k}(f)}\right).\nonumber
\ee
We always use $v,e,f$ to denote the vertices, edges, and faces of the root complex $\ck$. We denote the internal face by $h$ and the boundary faces by $b$. The $p_f$ stacked faces at each $f$, labelled by $i=1,\cdots,p_f$, share the same set of edges and thus the same set of group variables $g_{ve}\in\Slc$, just as the stacked links share the same pair of nodes in spin-network stack. The stacked faces carry spins $k_i(f)/2$, $i=1,\cdots,p_f$.

In the expression \eqref{saKampli}, the sums $\sum_{p_f=1}^\infty$ for all $f\subset \ck$ gives the sum over 2-complexes in the spinfoam stack. All these 2-complexes are obtained from the root complex $\ck$ by stacking faces. The integrals $\int\rmd g_{ve}$ are the Haar integrals over $\Slc$. The spins $k_i(f)/2$ are generally summed even for the boundary faces, since the boundary state is a spin-network stack, which has the sum over spins. The cut-off $A_f$ has been introduced to every face $f\subset\ck$ so that the sums over $p_f$ and $k_i$ are truncated to finite. Only for the boundary face $f=b$, the functions $\o_b$ and $\zeta^{(b)}_{k_{i}}$ has the dependence on the boundary state $\ff^{(i)}=\ff_2^{(i)}\otimes \ff_1^{(i)}\in \ch_{k_i}\otimes\ch_{k_i}^*$, where $i=1,\cdots,p_b$ labels the stacked faces at $b$. For internal face $f=h$, $\o_h$ and $\zeta^{(h)}_{k_{i}}$ do not have dependence on $\vec\ff$. The weights in the sum over complexes relates to the ``coupling constant'' $\l_f>0$, by $\zeta^{(f)}_k\propto \l_f$. The ampitude $\sa_{\ck}$ is a polynomial in $\vec\l=\{\l_f\}_f$ \footnote{The power series in $\vec\l$ is truncated due to the area cut-off.}, where the power of $\l_f$ counts the number of stacked faces on $f\subset \ck$. If we write $\l_f=e^{\Phi_f}$, $\Phi_f$ may be interpreted as a chemical potential. The expression of $\zeta^{(f)}_k$ is given below:

\begin{itemize}
\item Internal face $f=h$:
\be
\zeta_{k}^{(h)}=\l_h \tau_k^{(h)},\qquad \tau_k^{(h)}= d_k\tr_{(k,\rho)}\lt[\overrightarrow{\prod_{v\in\partial h}}P_kg_{ve}^{-1}g_{ve'}P_k\rt],\qquad \rho=\g (k+2)\label{zetakh}
\ee
where $d_k=k+1$ is the dimension of the spin-$k/2$ representation of the subgroup $\Su\subset \Slc$. For any vertex $v$, $e$ or $e'$ is the edge that enter or leave $v$ following the orientation of the face $h$. The trace $\tr_{(k,\rho)}$ is over the Hilbert space $\ch_{(k,\rho)}$ carrying the principal series unitary irrep of $\Slc$. The unitary irrep of $\Slc$ can be decomposed orthogonally into a direct sum of SU(2) irreps: $\ch_{(k,\rho)}\cong  \oplus_{k'=k}^\infty \ch_{k'}$. By this decomposition, we can find an orthonormal basis $|(k,\rho),k',m\rangle$ ($m=-k/2,\cdots,k/2$) known as the canonical basis of $\ch_{(k,\rho)}$. The projection map $P_k$ is given by $P_k:\ch_{(k,\rho)}\to \ch_k$ \cite{hanPI}:
\be
P_k=\sum_{m=-k/2}^{k/2}\big| (k,\rho),k,m\big\rangle\big\langle (k,\rho),k,m\big|\ ,
\ee
where $|(k,\rho),k,m\rangle$ is the basis in $\ch_k\subset\ch_{(k,\rho)}$.

\item Boundary face $f=b$:
\be
\zeta_k^{(b)}=\l_b\tau_k^{(b)},\qquad \tau_k^{(b)}= d_k\lag \ff_1\lt|\overrightarrow{\prod_{v\in\partial b}}P_kg_{ve}^{-1}g_{ve'}P_k\rt|  \ff_2\rag \label{zetakb}
\ee
for some normalized boundary states $\ff_2\otimes \ff_1\in\ch_k\otimes\ch_k^*$. The inner product is on the Hilbert space $\ch_{(k,\rho)}$. For example, in the case that the boundary intertwiners are the coherent intertwiners, $\ff_1,\ff_2$ are coherent states: $\langle\ff^{(i)}_{1}|=\langle k_i,\xi^{(i)}_{e_{1} b}|$ and $|\ff^{(i)}_{2}\rangle=|k_i,\xi^{(i)}_{e_{2} b}\rangle$, where $i=1,\cdots,p_b$ labels the stacked links on the boundary.  $e_1,e_2$ are edges connecting to the boundary. Our following discussion keeps $\vec \ff$ arbitrary.

\end{itemize}

% The function $\zeta_{k_i}$ can expressed as an integral over multiple $\mathbb{CP}^1$:
% \be
% \zeta^{(f)}_{k_{i}}=\l_f d_{k_i}^{V_f+1}\int\limits_{(\mathbb{CP}^1)^{V_f}}\rmd\rho_g \lt(z^{(i)}\rt)e^{k_{i}F_{f}\left(g_{ve},z_{vf}^{(i)};\xi_{eb}^{(i)}\right)},
% \ee
% where $V_f$ is the number of spinfoam vertices in $\partial f$. The expression of $F_f$ is given in e.g. \cite{hanPI,Han:2023cen}. $F_{f}$ depends on $\xi_{eb}^{(i)}$ only for the boundary face $f=b$. The integral is over $V_f$ spinors $z_{vf}^{(i)}\in \mathbb{CP}^1$ with the measure
% \be
% \rmd\rho_g \lt(z^{(i)}\rt)=\prod_{v\in f}\frac{\rmd^{2}z_{vf}^{(i)}}{\lag Z_{vef}^{(i)},Z_{vef}^{(i)}\rag \lag Z_{ve'f}^{(i)},Z_{ve'f}^{(i)}\rag}
% \ee
% In this formula, the subscript $g$ denotes the dependence on $\{g_{ve}\}$, $\mathrm{d^{2}}z=\frac{i}{2}\left(z_{0}dz_{1}-z_{1}dz_{0}\right)\wedge\left(\bar{z}_{0}d\bar{z}_{1}-\bar{z}_{1}d\bar{z}_{0}\right)$, $Z_{vef}^{(i)}=g_{ve}^\dagger z_{vf}^{(i)}$, and $\langle\ ,\ \rangle$ is the Hermitian inner product on $\C^2$. The measure $\rmd\rho_g(z^{(i)})$ is positive for any $\{g_{ve}\}_{(v,e)}$, since $\lag Z,Z\rag>0$.

\begin{figure}[t]
\centering
\includegraphics[width=0.6\textwidth]{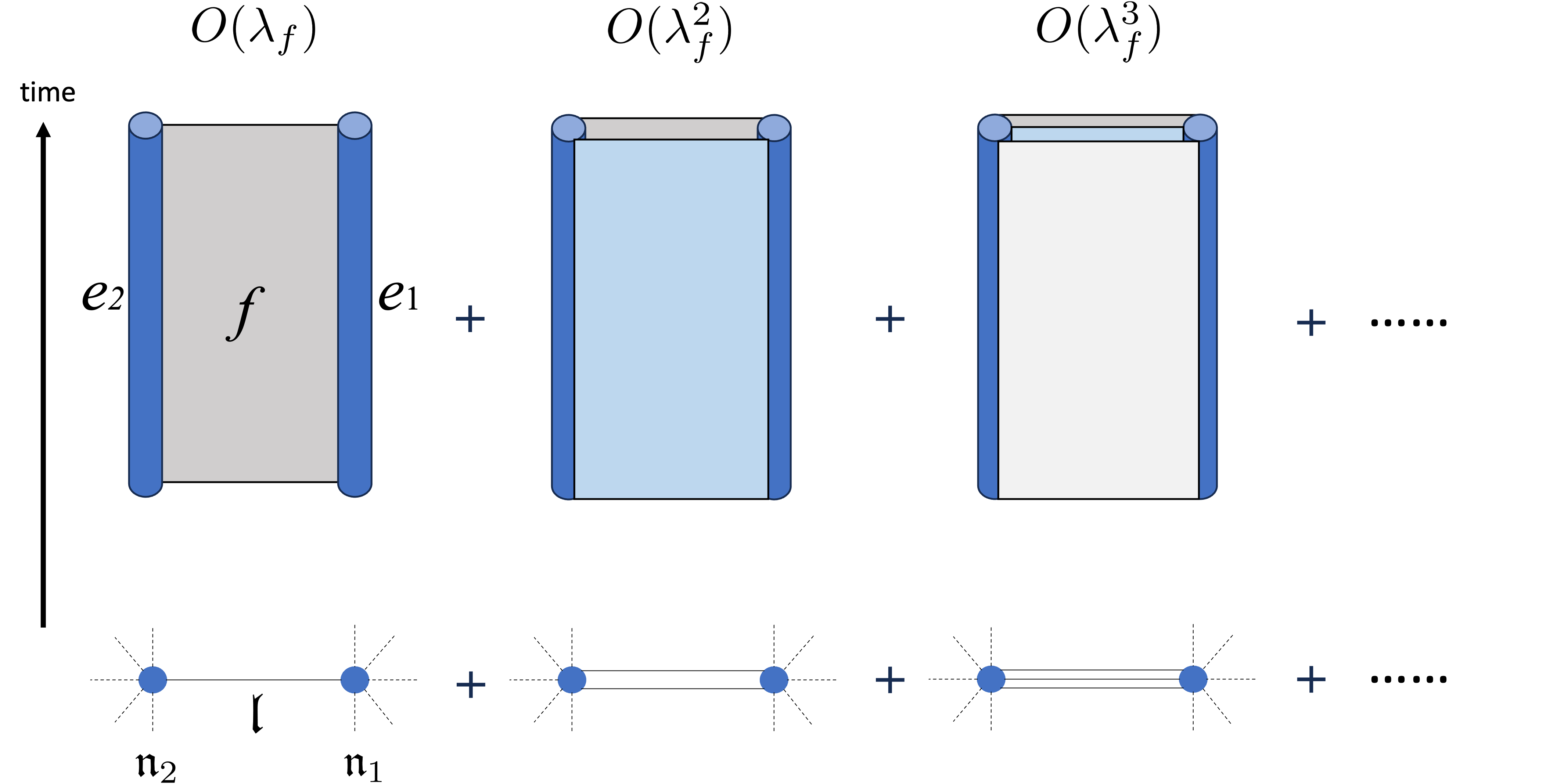}
\caption{The spin-network stack evolves to the spinfoam stack: The spin-network link $\fl$ evolves to the spinfoam face $f$. the spin-network nodes $\fn_1,\fn_2$ evolve to the spinfoam edges $e_1,e_2$. The faces evolves from the dashed links are not shown on this figure. The leftmost complex is the root complex. The power of coupling constant $\l_f$ counts the number of stacked faces.}
\label{sf_stack}
\end{figure}

The amplitude $\sa_\ck$ as a polynomial in $\vec{\l}$ may be viewed as a generating function of spinfoam amplitudes with a certain number of stacked faces at a give $f$. For instance, $\partial_{\l_f} \sa_\ck|_{\l_f=0}$ recovers the amplitude with a single face at $f$.

The integrand of $\sa_\ck$ is invariant under the following continuous gauge transformations
\be
g_{ve}\to x_v g_{ve} u_e,\qquad x_v\in\Slc,\quad u_e\in\Su.
\ee
In order to remove the divergence caused by the $\Slc$ gauge freedom, we choose an edge $e_0$ at every vertex $v$ and fixing $g_{ve_0}=1$ \cite{finite,Kaminski:2010qb} \footnote{Given a spin profile on the entire spinfoam, the sum of all $k_i(f)$'s connecting to an edge $e$ must be an even number constrained by integrating out $ g_{ve}$, otherwise the contribution from the spin profile vanishes. When we express again the nonvanishing partial amplitudes as the integrals over $g_{ve}$, the integrands of these partial amplitudes has the symmetry $g_{ve}\to\pm g_{ve}$ by $\prod D^{k}(-\bm{1})_{m,m'}=(-1)^{\sum k} \prod\delta_{m,m'}$. But this is not a symmetry of the integrand in \eqref{saKampli}, where the sum over spin profiles is not constrained.}. 

The cut-offs $A_f=A_b, A_h$ are the upper bound of the coarse-grained areas allowed for the boundary and intermediate states. In addition, $A_h$ plays the role of regularizing the divergence in the sum over internal spins.

\section{Laplace transform and state sum}\label{Laplace transform}

\subsection{Laplace transform}\label{Internal faces}

In the formulation \eqref{saKampli} of the stack amplitude, $\o_h$ in the integrand sums over the states associated to the stacked internal faces at $h$. The following inverse Laplace transform is useful in computing the state sum  $\o_h$, and we will see in Section \ref{Spinfoam entanglement entropy} that the entanglement entropy is computed in a similar manner.

\begin{lemma}
Given two sequence $\{\a_n\}_{n=1}^\infty $ and $\{\b_n\}_{n=1} ^\infty $ where $\b_n\in\C$ and $\a_n>0$, if $\sum_{n=1}^\infty\left|\beta_{n}\right|e^{-\alpha_{n}\re(s)}<\infty$ for some $\re(s)>0$, we have
\be
\sum_{n,\alpha_{n}< A}\beta_{n}=\sum_{n=1}^{\infty}\beta_{n}\Theta\left(A-\alpha_{n}\right) =\frac{1}{2\pi i}\int_{T-i\infty}^{T+i\infty}\frac{\rmd s}{s}\left[\sum_{n=1}^{\infty}\beta_{n}e^{-\alpha_{n}s}\right]e^{As}.\label{invLaplace}
\ee
for the cut-off $A$ that does not coincide with any $\a_n$. The parameter $T>0$ is greater than the real part of all singularities given by the integrand.
\end{lemma}

\begin{proof}
See \cite{lqgee1}.
\end{proof}

% \begin{proof}
% Let us consider the following Laplace transform for $\re(s)>0$
% \be
% s\int_{0}^{\infty}\rmd A\, e^{-As}\sum_{n=1}^{\infty}\beta_{n}\Theta\left(A-\alpha_{n}\right)=s\sum_{n=1}^{\infty}\beta_{n}\int_{\a_n}^{\infty}\rmd A\, e^{-As}=\sum_{n=1}^{\infty}\beta_{n}e^{-\alpha_{n}s}.
% \ee
% Interchanging the integral and sum is justified by the Fubini-Tonelli theorem and the following absolutely convergent integral (by the hypothesis):
% \be
% \sum_{n=1}^{\infty}\int_{0}^{\infty}\rmd A\left|e^{-As}\beta_{n}\Theta\left(A-\alpha_{n}\right)\right|=\sum_{n=1}^{\infty}\frac{\left|\beta_{n}\right|}{\mathrm{Re}\left(s\right)}e^{-\alpha_{n}\re(s)}.\nonumber
% \ee
% Moreover, given $\re(s)>0$ such that $\sum_{n=1}^\infty\left|\beta_{n}\right|e^{-\alpha_{n}\re(s)}$ converges, we have the uniform bound
% \be
% \sum_{n,\a_n<A}|\beta_n|e^{-A\re(s)}\leq \sum_{n,\a_n<A}|\beta_n|e^{-\alpha_n\re(s)}\leq C,
% \ee
% for some $C>0$. Therefore $|\sum_{n,\a_n<A}\beta_n|\leq C e^{A\re(s)}$. The inverse Laplace transform of $s^{-1}\sum_{n=1}^{\infty}\beta_{n}e^{-\alpha_{n}s}$ is uniquely determined by the Lerch's theorem. Eq.\eqref{invLaplace} is the formula of the inverse Laplace transform.

% \end{proof}

%\subsection{Internal faces}

The formula \eqref{invLaplace} has been used extensively in the computation of LQG black hole entropy, see e.g. \cite{Agullo:2010zz,BarberoG:2008dwr}. We apply this formula to the state-sum in $\o_h$ for internal face $h$: The sum over $n$ in \eqref{invLaplace} corresponds to the sum over $p_h$ and $k_1,\cdots,k_{p_h}$ in $\o_h$, and $\a_n$ corresponds to $\a_{p_h,\vec{k}}$. The summand $\b_n$ corresponds to $\prod_{i=1}^{p_h} \zeta_{k_i}^{(h)}(g_{ve},\l_h)$, where $\zeta_{k_i}^{(h)}(g_{ve},\l_h)$ depends on $i$ only through $k_i$. The hypothesis for \eqref{invLaplace} is valid for $\o_h$:

% For boundary face $f=b$, we assume that the boundary state $\ff^{(i)}_1,\ff^{(i)}_2$ are constant for different $i=1,\cdots,p_b$. 

\begin{lemma}\label{zetabound}
$|\zeta_k^{(h)}|\leq \l_h d_k^2$, the equality holds if and only if all $g_{ve}^{-1}g_{ve'}\in \Su$, $e,e'\subset\partial h$, and $\overrightarrow{\prod}_{v\in\partial h} g_{ve}^{-1}g_{ve'} =\pm 1$.

%and $|\zeta_k^{(b)}|\leq  \l_b d_k$.
\end{lemma}

\begin{proof}
We denote by $D^{(k,\rho)}(g)$ the unitary irrep of $g\in\Slc$. For any state $\ff\in\ch_k\subset \ch_{(k,\rho)}$, we have $\Vert P_k D^{(k,\rho)}(g_{ve}^{-1}g_{ve'})P_k |\,\ff\,\rangle\Vert \leq  \Vert \,\ff\,\Vert $, since $D^{(k,\rho)}(g_{ve}^{-1}g_{ve'})$ is unitary and $P_k$ is an orthogonal projection onto $\ch_k$. Similarly, $\Vert \overrightarrow{\prod}_{v} P_k D^{(k,\rho)}(g_{ve}^{-1}g_{ve'})P_k |\,\ff\,\rangle\Vert \leq  \Vert \,\ff\,\Vert $. The equality holds if and only if all $g_{ve}^{-1}g_{ve'}\in \Su$
\footnote{
Suppose any $g\equiv g_{ve}^{-1}g_{ve'}$ not in SU(2). We use the decomposition $g=k_1 e^{Y} k_2$ where $Y=r \sigma_3/2$ with real $r\neq 0$ and $k_1,k_2\in\Su$. Then $D^{(k,\rho)}(g)P_k |\,\ff\,\rangle=D^{(k,\rho)}(k_1)D^{(k,\rho)}(e^{Y})D^{(k,\rho)}(k_2)|\,\ff\,\rangle$ does not belong to $\ch_k$ for any $\ff\in\ch_k$. To see this, it is sufficient to show $D^{(k,\rho)}(e^{Y})|\,v\,\rangle\not\in\ch_k$ for any $|v\rangle\in\ch_k$ and any real $r\neq 0$. Indeed, every state $v$ in $\ch_k$ can be represented as homogeneous functions of $z_1,z_2,z_1^*,z_2^*$ by $v=(|z_1|^2+|z_2|^2)^{i{\rho}/2-1-k/2}P_k(z_1,z_2)$ where $P_k(z_1,z_2)$ is a holomorphic and homogeneous polynomial of $z_1,z_2$ of degree $k$. $D^{(k,\rho)}(e^{Y}) v=(e^r|z_1|^2+e^{-r}|z_2|^2)^{i{\rho}/2-1-k/2}P'_k(z_1,z_2)$ does not belong to $\ch_k$ because $D^{(k,\rho)}(e^{Y}) v/(|z_1|^2+|z_2|^2)^{i{\rho}/2-1-k/2}$ is not even holomorphic, unless $r=0$.  
%Suppose there exists a vector $|v\rangle\in\ch_k$ such that $D^{(k,\rho)}(e^{Y})|\,v\,\rangle\in\ch_k$, then $D^{(k,\rho)}(e^{Y})$ leaves entire $\ch_k$ invariant, by the commutation relations between $e^Y$ and $J_\pm,J_3\in \fs\fu(2)$. Then it follows that $D^{(k,\rho)}(e^{t Y})$ leaves entire $\ch_k$ invariant
}
. By using this inequality and the Schwartz inequality, we obtain
\be
\lt|\lag \ff_1\lt|\overrightarrow{\prod_{v\in\partial h}}P_kD^{(k,\rho)}\lt(g_{ve}^{-1}g_{ve'}\rt)P_k\rt|  \ff_2\rag\rt|\leq \Vert \ff_1\Vert\Vert \ff_2\Vert=1,\label{schwartzineq}
\ee
for normalized states $\ff_1,\ff_2$, where the equality holds if and only if for all $g_{ve}^{-1}g_{ve'}\in \Su$ and $\overrightarrow{\prod}_{v\in\partial h} g_{ve}^{-1}g_{ve'}\mid  \ff_2\rangle= e^{i\theta}\mid\ff_1 \rangle$ with $\theta\in\R$. Moreover,
\be
\lt|\tr_{(k,\rho)}\lt[\overrightarrow{\prod_{v\in\partial h}}P_kg_{ve}^{-1}g_{ve'}P_k\rt]\rt|\leq  \sum_{m=-k/2}^{k/2}\lt|\lag k,m\lt| \overrightarrow{\prod_{v\in\partial h}} P_k D^{(k,\rho)}(g_{ve}^{-1}g_{ve'})P_k  \rt| k,m\rag\rt|\leq d_k,\label{2ineq}
\ee
The second inequality in \eqref{2ineq} holds if and only if all $g_{ve}^{-1}g_{ve'}\in \Su$ and $\overrightarrow{\prod}_{v\in\partial h} g_{ve}^{-1}g_{ve'} = \exp(i\theta J_3)$ for some $\theta\in\R$, where $J_3|k,m\rangle = m |k,m\rangle$ \footnote{
Denote $O=\overrightarrow{\prod}_{v\in\partial h} P_k D^{(k,\rho)}(g_{ve}^{-1}g_{ve'})P_k $,
the equality $\sum_m |\langle k,m| O |k,m\rangle| = d_k$ and \eqref{schwartzineq} requires $|\langle k,m| O |k,m\rangle| = 1$ for every $m$. This implies that all $g_{ve}^{-1}g_{ve'}$ must be in SU(2). Let $G = \overrightarrow{\prod}_{v\in\partial h} g_{ve}^{-1}g_{ve'} \in \text{SU(2)}$. The operator $O$ now simply becomes the SU(2) representation matrix $D^{(k/2)}(G)$. The condition $|\langle k,m| D^{(k/2)}(G) |k,m\rangle|=1$ also means $|k,m\rangle$ must be an eigenvector of $D^{(k/2)}(G)$. Since this must hold for all basis vectors $|k,m\rangle$, the operator $D^{(k/2)}(G)$ must be diagonal in this basis. The basis $|k,m\rangle$ is the eigenbasis of the generator $J_3$, so $G$ must be a rotation about the z-axis, i.e., $G = \exp(i\theta J_3)$.
}. Then $\tr_{(k,\rho)}\lt[\overrightarrow{\prod}_{v\in\partial h} P_kg_{ve}^{-1}g_{ve'}P_k\rt]=\frac{\sin \left(  (k+1)\theta/2\right)}{\sin \left({\theta }/{2}\right)}$ is the SU(2) character. $|\tr_{(k,\rho)}\lt[\overrightarrow{\prod}_{v\in\partial h}P_kg_{ve}^{-1}g_{ve'}P_k\rt]|=d_k$ if and only if $\theta=0,2\pi$, i.e. $\overrightarrow{\prod}_{v\in\partial h} g_{ve}^{-1}g_{ve'} =\pm 1$.

\end{proof}

% By the Gauss decomposition of $\Slc$, $g_{ve}^{-1}g_{ve'}=R d(r) R'$ where $R,R'\in\Su$ and
% \be
% d(r)=\begin{pmatrix} e^{r/2} & 0\\
% 0 & e^{-r/2}
% \end{pmatrix},\qquad r>0.
% \ee
% On the Hilbert space $\ch_{(k,\rho)}$, the representations of $R,R'$ commute with $P_k$. The representation of $d(r)$ is diagonal and has the following matrix elements in $\ch_k\subset \ch_{(k,\rho)}$:
% \be
% \lag (k,\rho), k,m \big| d(r) \big|(k,\rho), k,n \rag=d(r)^{(k,\rho)}_{k,m}\delta_{m,n}\ ,
% \ee
% The representation is unitary implies $|d(r)^{(k,\rho)}_{k,m}|\leq 1$, where the equality holds at $r=0$ \footnote{The explicit expression of $d(r)^{(k,\rho)}_{k,m}$ is given by \cite{Speziale:2016axj}
% \be
% d(r)^{(k,\rho)}_{k,m}=e^{-\lt(\frac{k}{2}-i \gamma \frac{k}{2}+m+1\rt) r}{ }_2 F_1\left[\frac{k}{2}+m+1, \frac{k}{2}\lt(1-i \gamma\rt)+1,k+2,1-e^{-2 r}\right] .
% \ee
% }. For internal face,
% \be
% \lt|\tr_{(k,\rho)}\lt[\overrightarrow{\prod_{v\in\partial h}}P_kg_{ve}^{-1}g_{ve'}P_k\rt]\rt|\leq \sum_{m=-k/2}^{k/2}\lt|\langle k,m|R'_nR_1|k,m_1\rangle \rt|\delta_{m_1n_1}\lt|\langle k,n_1|R'_1R_2|k,m_2\rangle \rt|
% \ee

\begin{lemma}
For some $\re(s)>0$,
\be
\sum_{p=1}^\infty\sum_{k_1\cdots k_p=1}^\infty\prod_{i=1}^{p}\lt| \zeta^{(h)}_{k_i}\rt|e^{-\a_{p,\vec{k}}\re(s) }<\infty.\label{zetainequality}
\ee

\end{lemma}

\begin{proof}
By the bound of $|\zeta_k^{(h)}|$ proven in Lemma \ref{zetabound}, the left hand side of \eqref{zetainequality} has the bound
\be
\sum_{p=1}^\infty\sum_{k_1\cdots k_p=1}^\infty\prod_{i=1}^{p}\lt| \zeta_{k_i}^{(h)}\rt|e^{-\a_{p,\vec{k}}\re(s) }
\leq  \sum_{p=1}^\infty\lt[\l_f\sum_{k=1}^\infty d_k^{2}e^{-\re(s)\sqrt{k(k+2)}}\rt]^p
\ee
$\sum_{k=1}^\infty d_k^{2}e^{-\re(s)\sqrt{k(k+2)}}$ decrease monotonically as $\re(s)$ grows. The sum in the square bracket is less than 1 for sufficiently large $\re(s)$, then the sum of $p$ converge.

\end{proof}

Applying \eqref{invLaplace} to $\o_h$ gives
\be
\o_h=\frac{1}{2\pi i}\int\limits_{T-i\infty}^{T+i\infty}\frac{\rmd s_h}{s_h}e^{A_h s_h}\sum_{p_h=1}^{\infty}\lt[\sum_{k=1}^\infty\zeta_k^{(h)} e^{-s_h\sqrt{k(k+2)}}\rt]^{p_h}=\frac{1}{2\pi i}\int\limits_{T-i\infty}^{T+i\infty}\frac{\rmd s_h}{s_h}e^{A_h s_h}\frac{\sum_{k=1}^\infty \zeta^{(h)}_k e^{- s_h\sqrt{k(k+2)}}}{1-\sum_{k=1}^\infty\zeta^{(h)}_k e^{- s_h\sqrt{k(k+2)}}}\ .\label{shintegral}
\ee
For a large cut-off $A_h$, the integral is dominated by the pole (in the $s_h$-plane) of the integrand with the largest $\re(s_h)$. The pole with $\re(s_h)>0$ can only be obtained by 
\be
\sum_{k=1}^\infty\zeta^{(h)}_k(g_h)\, e^{- s_h(g_h)\sqrt{k(k+2)}}=1.\label{poleeqn}
\ee
The solution $s_h(g_h)$ depends on the group variables $g_h=\{g_{ve}\}_{e\subset\partial h}$ and $\l_h$. We can obtain the bound of the largest $\re(s_h)$ at the pole by
\be
1\leq \sum_{k=1}^\infty\lt|\zeta^{(h)}_k(g_h) \rt|e^{- \re\lt(s_h(g_h)\rt)\sqrt{k(k+2)}}\leq \l_h\sum_{k=1}^\infty d_k^{2}e^{- \re(s_h(g_h))\sqrt{k(k+2)}}.\label{bound28}
\ee
The right-hand side monotonically decreases as $\re(s)$ grows. Therefore, the solution $s_h(g_h)$ of \eqref{poleeqn} satisfies $\re(s_h(g_h))\leq \b_h$, where $\b_h>0$ satisfies 
\be
\l_h\sum_{k=1}^\infty d_k^{2}e^{- \b_h\sqrt{k(k+2)}}=1,\qquad %\l_b\sum_{k=1}^\infty d_k e^{- \b_b\sqrt{k(k+2)}}=1.
\ee
It is clear that $\b_h$ depends on $\l_h$. When $\re(s_h(g))$ reaches the maximum: $\re(s_h(g))=\b_h$, the inequality \eqref{bound28} implies $\sum_{k=1}^{\infty}|\zeta_{k}^{(h)}(g_h)|e^{-\beta_{h}\sqrt{k(k+2)}}=1$, then the equality $|\zeta_{k}^{(h)}(g_h)|=\lambda_{h}d_{k}^{2}$ in Lemma \ref{zetabound} must hold, and it restricts $g_h$ to
\be
\text{all}\  g_{ve}^{-1}g_{ve'}\in \Su,\quad e,e'\subset\partial h, \qquad \overrightarrow{\prod_{v\in\partial h}} g_{ve}^{-1}g_{ve'} =1\ .\label{g0solution} %\text{for}\ f=h, \qquad \overrightarrow{\prod_{v\in\partial b}} g_{ve}^{-1}g_{ve'}\mid \ff_2\rangle=\mid\ff_1 \rangle\ \text{for}\ f=b.\label{criticalcondi}
\ee
This restriction implies $\zeta_k^{(h)}=\l_h d_k^2$ and $s_h(g)=\b_h$. 

Note that the other case with $\overrightarrow{\prod}_{v\in\partial h} g_{ve}^{-1}g_{ve'}= -1$ is ruled out, because in this case, Eq.\eqref{poleeqn} gives 
\be
\l_h\sum_{k=1}^\infty (-1)^k d_k^{2}e^{- s_h\sqrt{k(k+2)}}=1,
\ee
which implies $\re(s_h)<\b_h$ strictly, see Appendix \ref{Solutions of sum=lambda equations}.

%For large $A_h$, if a $\{g_{ve}\}_{e\subset\partial h}$ satisfying \eqref{g0solution} exists inside the integration domain of $\sa_\ck$, a neighborhood of this $\{g_{ve}\}_{e\subset\partial h}$ contributes dominantly to $\sa_\ck$, since the integrand can grow exponentially as $A_h\to\infty$ at the fastest rate $\b_h$. Indeed, in this neighborhood, 

For large $A_h$, $\o_h$ is given by
\be
&&\o_h= e^{A_h s_h(g_h)}\cf_h(g_h) +\Fr_h ,\qquad \cf_h(g_h)=\frac{1}{s_h(g)\lag \sqrt{k(k+2)} \rag_{h,g}},\label{omegah31}\\
&&\lag \sqrt{k(k+2)} \rag_{h,g}=\sum_{k=1}^\infty \sqrt{k(k+2)} \zeta^{(h)}_k(g_h) e^{- s_h(g_h)\sqrt{k(k+2)}},
\ee
The maximum of $\re(s_h(g_h))$ equals $\beta_h$. All other poles\footnote{Other possible poles includes $s_h=0$ and $s_h(g)$ that solves \eqref{poleeqn} but cannot reach $\re(s_h(g))=\b_h$ (the solution of \eqref{poleeqn} is generally non-unique). } in \eqref{shintegral} have $\re(s_h)$ strictly less than $\b_h$. Their contributions collected by $\Fr_h$ are subleading in $\sa_\ck$ since $e^{A_h\re(s_h)}\ll e^{A_h\b_h}$, as we will discuss in a moment.

\subsection{Stationary phase approximation}\label{Stationary phase approximation}

There are two different asymptotic limits for the stack amplitude $\sa_\ck$: For $h$ an internal face in $\ck$, $A_h$ is the cut-off of internal spins, $A_h\to\infty$ means to remove the cut-off and include all spin in the sum. However, $A_b$ are the coarse-grained areas as macroscopic observables associated to the boundary state. $A_b$ relates to the dimensionful area by $\Ar(b)=4\pi\g \ell_P^2 A_b$, so $A_b\to\infty $ relates to the classical limit $\hbar\to0$. In this subsection, we focus on the limit $A_h\to\infty$ whereas keeping $A_b$ fixed.

We apply the approximation \eqref{omegah31} for $\o_h$ to $\sa_\ck$, neglecting $\Fr_h$:
\be
\sa_\ck\simeq e^{\sum_h \beta_h A_h}\int \prod_{(v,e)}\rmd g_{ve}\, e^{S(g)}\,{\prod_h \cf_h(g_h)}\prod_b\o_b\lt(A_b,\{g_{ve}\}_{e\subset\partial b},\vec\ff,\l_b\rt),\qquad S(g)=\sum_h A_h \lt[s_h(g_h)-\b_h\rt] .\label{AKexponen0}
\ee
We apply the stationary phase approximation to the integration over $\{g_{ve}\}_{e\in E_{\rm int}}$ for uniformly large $A_h$. We denote by $E_{\rm int}$ the set of the edges that does not connect to the boundary. We aim at an asymptotic expansion in cut-offs $A_h$ of internal spins. The ``action'' $S$ relating only to internal faces $h$ only depends on $\{g_{ve}\}_{e\in E_{\rm int}}$. The real part of $S$ reaches the maximum $\re(S)=0$ when \eqref{g0solution} is satisfied by $g_h$ for all $h$, and it implies $S=0$. We denote by $g_{\rm 0,int}(\vec{u})$ the set of $g_{ve}$, $e\in E_{\rm int}$, solving \eqref{g0solution} for all internal faces $h$. $\vec{u}$ parametrizes the solutions in the case that the solution is not unique. We denote by $\cc_{\rm int}$ the space of solutions. $\vec{u}$ are the coordinates on $\cc_{\rm int}$. The space of solutions $\cc_{\rm int}$ closely relates to the moduli space of SU(2) flat connections on the ambient 4-manifold $\sm_4$. The detailed discussion is given in the companion paper \cite{spinfoamstack1}.

The space $\cc_{\rm int}$ is the critical manifold of $S$. Indeed, at any $g_{\rm 0,int}(\vec u)\in \cc_{\rm int}$
\be
\frac{\partial}{\partial{g_{ve}}}\sum_h A_h s_h\Big|_{g_{\rm 0,int}}=\sum_{h;e\subset\partial h}A_h\sum_{k_h=1}^\infty\frac{\partial}{\partial{g_{ve}}}\zeta^{(h)}_{k_h}\Big|_{g_{\rm 0,int}} e^{- \b_h\sqrt{k_h(k_h+2)}}=0.\label{EOMclos}
\ee
To compute the derivative with respect to $g_{ve}$, we deform $g_{ve}\to g_{ve}(1+t_{IJ}J^{IJ})$ and compute the derivative with respect to $t_{IJ}$. The derivative of $\zeta^{(h)}_k$ at $g_{\rm 0,int}$ is proportional to $\tr_{(k,\rho)}[P_kJ^{IJ}P_k]$, which vanishes for all so(1,3) generator $J^{IJ}$. Therefore, $\partial_g S=0$ on $\cc_{\rm int}$.

The exponent $S$ satisfies $\re(S)\leq 0$, and $\re(S)=\partial_g S=0$ only on $\cc_{\rm int}$. Therefore, we only focus on the integral in \eqref{AKexponen0} on a neighborhood $\cu_{\rm 0,int}$ of the critical manifold $\cc_{\rm int}$, where the integral receives the dominant contribution \cite{stationaryphase}. The dominant asymptotic behavior of $\sa_{\ck}$ as $A_h\to\infty$ gives the exponential growth \footnote{$\cf_h=\lt[\b_h\l_h\sum_{k=1}^\infty \sqrt{k(k+2)} d_k^2 e^{- \b_h\sqrt{k(k+2)}}\rt]^{-1}$ is constant on $\cc_{\rm int}$. }
\be
\sa_{\ck}
&=& e^{\sum_h \beta_h A_h}\bA^{-\mathscr{D}_{\rm int}}\int_{\cc_{\rm int}}\rmd \mu(\vec u) \int \prod_{(v,e_b)} \rmd g_{ve_b}\, {\prod_b\o_b|_{g_{\rm 0,int}(\vec u)}}\lt[\phi(\vec{u})+O(\bA^{-1})\rt] ,\label{AKexponen}
\ee
where $e_b$ denotes the edges connecting to the boundary. The exponent $\mathscr{D}_{\rm int}>0$ is one half of the dimension of the Hessian matrix $\partial^2_g S$, and $\phi(\vec{u})$ relates to the determinant of the Hessian matrix. The nondegeneracy of the Hessian matrix is proven in \cite{spinfoamstack1} for simply connected $\ck$. $\bar{A}$ is the mean value of $\{A_h\}_h$. The asymptotics is derived by expressing $A_h=\bA a_h$ and scale $\bA\to \infty$. This asymptotic behavior shows that the divergence of $\sa_\ck$ is at the prefactor $e^{\sum_h \beta_h A_h}\bA^{-\mathscr{D}_{\rm int}}$, which can be removed by normalizing the amplitude or equivalently the state $\psi$ in \eqref{psi8}.

% The formula \eqref{AKexponen} assumes the Hessian matrix on the critical manifold to be nondegenerate. The evidence for this nondegeneracy is discussed in \cite{spinfoamstack1}. However, even if the Hessian matrix has degeneracy, it at most changes $\mathscr{D}_{\rm int}$ to a different rational number and modify the expression of the integral in \eqref{AKexponen}, whereas the exponential growth $e^{\sum_h \beta_h A_h}$ is unchanged. We refer the reader to \cite{arnold2012singularities} and the references therein for the general discussion on the stationary phase approximation, including the situation of degenerate Hessian.

The above derivation ignores $\Fr_h$ in \eqref{omegah31}, which collects the contributions from the poles with $s_h=s_h'(g)$ strictly less than $\b_h$. In the integrand of \eqref{AKexponen0}, the corrections from some $\Fr_{h_i}$'s ($i=1,\cdots,m$ labels a subset of $h$) gives $\exp(S')$ where $S'=\sum_{i=1}^m A_{h_i}[s_{h_i}'(g)-\b_{h_i}]+\sum_{h\neq h_i}A_{h}[s_{h}(g)-\b_{h}]$. But $\re(S')$ is strictly less than zero. Even if the solution of $\partial_g S'=0$ exists, the corresponding contribution to the stationary phase approximation is exponentially small in comparison with \eqref{AKexponen0} or \eqref{AKexponen} when $A_h\gg1$. This justifies ignoring $\Fr_{h}$ in the above discussion.

\section{Modular entanglement entropy}\label{Modular entanglement entropy}

Consider a 4-manifold $\sm_4$ only has two disconnected boundary components $\Sig_0$ and $\Sig$, which are interpreted as the initial and final slices. We assume the root complex $\ck$ to be a partition of $\sm_4$, and the boundaries of $\ck$ are the root graphs $\G\in\Sig$ and $\G_0\in\Sig_0$. Any spinfoam amplitude on $\sm_4$ defines a linear map between the LQG Hilbert spaces on $\Sig_0$ and $\Sig$. Given any initial state $\psi_0$ of spin-network stack on $\Sig_0$, its evolution by the stack amplitude $\sa_\ck$ gives a state $\psi$ on $\Sig$, schematically,
\be
|\psi\rangle = \widehat{\sa}_\ck \mid \psi_0\rangle.\label{spinfoamstate}
\ee
We call $\psi$ the spinfoam state. The detailed expression of the state is given in Section \ref{the spinfoam state}. See FIG.\ref{psievo} for illustration. The states $\psi_0$ and $\psi$ are generally the spin-network stacks based on $\G_0$ and $\G$.

\begin{figure}[h]
\centering
\includegraphics[width=0.2\textwidth]{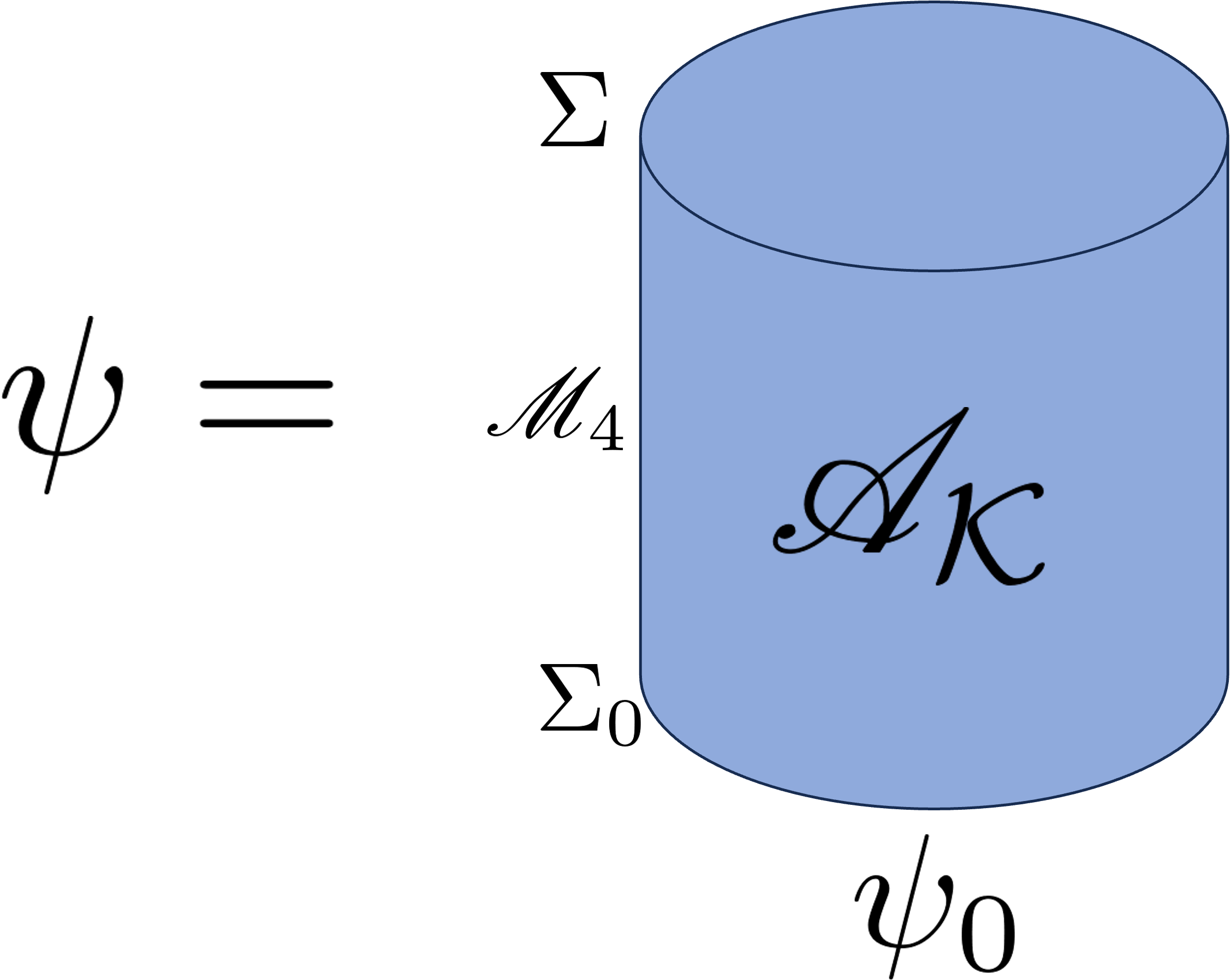}
\caption{A spinfoam state $\psi$ with the initial state $\psi_0$.}
\label{psievo}
\end{figure}

Consider the graph $\G_{\rm max}$ that includes the root graph $\G$ and all graphs $G_\G$ of $\psi$ obtained from $\G$ by stacking links. The maximal graph $\G_{\rm max}$ exists due to the cut-offs $A_b$ in $\sa_\ck$. The Hilbert space $\ch_\Sig$ where $\psi$ lives is spanned by the cylindrical functions on $\G_{\rm max}$:
\be
\ch_\Sig=\bigoplus_{\{k_\fl\in\mathbb{N}_0\}_{\fl\subset \G_{\rm max}}}\bigotimes_{\fn\in \G_{\rm max}}\ch_\fn (\vec{k})
\ee
where $k=2j$ and $\ch_\fn (\vec{k})$ is the Hilbert space of intertwiners at the node $\fn$. The cylindrical function with trivial link ($k_\fl =0$) is allowed in $\ch_\Sig$. The Hilbert space $\ch_{\vec A}$ is a proper subspace in $\ch_\Sig$.

We make an arbitrary bipartition of $\Sig$ into the region $R$ and the complement $\bR$, assuming that the interface $\Fs=\partial R$ only intersect with links of $\G_{\rm max}$. The links intersecting $\Fs$ are denoted by $\fl_*$. The Hilbert space $\ch_\Sig$ has the following decomposition \cite{lqgee1}\footnote{A basis of $\ch_\Sig$ is given by the spin-network states, each of which is a pure tensor products of the intertwiners at the nodes and thus is factorized into $R$ and $\bR$. The direct sums over intertwiners are inside $R$ and $\bR$. The direct sums over spins can be split into the sums over the spins inside $R$ and $\bR$ and the sums over the spins at $\Fs$} 
\be
\ch_{\Sig}=\bigoplus_{\mu} \ch_{R,\mu}\otimes\ch_{\bR,\mu},\qquad \mu=(\calp,\cj),\label{decompchSig}
\ee
where $\calp$ is the set of links $\fl_*$ intersecting with $\Fs$ carrying nonzero spins, and $\cj$ is the set of spins on these links. The Hilbert spaces $\ch_{R,\mu}$ and $\ch_{\bR,\mu}$ contains the bulk degrees of freedoms.

We consider the type-I von Neumann algebra $\cm$ of the observables in $R$ by
\be
\cm=\bigoplus_\mu \cm_\mu,\qquad \cm_\mu= \cb(\ch_{R,\mu})\otimes \bm{I}_{\bR,\mu}
\ee
where $\cb(\ch_{R,\mu})$ denotes the space of bounded operators on $\ch_{R,\mu}$, and $ \bm{I}_{\bR,\mu}$ is the identity operator on $ \ch_{\bR,\mu}$. It is shown in \cite{lqgee1} that $\cm$ is isomorphic to the weakly closure of the algebra generated by the holonomies and fluxes in $R$. The von Neumann algebra $\cm$ is not a factor due to the direct sum over $\mu$: The algebra $\cm$ has a non-trival center generated by area operators at the punctures on $\Fs$. The data $\mu$ collects the eigenvalues of these area operators.

For any $\bm \co=\oplus_\mu\bm \co_\mu\in\cm$ where $\bm \co_\mu=\bm \co_{R,\mu}\otimes \bm I_{\bR,\mu}\in\cm_\mu$ such that $\bm \co_{R,\mu}$ is in trace-class, we define the renormalized trace by \cite{lqgee1,Colafranceschi:2023moh}
\be
\hat{\tr}(\bm \co)=\sum_\mu \hat{\tr}_\mu (\bm \co_\mu), \quad \text{where}\quad \hat{\tr}_\mu(\bm \co_\mu)=\kappa_\mu \tr_{R,\mu} (\bm \co_{R,\mu}),\quad \kappa_\mu\in\Z_+. 
\ee
where $\tr_{R,\mu}$ is the standard trace on $\ch_{R,\mu}$. Here the set of $\kappa_\mu$ is the freedom in defining the renormalized trace on non-factor type-I algebra, and $\kappa_\mu$ can be interpreted as the dimension of a hidden Hilbert space relating to the edge modes at $\Fs$. There is a canonical choice of the hidden Hilbert space to be $\otimes_{\fl_*}\ch_{k_{\fl_*}}$ carrying the SU(2) gauge transformations at $\Fs$, then
\be
\kappa_\mu=\prod_{\fl_*}d_{k_{\fl_*}}.
\ee
Under this choice, the renormalized trace coincides with the standard partial trace of $R$ on a enlarged Hilbert space that relaxes the gauge invariance at $\Fs$ \cite{lqgee1,HanTobin,Lin:2018bud}.

Given any state $\psi=\sum_\mu\psi_\mu\in\ch_\Sig$, the pure-state density operator is given by $\bm \rho =|\psi\rangle\langle\psi|$. A difference between the discussions below and in \cite{lqgee1} is that we do not assume the state $\psi$ is normalized. We only focus on the diagonal block $\bm\rho_{\mu\mu}=|\psi_\mu \rangle\langle\psi_\mu|$, since only the diagonal blocks contributes to the expectation value of any operator $\bm{\co}\in \cm$. We define the reduced density matrix $\bm\rho_{R,\mu}$ in each $\mu$-sector and the density operator $\hat{\bm\rho}_\cm\in\cm$ by 
\be
\bm\rho_{R,\mu}=\mathrm{Tr}_{\bR,\mu}(\bm\rho_{\mu\mu}),\qquad \hat{\bm\rho}_\cm=\bigoplus_\mu \kappa_\mu^{-1}\bm\rho_{R,\mu}\otimes\bm{I}_{\bR,\mu}.
\ee
Note that both $\bm\rho_{R,\mu}$ and $\hat{\bm\rho}_\cm$ are not normalized. It is clear that $\hat{\tr}(\hat{\bm\rho}_\cm)=\tr(\bm\rho)$ where $\tr$ is the standard trace on $\ch_\Sig$.

We define the replica partition function ${\cal Z}_{n}$ by
\be
{\cal Z}_{n}=\hat{\mathrm{Tr}}\left(\hat{\bm{\rho}}_{{\cal M}}^{n}\right)=\sum_{\mu}\kappa^{1-n}_\mu\cz_{\mu,n},\qquad \cz_{\mu,n}=\mathrm{Tr}_{R,\mu}\lt(\bm{\rho}_{R,\mu}^{n}\rt).
\ee
The partition function $\cz_n$ can be viewed as an analog of the thermal partition function with the temperature $1/n$: Given that every $\bm\rho_{R,\mu}$ is a positive operator on $\ch_{R,\mu}$, we write $\bm\rho_{R,\mu}=\exp(-\bm{H}_{R,\mu})$ where $\bm{H}_{R,\mu}$ is the modular Hamiltonian in the $\mu$-sector. The modular Hamiltonian of $\hat{\bm\rho}_\cm$ is given by
\be
\hat{\bm H}_R=\bigoplus_\mu\lt[\bm{H}_{R,\mu}+\log(\kappa_\mu)\bm{I}_{R,\mu}\rt]\otimes \bm{I}_{\bR,\mu},\qquad \hat{\bm\rho}_\cm=\exp\lt(-\hat{\bm H}_R\rt).
\ee
The partition function $\cz_n=\hat{\tr}[\exp(-n \hat{\bm H}_R)]$ with the analytically continued $n$ can be viewed as a canonical partition function with the Hamiltonian $\hat{\bm H}_R$ and the inverse temperature $n$. The thermal entropy from $\cz_n$ with respect to the temperature is called the modular entanglement entropy of the state $\psi$ \cite{Dong:2016fnf}:
\be
\widetilde{S}_n(\bm\rho,R)=-n^2\partial_n\lt[\frac{1}{n}\log\lt(\cz_n\rt)\rt].
\ee
The explicit relation between $\widetilde{S}_n$ and the density operator $\hat{\bm\rho}_\cm$ can be derived
\be
\widetilde{S}_n(\bm\rho,R)	
&=&\log\lt(\cz_n\rt)+n\cz_n^{-1}\hat{\tr}\lt[\hat{\bm H}_R\exp\lt(-n \hat{\bm H}_R\rt)\rt]=\log\lt(\cz_n\rt)-\cz_n^{-1}\hat{\tr}\lt(\hat{\bm\rho}^n_\cm\log \hat{\bm\rho}_\cm^n\rt)\nonumber\\
&=&-\hat{\mathrm{Tr}}\left(\frac{\hat{\bm{\rho}}_{{\cal M}}^{n}}{{\cal Z}_{n}}\log\frac{\hat{\bm{\rho}}_{{\cal M}}^{n}}{{\cal Z}_{n}}\right)
\ee
where ${\hat{\bm{\rho}}_{{\cal M}}^{n}}/{{\cal Z}_{n}}$ is the analog of the thermal density operator. When $n\to 1$, the thermal density operator ${\hat{\bm{\rho}}_{{\cal M}}^{n}}/{{\cal Z}_{n}}$ reduced to the normalized density operator, then $\widetilde{S}_n(\bm\rho,R)$ reduces to the von Neumann entanglement entropy $S(\bm\rho,R)$ with respect to $\cm$ \cite{lqgee1}.

\section{Entanglement entropy of spinfoam state}\label{Spinfoam entanglement entropy}

\subsection{The spinfoam state}\label{the spinfoam state}

In this section, we apply the above formalism to compute the entanglement entropy of the spinfoam state $\psi$ in \eqref{spinfoamstate}. The spinfoam state $\psi$ has the following expression
\be
\psi=\int \prod_{(v,e)}\rmd g_{ve}\, \prod_{f = h,b,b_0}\o_f , \qquad \o_f=\sum_{p_f=1}^{\infty}\prod_{i=1}^{p_f}\sum_{k_{i}=1}^{\infty}\zeta^{(f)}_{k_{i}}\Theta\left(A_{f}-\alpha_{p_f,\vec{k}}\right) .\label{psistate00}
\ee
The spinfoam has two boundaries at $\Sig$ and $\Sig_0$. We denote by $b$ and $b_0$ the boundary faces intersecting with $\Sig$ and $\Sig_0$ respectively. The intersections are the links of $\G$ and $\G_0$. The initial stack state $\psi_0$ is based on the graphs generated from the root graph $\G_0$. $\psi$ is a function of SU(2) holonomies on the graphs generated from $\G$. In \eqref{psistate00}, $\zeta_k^{(f)}$ for $f=h,b,b_0$ are given respectively by
\be
\zeta_{k_i}^{(h)}&=&\l_h d_{k_i}\tr_{(k_i,\rho_i)}\lt[\overrightarrow{\prod_{v\in\partial h}}P_{k_i}g_{ve}^{-1}g_{ve'}P_{k_i}\rt],\qquad
\zeta_{k_i}^{(b)}=\l_b d_{k_i}\tr_{({k_i},\rho_i)}\lt[\lt(\overrightarrow{\prod_{v\in\partial b}}P_{k_i}g_{ve}^{-1}g_{ve'}P_{k_i}\rt) H_{\fl(i)}\rt],\\
\zeta_{k_i}^{(b_0)}&=&\l_{b_0} d_{k_i}\lag \ff_1^{(i)}\lt| \overrightarrow{\prod_{v\in\partial b_0}}P_{k_i}g_{ve}^{-1}g_{ve'}P_{k_i}\rt|  \ff_2^{(i)}\rag.
\ee
The links $\fl(i)$, $i=1,\cdots,p_b$, are in $\Sig$ and belong to the boundary of the stacked faces associated to $b$. In the trace, the SU(2) holonomy $H_{\fl(i)}$ acts on $\ch_{k_i}\subset \ch_{(k_i,\rho_i)}$. The state $\psi$ is a function of the SU(2) holonomies $H_{\fl(i)}$ in $\Sig$. $\ff_{1}^{(i)},\ff_2^{(i)}$  relate to the initial state $\psi_0$.

Due to the bipartition $\Sig=R\cup\bR$, we pick up the faces $b$ whose boundary on $\Sig$ intersecting with the interface $\Fs$. In the following, we use $b\cap\Fs\neq \emptyset$ to characterize these (root) faces. $\mu=(\calp,\cj)$ in \eqref{decompchSig} is given by the spin data associated to these faces.
\be
\mu=\lt(\{p_b,\vec{k}(b)\}\rt)_{b\cap \Fs\neq \emptyset},
\ee
where $\vec{k}(b)=(k_1,\cdots,k_{p_b})$ and $p_b,k_i\in\Z_+$. 

When we write $\psi=\sum_\mu \psi_\mu$, the state in each $\mu$-sector, $\psi_\mu\in \ch_{R,\mu}\otimes\ch_{\bR,\mu}$, is given by 
\be
\psi_\mu=\int \prod_{(v,e)}\rmd g_{ve}\prod_{f \cap \Fs= \emptyset}\o_f \prod_{b\cap \Fs\neq \emptyset}\prod_{i=1}^{p_b}\zeta^{(b)}_{k_{i}}\Theta\left(A_{b}-\alpha_{p_b,\vec{k}}\right).\label{psimustate00}
\ee

The density operator contains both $|\psi\rangle$ and the conjugate state $\langle\psi|$. The complex conjugate $\psi^*$ is given by $\psi^*=\sum_\mu \psi_\mu^*$ and
\be
\psi_\mu^*=\int \prod_{(v,e)}\rmd g_{ve}\prod_{f \cap \Fs= \emptyset}\o_f^* \prod_{b\cap \Fs\neq \emptyset}\prod_{i=1}^{p_b}\zeta^{(b)*}_{k_{i}}\Theta\left(A_{b}-\alpha_{p_b,\vec{k}}\right).
\ee
The integrand contains the complex conjugates of $\zeta_k^{(h)}$ and $\zeta_k^{(b)}$:
\be
\zeta_{k_i}^{(h)*}&=&\l_h d_{k_i}\tr_{(k_i,\rho_i)}\lt[\overleftarrow{\prod_{v\in\partial h}}P_{k_i} g_{ve'}^\dagger (g_{ve}^\dagger)^{-1}P_{k_i}\rt],\qquad 
\zeta_{k_i}^{(b)*}=\l_b d_{k_i}\tr_{({k_i},\rho_i)}\lt[\lt(\overleftarrow{\prod_{v\in\partial b}}P_{k_i}g_{ve'}^\dagger (g_{ve}^\dagger)^{-1}P_{k_i}\rt) H_{\fl(i)}^\dagger\rt],\\
\zeta_{k_i}^{(b_0)*}&=&\l_{b_0} d_{k_i}\lag \ff_2^{(i)}\lt| \overleftarrow{\prod_{v\in\partial b_0}}P_{k_i}g_{ve'}^\dagger (g_{ve}^\dagger)^{-1} P_{k_i}\rt|  \ff_1^{(i)}\rag.
\ee
By the invariance of the Haar measure $\rmd g_{ve}=\rmd (g_{ve}^\dagger)^{-1}$, the conjugate state $\psi_\mu^*$ are given by $\psi_\mu$ in \eqref{psimustate00} with flipping orientations of all faces. 

\begin{figure}[t]
\centering
\includegraphics[width=0.7\textwidth]{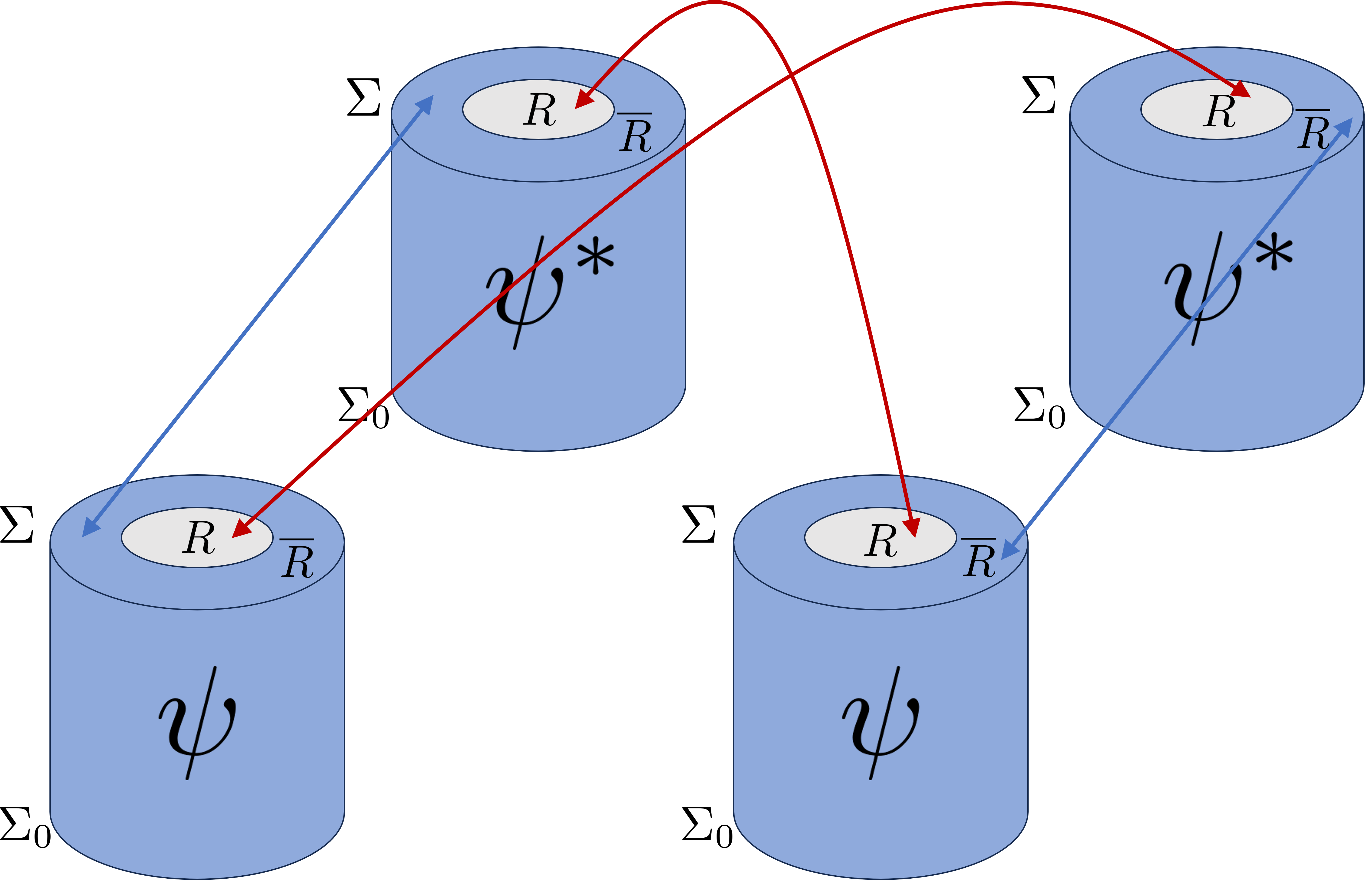}
\caption{The replica partition function $\cz_{\mu,n}$ for $n=2$: The gluing with the blue arrows is due to $\mathrm{Tr}_{\bR,\mu}$ defining the reduced density operator $\bm\rho_{R,\mu}$. The gluing with the red arrows corresponds to the product of $\bm\rho_{R,\mu}$ and the trace $\mathrm{Tr}_{R,\mu}$. The gluing gives a manifold with branch surface $\Fs=R\cap\bR$.}
\label{replica}
\end{figure}

\subsection{Replica trick}\label{Replica trick}

We define $\bm \rho =|\psi\rangle\langle\psi|$ with the diagonal $\bm\rho_{\mu\mu}=|\psi_\mu \rangle\langle\psi_\mu|$. In each $\mu$-sector, the reduced density operator is $\bm\rho_{R,\mu}=\mathrm{Tr}_{\bR,\mu}(\bm\rho_{\mu\mu})$, and the replica partition function $\cz_{\mu,n}=\tr_{R,\mu}(\bm\rho_{R,\mu}^n)$ is given by gluing $n$ pairs of $\psi_\mu,\psi^*_\mu$ according to the replica trick FIG.\ref{replica}. We denote by $\ck_n$ the 2-complex obtained by the replica gluing of $2n$ copies of $\ck$'s as in FIG.\ref{replica}. The replica partition function $\cz_{\mu,n}$ is a stack amplitude based on the root complex $\ck_n$ with the fixed spin profiles $\mu$ at the branch surface $\Fs$, while the full partition function $\cz_{n}$ sums $\cz_{\mu,n}$ over $\mu$. $\cz_{\mu,n}$ can be expressed by using the following ingredients:

\begin{enumerate}

\item We label each of the $n$ copies of $\psi_\mu$ by $(I,+)$, $I=1,\cdots,n$ and label each of the $n$ copies of $\psi^*_\mu$ by $(I,-)$, $I=1,\cdots,n$. The associated integration variables are $g_{ve}^{I,\pm}\in\Slc$ where $(v,e)$ belongs to the root complex $\ck$ and $i=1,\cdots,n$. We denote the product of Haar measure by
\be
\lt[\rmd g\rt]_n=\prod_{(v,e),(I,\pm)}\rmd g_{ve}^{I,\pm}\ .
\ee
The collection $\{g_{ve}^{I,\pm}\}$ contains all the integration variables on $\ck_n$. The partition function $\cz_{\mu,n}$ is an integral with the measure $[\rmd g]_n$.

\item Consider any face $f\in\ck$ that has no intersection with the spatial slice $\Sig$. Here $f$ may be an internal face $f=h$ or a boundary face $f=b_0$ only intersecting with $\Sig_0$. From $f$, the replica trick generates $2n$ faces $f_I^\pm$, $I=1,\cdots,n$, in $\ck_n$. If we still denote a face in $\ck_n$ by $f$, any $f$ that has no intersection with $\Sig$ contributes to the integrand of $\cz_{\mu,n}$
\be
\o_f\lt(A_f;g_{ve}^{I,\pm},\vec{\ff},\l_f\rt),
\ee
$\o_f$ for $f=f_I^+$ and $f=f_I^-$ associate to two opposite orientations along $\partial f$.

\item Given any boundary face $b^+$ in $\psi_\mu$ that intersects with $\Sigma$ (in $R$ or $\bR$) but does not intersect with $\Fs$: $b^+\cap \Fs=\emptyset$, it glues with the corresponding face $b^-$ in $\psi^*$ ($b^\pm$ are the same face with opposite orientation), and the gluing is given by the inner product\footnote{The trace in $\xi_k^{(b)}$ can be understood as a trace over $\ch_k$ of SU(2) irrep due to $P_k$. The operator $\overrightarrow{\prod}_{v\in\partial b}P_{k_i}g_{ve}^{-1}g_{ve'}P_{k_i}$ is an operator on $\ch_k$. Then for any operators $A\in\cb(\ch_{k'})$ and $B\in\cb(\ch_k)$, $\int \rmd H\,\tr_{k'}(A H)\tr_k( H^{-1} B)=A_{i j}B_{m n}\int \rmd H\, D^{k'}(H)_{ji}D^k(H)^*_{m n}=A_{i j}B_{m n}d_k^{-1}\delta^{kk'}\delta_{jm}\delta_{in}=d_k^{-1}\delta^{kk'}\tr_k(AB)$. The traces are ``glued'' by canceling $h$ and $H^{-1}$. FIG.\ref{bpbm} may be seen as the graphical illustration of this relation.}
\be
\sum_{k_i,k_i'}\int \rmd H_{\fl(i)}\, \zeta^{(b^+)}_{k_i'}\,\zeta^{(b^-)*}_{k_i}=\sum_{k_i}\l_b^2\tau^{(b^+\cup b^-)}_{k_i},\qquad \tau^{(b^+\cup b^-)}_{k_i}=d_{k_i}\tr_{(k_i,\rho_i)}\lt[\overrightarrow{\prod_{v\in\partial (b^+\cup b^-)}}P_{k_i}g_{ve}^{-1}g_{ve'}P_{k_i}\rt]
\ee
for each $i=1,\cdots,p_b$. See FIG.\ref{bpbm} for an illustration. The inner product is non-trivial only when pairing the terms with the same $p_b$ from $\psi_\mu$ and $\psi_\mu^*$. Two copies of $b$ makes $b^+\cup b^-$ as an internal face $h$ in $\ck_n$. $\tau^{(b^+\cup b^-)}_{k_i}$ is understood in the same way as $\tau^{(h)}_{k_i}$ in \eqref{zetakh}. The corresponding contribution from $h=b^+\cup b^-$ is given by 
\be
\o_h\lt(A_h;g_{ve}^{I,\pm},\l_h\rt),\qquad A_h=A_{b},\quad \l_h=\l_b^{2},
\ee
where $b=b^\pm\in\ck$ and $h=b^+\cup b^-\in\ck_n$.

\item Given any boundary face $b\in\ck$ intersecting with $\Fs$: $b\cap \Fs\neq \emptyset$, the link $\fl_*=b\cap \Sig$ is divided by $\Fs$ into half-links $\fl_*^R$ and $\fl_*^{\bR}$, and we write 
\be
\zeta_{k_i}^{(b)}(H^R_{\fl_*(i)},H^{\bR}_{\fl_*(i)})=\l _b d_{k_i}\tr_{({k_i},\rho_i)}\lt[\lt(\overrightarrow{\prod_{v\in\partial b}}P_{k_i}g_{ve}^{-1}g_{ve'}P_{k_i}\rt) H^R_{\fl_*(i)}H^{\bR}_{\fl_*(i)}\rt]
\ee
where $H^R_{\fl_*(i)},H^{\bR}_{\fl_*(i)}$ are the SU(2) holonomies along the half-links $\fl_*^R,\fl_*^{\bR}$ respectively. As a part of the data $\mu$, $p_b$ and $k_1,\cdots,k_{p_b}$ are fixed for $\psi_\mu$. We label this boundary face in each of the $2n$ copies of $\psi,\psi^*$ by $b^{\pm}_I$, $I=1,\cdots,n$. Gluing these $2n$ faces gives\footnote{$\int \rmd H_1\,\tr_{k'}(\bar{H}_1 A H_1)\tr_k(H_1^{-1}B \overline{H}^{-1}_2 )=d_k^{-1}\delta^{kk'}\tr_k(\overline{H}_1AB\overline{H}_2^{-1} )$ again ``glues'' the traces by canceling $H_1$ and $H_1^{-1}$, then as the next step, $\int \rmd \overline{H}_2 \tr_k(\overline{H}_1AB\overline{H}_2^{-1} )\tr_{k'}(\overline{H}_2 CD\overline{H}_3^{-1} )= d_k^{-1}\delta^{kk'} \tr_k(\overline{H}_1ABCD\overline{H}_3^{-1} )$. Iteratively using these relations gives \eqref{bbbbbbbb} as illustrated by FIG.\ref{replica_face}. Note that there are only $2n-1$ non-trivial integrals in \eqref{bbbbbbbb}: $\overline{H}_1$ cancels after integrating $H_1\cdots H_n$ and $\overline{H}_2\cdots\overline{H}_n$. For example, $\int \rmd h\,\tr_{k'}(\overline{H}_1A H)\tr_k(H^{-1}B \overline{H}^{-1}_1 )=d_k^{-1}\delta^{kk'}\tr_k(AB )$.}
\be
&&\sum_{\{k_I^\pm\}}\int \prod_{I=1}^n\rmd H_I\rmd \overline{H}_I\, \zeta_{k_1^+}^{(b_1^+)}(H_1,\overline{H}_1)\,\zeta_{k_1^-}^{(b_1^-)}(H_1,\overline{H}_2)^*\zeta_{k_2^+}^{(b_2^+)}(H_2,\overline{H}_2)\,\zeta_{k_2^-}^{(b_2^-)}(H_2,\overline{H}_3)^*\cdots \zeta_{k_n^+}^{(b_n^+)}(H_n,\overline{H}_n)\,\zeta_{k_n^-}^{(b_n^-)}(H_n,\overline{H}_1)^*\nonumber\\
&=&\sum_k \l_b^{2n}\tau_k^{( b_1^+\cup b_1^- \cup b_2^+\cup b_2^-\cdots b_n^+\cup b_n^-)}\equiv \sum_k\l_h^{2n}\tau_k^{(h)},\label{bbbbbbbb}
\ee
where $H$ and $\overline{H}$ stands for the half-link holonomies $H^R$ and $H^{\bR}$. The glued face $h=b_1^+\cup b_1^- \cup b_2^+\cup b_2^-\cdots b_n^+\cup b_n^-$ is an internal face in the root complex $\ck_n$, and $\t^{( b_1^+\cup b_1^- \cup b_2^+\cup b_2^-\cdots b_n^+\cup b_n^-)}_{k}$ is understood in the same way as $\t^{(h)}_{k}$ in \eqref{zetakh} for the internal face $h=b_1^+\cup b_1^- \cup b_2^+\cup b_2^-\cdots b_n^+\cup b_n^-$ in $\ck_n$. We have identified $\l_h=\l_b$ where $h=b_1^+\cup b_1^- \cdots b_n^+\cup b_n^-\in\ck_n$ and $b=b_I^\pm\in\ck$. See FIG.\ref{replica_face} for an illustration of \eqref{bbbbbbbb}.

\end{enumerate}

\begin{figure}[t]
\centering
\includegraphics[width=0.4\textwidth]{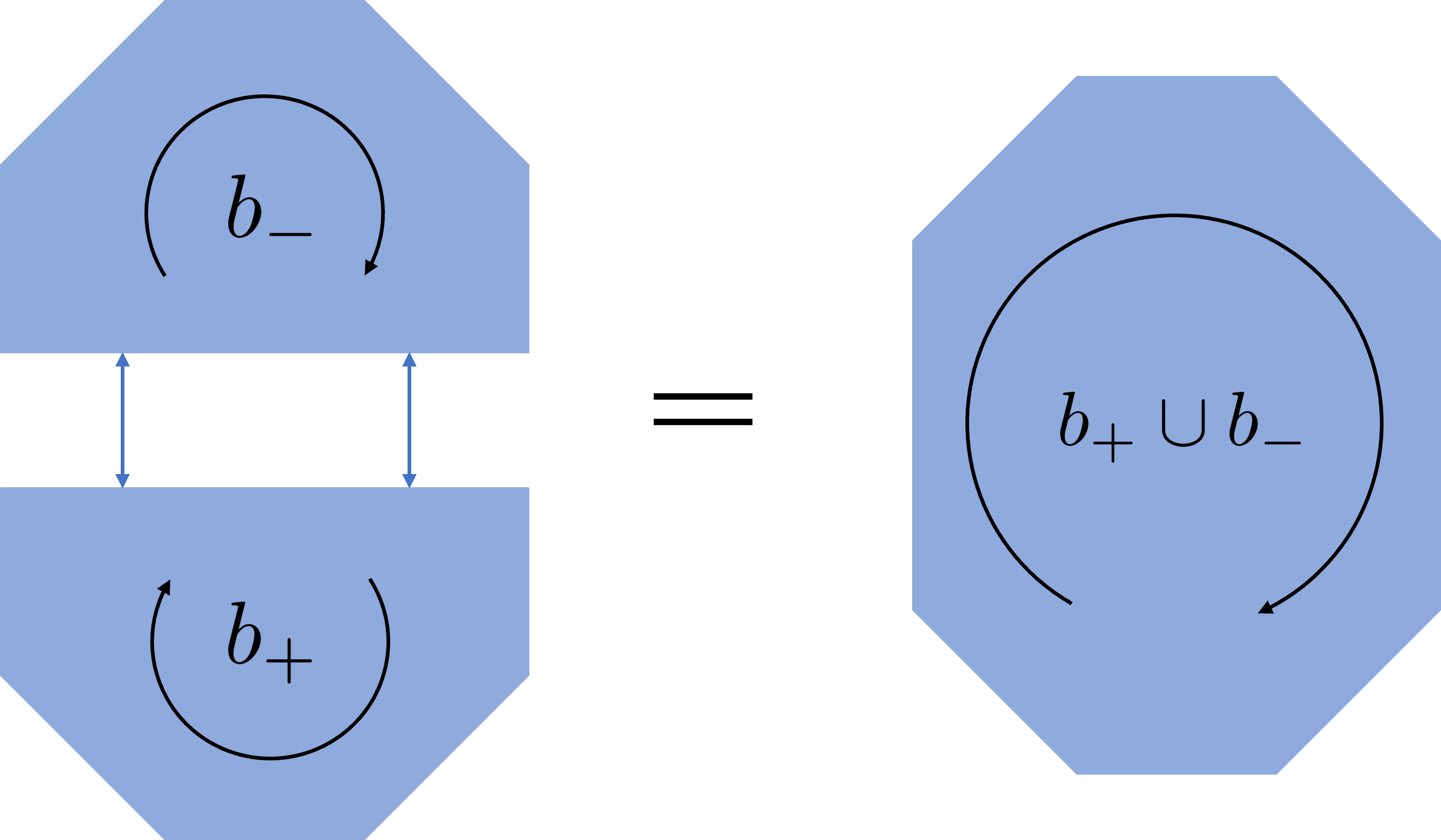}
\caption{The gluing inside $R$ or $\bR$.}
\label{bpbm}
\end{figure}

\begin{figure}[t]
\centering
\includegraphics[width=1\textwidth]{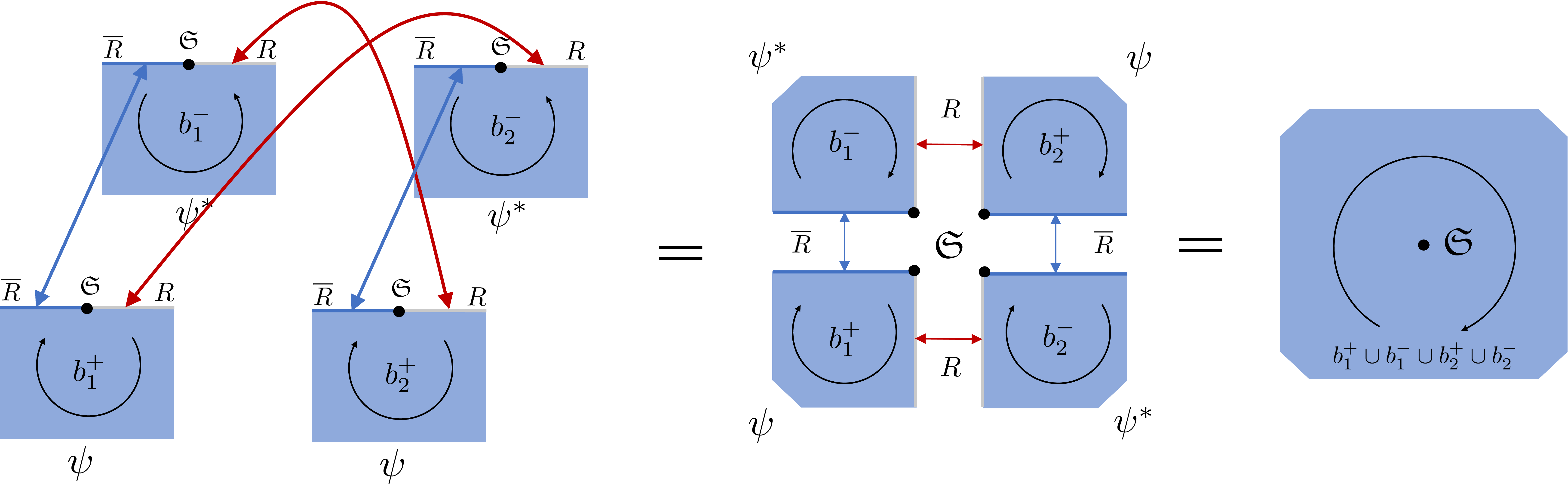}
\caption{Gluing faces at $\Fs$ for $n=2$. The intersection $b_I^{\pm}\cap\Fs$ is the branch point (the black point labelled by $\Fs$) of the resulting face. The holonomies $H_I$ and $\overline{H}_I$ corresponds respectively to the gray and blue edges, where the gluing happens.}
\label{replica_face}
\end{figure}

The replica partition function $\cz_{\mu,n}$ is the stack amplitude based on the root complex $\ck_n$ with fixed $\mu=(\calp,\cj)$ on $\Fs$:
\be
\cz_{\mu,n}=\int\lt[\rmd g\rt]_n\prod_{f\cap\Fs= \emptyset}\o_f\prod_{h\cap \Fs\neq \emptyset}\l_h^{2n p_h}\prod_{i=1}^{p_h}\t^{(h)}_{k_{i}}\Theta\left(A_{h}-\alpha_{p_h,\vec{k}}\right).
\ee
where $f$ and $h$ in the integrand labels the faces in $\ck_n$. The product over the faces $f$ with $f\cap\Fs= \emptyset$ contains the ingredients in the above items 2 and 3, while the faces $h$ with $h\cap\Fs\neq \emptyset$ corresponds to the item 4. The partition function $\cz_{\mu,n}$ is real and positive by construction.

The full replica partition function $\cz_n$ is given by $\cz_n=\sum_\mu \kappa_\mu^{1-n}\cz_{\mu,n}$:
\be
\cz_{n}&=&\int\lt[\rmd g\rt]_n \prod_{f\cap\Fs= \emptyset}\o_f\prod_{h\cap\Fs\neq \emptyset}\o_h^{(n)}\label{Zleqna000}\\
\o_h^{(n)}&=&\sum_{p_h \in \Z_+} \l_h^{2n p_h}\prod_{i=1}^{p_h}\sum_{k_i\in\Z_+}d_{k_i}^{1-n}\tau^{(h)}_{k_{i}}\Theta\left(A_{h}-\alpha_{p_h,\vec{k}}\right).
\ee

\subsection{Computing the replica partition function}

The same approximation as in Section \ref{Internal faces} can be applied to $\o_h^{(n)}$: For some sufficiently large $\re(s)>0$, the following sum converges for any $n\geq 1$
\be
&&\sum_{p_h \in \Z_+} \sum_{k_{1},\cdots, k_{p_h}\in\Z_+} \lt|\l_h^{2n p_h}\prod_{i=1}^{p_h}d_{k_i}^{1-n}\tau^{(h)}_{k_{i}}\rt|e^{-\re(s)\alpha_{p_h,\vec{k}}}
\leq \sum_{p_h \in \Z_+} \sum_{k_{1},\cdots, k_{p_h}\in\Z_+} \l_h^{2n p_h}\prod_{i=1}^{p_h}d_{k_i}^{3-n}e^{-\re(s)\alpha_{p_h,\vec{k}}}\nonumber\\
&=&\sum_{p_h=1}^\infty \lt(\l_h^{2n}\sum_{k=1}^\infty d_k^{3-n}e^{-\re(s)\sqrt{k(k+2)}} \rt)^{p_h}<\infty.
\ee
Then the sum in $\o^{(n)}_{h}$ can be computed by using \eqref{invLaplace},
\be
\o^{(n)}_{h}&=&\frac{1}{2\pi i}\int\limits_{T-i\infty}^{T+i\infty}\frac{\rmd s_h}{s_h}e^{A_h s_h}\lt[\sum_{p_h=1}^\infty \lt(\l_h^{2n}\sum_{k=1}^\infty d_k^{1-n}\t_k^{(h)}e^{-s_h\sqrt{k(k+2)}} \rt)^{p_h}\rt]\nonumber\\
&=&\frac{1}{2\pi i}\int\limits_{T-i\infty}^{T+i\infty}\frac{\rmd s_h}{s_h}e^{A_h s_h}\frac{\l_h^{2n}\sum_{k=1}^\infty d_k^{1-n}\t_k^{(h)}e^{-s_h\sqrt{k(k+2)}}}{1-\l_h^{2n}\sum_{k=1}^\infty d_k^{1-n}\t_k^{(h)}e^{-s_h\sqrt{k(k+2)}}}\ .\label{ratio590}
\ee 
We denote by $g_h=\{g_{ve}\}_{e\subset h}$ the set of group variables along $\partial h$. The pole $s_h^{(n)}(g_h)$ with the maximal possible $\re(s_h)$ in the $s$-plane is given by
\be
\l_h^{2n}\sum_{k=1}^\infty d_k^{1-n}\t_k^{(h)}(g_h)e^{-s_h^{(n)}(g_h)\sqrt{k(k+2)}}=1.\label{polegh}
\ee
By the bound $|\t_k^{(h)}|\leq d_k^2$, the maximum of $\re(s_h^{(n)}(g_h))$ on the space of $g_h$ is given by $\b_h(n)>0$ satisfying
\be
\l^{2n}_h\sum_{k=1}^\infty d_k^{3-n} e^{-\b_h(n)\sqrt{k(k+2)}}=1.\label{spole2n0}
\ee
The maximum is reached when $g_h$ along the boundaries of $h=b_1^+\cup b_1^- \cup b_2^+\cup b_2^-\cdots b_n^+\cup b_n^-$ satisfy
\be
\text{all}\  g_{ve}^{-1}g_{ve'}\in \Su,\quad e,e'\subset\partial h ,\qquad \overrightarrow{\prod_{v\in\partial h}} g_{ve}^{-1}g_{ve'} =1\ .\label{g0solution110}
\ee 
For large $A_h$, we obtain the approximation similar to \eqref{omegah31}
\be
&&\o_h^{(n)}= e^{A_h s_h^{(n)}(g_h)}\cf^{(n)}_h(g_h) +\Fr_h^{(n)} ,\qquad \cf^{(n)}_h(g_h)=\frac{1}{s_h^{(n)}(g_h)\lag \sqrt{k(k+2)} \rag_{h,g}^{(n)}},\label{omegah80}\\
&&\lag \sqrt{k(k+2)} \rag_{h,g}^{(n)}=\l_n^{2n}\sum_{k=1}^\infty \sqrt{k(k+2)} d^{1-n}_k \t_k^{(h)} (g_h) e^{- s_h^{(n)}(g_h)\sqrt{k(k+2)}},
\ee
where $\Fr_h^{(n)}$ is the subleading contribution from the poles with $\re(s_h)$ restrictly less than $\b_h(n)$.

Insert \eqref{omegah80} and the approximation \eqref{omegah31} for $\o_h$ with $h\cap\Fs = \emptyset$, we obtain 
\be
\cz_n&\simeq& e^{\sum\limits_{h\cap\Fs\neq\emptyset} \beta_h(n) A_h+\sum\limits_{h\cap\Fs=\emptyset} \beta_h A_h} \int [\rmd g]_n\, e^{I_n+I_{\rm int}+I_{\rm glue}} {\prod_{h\cap\Fs\neq\emptyset} \cf^{(n)}_h}{\prod_{h\cap\Fs=\emptyset} \cf_h }\prod_b\o_b.\label{cznapprox}
\ee
In this expression, $h,b$ labels faces in $\ck_n$. In particular, $b$ labels the boundary faces that corresponds to $b_0$ for each copy of $\ck$. The terms in $I_n+I_{\rm int}+I_{\rm glue}$ are based on the classification in Section \ref{Replica trick}: $I_n$ collects the faces in $\ck_n$ with $h\cap\Fs\neq\emptyset$, i.e. $h=b_1^+\cup b_1^- \cdots b_n^+\cup b_n^-$. $I_{\rm glue}$ contains the faces made by the gluing $h= b_+\cup b_-$. $I_{\rm int}$ contains all the internal faces in $2n$ copies of $\ck$.
\be
&& I_n=\sum_{h\cap\Fs\neq\emptyset} A_h \big[s^{(n)}_h(g_h)-\b_h(n)\big],\qquad I_{\rm int}=\sum_{\substack{h\cap\Fs=\emptyset }}A_h \big[s_h(g_h)-\b_h\big],\qquad I_{\rm glue}=\sum_{\substack{h\cap\Fs=\emptyset \\ h= b_+\cup b_-}}A_h \big[s_h(g_h)-\b_h\big]. 
\ee
Based on the classification in Section \ref{Replica trick}, we re-express the first exponent in \eqref{cznapprox} in terms of the faces $h,b$ in $\ck$ and find the second term is proportional to $n$, 
\be
\sum_{h\cap\Fs\neq\emptyset} \beta_h(n) A_h=\sum_{\substack{b\subset\ck \\ b\cap\Fs\neq\emptyset} }\beta_b(n) A_b,\qquad 
\sum_{h\cap\Fs=\emptyset} \beta_h A_h= n\lt( 2\sum_{h\subset\ck} \beta_h A_h+\sum_{\substack{b\subset\ck \\ b\cap\Fs=\emptyset}} \beta_b A_b\rt).
\ee

Recall that there are two asymptotic limits for the spinfoam states $\psi$: removing cut-offs of internal spins: $A_h\to\infty$ and the classical limit $\hbar\to 0$ which relates to $A_b\to\infty$. Here $h,b$ are faces in $\ck$ where $\psi$ is defined. In the following, we take the limit $A_h\to\infty$. On the other hand, instead of the classical limit, we keep $A_b$ large but finite for $b$ intersecting $\Sig$ and compute the scaling behavior of the entanglement entropy for large $A_b$. The initial state $\psi_0$ and $A_{b_0}$ for $b_0$ intersection $\Sig_0$ are arbitrary. In particular, we do not consider the scaling of $A_{b_0}$. We will see that the leading order behavior of the entanglement entropy does not depend on the initial state $\psi_0$.

Following the discussion in Section \ref{Stationary phase approximation}, when we send the cut-off $A_h\to\infty$ uniformly, the integrals over $g_{ve}^{I,\pm}$ for $e$ not connecting to $\partial\ck$ localizes on $2n$ copies of $\cc_{\rm int}$. The asymptotics of $\cz_n$ gives
\be
\cz_n = e^{\sum_{ b\cap\Fs\neq\emptyset} \beta_b(n) A_b+n \Fl}\int [\rmd g_{b}]_n\,f_n,\qquad \Fl =\sum\limits_{b\cap\Fs=\emptyset} \beta_b A_b+2\sum\limits_{h} \beta_h A_h-\mathscr{D}_{\rm int}\log \bA_h
\ee
All $h$ and $b$ in this formula labels the faces in a copy of $\ck$. The integral $\int[\rmd g_{b}]_n=\int\prod_{I=1}^{n}\prod_{\pm}\prod_{(v,e_{b})} \rmd g_{ve_{b}}^{I, \pm}$ is over all group variables associated to the edges $e_b\subset\ck$ connecting to the boundary $\Sig$ where the gluing happens. It is clear that $\Fl$ is $n$-independent. $\bA_h$ is the mean value of $\{A_h\}_{h\subset\ck}$. The integrand $f_n$ does not scale with the cut-offs $A_h$
\be
f_n
&\simeq& \int\limits_{\cc_{\rm int}^{\times 2n}}\rmd^{2n} \mu(u) \, e^{I_n(u)+I_{\rm glue}(u)}\prod_{\substack{h\cap\Fs\neq \emptyset\\h=b_1^+\cup b_1^- \cdots b_n^+\cup b_n^-}}\cf_h^{(n)}(u)\prod_{\substack{h\cap\Fs=\emptyset \\ h= b_+\cup b_-}}\cf_h(u)\nonumber\\
&&\qquad \prod_{I,\pm}\int \prod_{(v,e_{b_0})} \rmd g_{ve_{b_0}}^{I, \pm}\, {\prod_{b_0^\pm}\o_{b_0^\pm}(\vec u_I^{\pm})}\lt[\phi_I^\pm (u_I^{\pm})+O(\bA^{-1})\rt] .
\ee
where $u\equiv \{\vec{u}_I^{\pm}\}_{I=1}^n$ are the coordinates on the critical manifold $\cc_{\rm int}^{\times 2n}$. $\cf_h(u)$ is the function $\cf_h(g_h)$ evaluated on the critical manifold, and similar for other $u$-dependent functions. The edges $e_{b_0}\subset\ck$ connect to the boundary $\Sig_0$ where the initial state $\psi_0$ is defined.

The action $I_n+I_{\rm glue}$ is linear to $A_b\gg1 $. The stationary phase approximation localizes the remaining integral on the solutions to \eqref{g0solution110} for all $h=b_1^+\cup b_1^- \cdots b_n^+\cup b_n^-$ and $h= b_+\cup b_-$. At this point, the integral of $\cz_n$ localizes on the space of solutions to \eqref{g0solution110} for all internal faces $h\subset\ck_n$. We denote this solution space by $\cc^{(n)}_{\rm int}$ \footnote{$\cc^{(n)}_{\rm int}$ relates to the moduli space of SU(2) flat connections on the replicated manifold $\sm_4^{(n)}$ (the manifold obtained by gluing $\sm_4$ as in FIG.\ref{replica}).}. As a result,
\be
\cz_n = e^{\sum_{ b\cap\Fs\neq\emptyset} \beta_b(n) A_b+n \Fl}\bA_b^{-\mathscr{D}_{\rm glue}(n)}\iota_n
\ee
where $\bA_b$ is the mean value of $\{A_b\}_{b\subset \ck\cap\Sig}$. $\mathscr{D}_{\rm glue}(n)\propto n$ equals to half of the dimension of the integral. $\iota_n$ does not scale with both $A_b$ and $A_h$:
\be
\iota_n=\int\limits_{\cc^{(n)}_{\rm int}}\rmd \mu(u)\prod_{I,\pm}\int \prod_{(v,e_{b_0})} \rmd g_{ve_{b_0}}^{I, \pm}\, {\prod_{b_0^\pm}\o_{b_0^\pm}(\vec u_I^{\pm})}\lt[\varphi_n (u)+O(\bA^{-1})\rt].
\ee
We still use $u$ to denote the coordinates on $\cc^{(n)}_{\rm int}$. Note that $\cf_h$ and $\cf_h^{(n)}$ are constant on $\cc^{(n)}_{\rm int}$.

%Let us divide the process into two steps:

% \begin{enumerate}

% \item We first integrate out $g_{ve_{b}}^{I, \pm}$ for all $e_b$ that do not involves in $\partial h$ with $h\cap\Fs\neq\emptyset$. These integrals is localized onto the solution to \eqref{g0solution110} for all $h= b_+\cup b_-$. Combining localization of the internal $g_{ve}$, the integral of $\cz_n$ reduces to the integral over the solution space $\widetilde\cc_{\rm int}^{(n)}$ of \eqref{g0solution110} at all $h\subset\ck_n$ except for $h\cap\Fs\neq\emptyset$. By the same argument as proving Lemma \ref{SU2flatconn}, $\widetilde\cc_{\rm int}^{(n)}$ is identical to the moduli space of SU(2) flat connections on $\sm_4^{(n)}\setminus N(\Fs)$, where $\sm_4^{(n)}$ is the replicated manifold obtained by gluing $\sm_4$ as in FIG.\ref{replica}. All internal $g_{ve}^{I, \pm}$ in $\ck_n$ are determined as the holonomies in the moduli space, since every internal $g_{ve}^{I, \pm}$ is along $\partial h$ for either $h$ internal in $\ck$ or $h= b_+\cup b_-$. These clearly include the holonomies involved in $h\cap\Fs\neq\emptyset$.

% \item If we integrate all $g_{ve_{b}}^{I, \pm}$, $\cz_n$ localizes on the space of solutions to \eqref{g0solution110} for all internal faces $h\subset\ck_n$, including $h\cap\Fs\neq\emptyset$. We denote this solution space by $\cc^{(n)}_{\rm int}$, which is identical to the moduli space of SU(2) flat connections on the replicated manifold $\sm_4^{(n)}$

% \end{enumerate}

From now on, we assume for simplicity that at all boundary faces $b \subset \ck$ intersecting $\Sig$, the coupling constants $\l_b=\l$ are constant in the definition of $\psi$. This implies $\b_b(n)=\b(n)$ to be an $n$-dependent constant. Then $\log(\cz_n)$ relates to the total area $a$ of $\Fs$:
\be
\log(\cz_n)= \b(n) a+n\Fl-\mathscr{D}_{\rm glue}(n) \log(a)+O(1),\qquad a=\sum_{b\cap \Fs\neq \emptyset} A_b.\label{logZn0}
\ee
The von Neumann entanglement entropy given by the limit of the modular entropy satisfies the area law at the leading order:
\be
S(\bm\rho,R)=\lim_{n\to 1}\widetilde{S}_n=-\lim_{n\to 1}n^{2}\partial_{n}\left[\frac{1}{n}\log(\cz_n)\right]=\b a+O(1).\label{arealawSrhoR0}
%\lim_{n\to 1}\widetilde{S}_n=-\lim_{n\to 1}n^{2}\partial_{n}\left[\frac{1}{n}\log(\cz_n)\right]=\b a-\mathscr{D}\log(a)+O(1)
\ee
It is remarkable that $S(\bm\rho,R)$ does not depend on the cut-offs $A_h$ of internal spins, since the term in $\log(\cz_n)$ proportional to $n$ does not contribute to $\widetilde{S}_n$. The coefficient $\b$ of the area law is given by
\be
\b=\b(1)-\partial_n\b(1)=\beta(1)+\frac{\lambda^{-2}\log\left(\lambda^{-2}\right)+\sum_{k=1}^{\infty}d_{k}^{2}\log d_{k}e^{-\beta(1)\sqrt{k(k+2)}}}{\sum_{k=1}^{\infty}d_{k}^{2}e^{-\beta(1)\sqrt{k(k+2)}}\sqrt{k(k+2)}}.
\ee
where $\partial_n\b(1)$ is obtained by the derivative of \eqref{spole2n0}. The value of $\b$ depending on $\l$ can be computed numerically: For $\l=1$, $\b\simeq 1.595742$. FIG.\ref{beta0plot} shows the behavior of $\b$ as $\l$ changes. 

%Similarly, $\mathscr{D}=\mathscr{D}_{\rm glue}(1)-\partial_n\mathscr{D}_{\rm glue}(1)$, and $\mathscr{D}=0$ in the case of nondegnerate Hessian. 

\begin{figure}[t]
\centering
\includegraphics[width=0.6\textwidth]{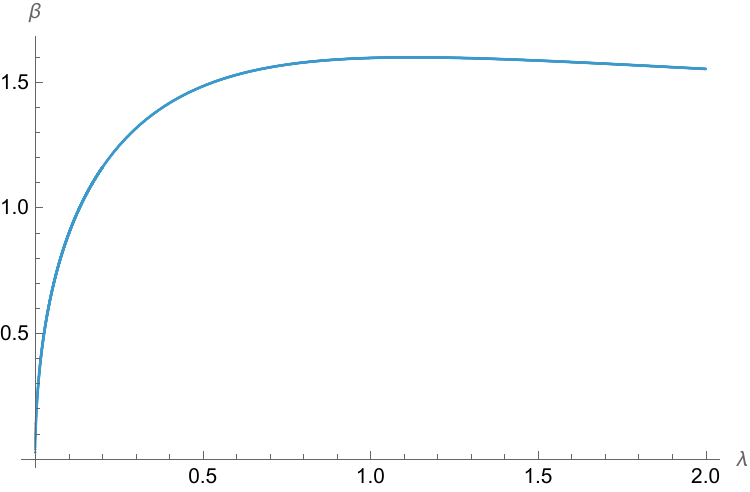}
\caption{This figure plots $\b$ as a function of $\l$.}
\label{beta0plot}
\end{figure}

% Then $\psi$ relates to $\psi_\leq (a)$ by
% \be
% \psi=\psi_\leq \lt(A_\Fs+\delta\rt)-\psi_\leq \lt(A_\Fs-\delta\rt).
% \ee

The entropy formula in Eq. \eqref{arealawSrhoR0} has several interesting features: First, $\b$ does not depend on the choice of root complex $\ck$ in defining spinfoam state $\psi$, so the entropy formula \eqref{arealawSrhoR0} is a triangulation-independent prediction from the spinfoam LQG.

The dimensionless area $a$ relates to the dimensionful area of the surface $\Fs$ by $\Ar(\Fs)=4\pi \g\ell_P^2 a$. By setting the parameter $\b=\pi\g$, we recover the standard Bekenstein-Hawking area law for gravitational entropy:
\be
S(\bm\rho,R)\simeq \frac{\Ar(\Fs)}{4\ell_P^2}.
\ee
The condition $\b=\pi\g$ fixes the coupling constant $\l$ as a function of the Barbero-Immirzi parameter $\g$ for $0<\g\lesssim 1/2$, as shown in FIG. \ref{lambdaPlot}. $\l(\g)$ is small in the regime of small $\g$. Consequently, the sum over graphs in the state $\psi$ becomes a perturbative expansion in the small parameter $\l(\g)$, where the leading order term is given by the root graph.

Motivated by the fact that any spatial slice of a spinfoam stack yields a spin-network stack, we generalize this $\g$-dependence to all coupling constants $\l_f$. This implies that the stack amplitude $\sa_\ck$ is also a perturbative expansion in $\l(\g)$ for small $\g$. The leading order of this expansion is the spinfoam amplitude on the root complex $\ck$.

The existing results on the semiclassical analysis of spinfoam amplitude indicates that the classical gravity should relate to the regime of small $\g$, and these results are only based on the root complex (see e.g.\cite{Han:2023cen,claudio,claudio1,propagator3,propagator2,propagator1}). These results still hold for the full stack amplitude, as contributions from higher-order graphs in the stack are naturally suppressed by powers of the small coupling $\l(\g)$.

\begin{figure}[t]
\centering
\includegraphics[width=0.6\textwidth]{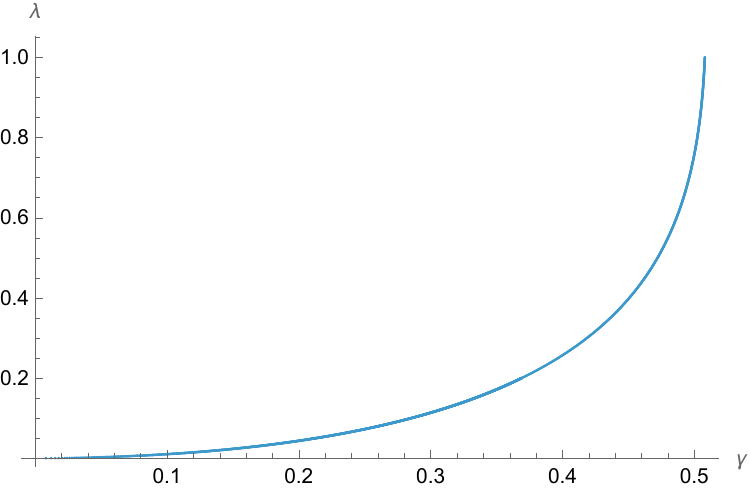}
\caption{This figure plots $\l$ as a function of $\g$ in the small $\g$ regime. The relation between $\l$ and $\g$ is fixed by $S(\bm\rho,R)\simeq {\Ar(\Fs)}/{4\ell_P^2}.$}
\label{lambdaPlot}
\end{figure}

The above derivation does not taking into account the scaling of $A_{b_0}$ for $b_0$ connecting to the initial slice $\Sig_0$. The leading order area law and the logarithmic term in \eqref{arealawSrhoR0} does not depend on the initial state $\psi_0$. In particular, one obtain the same result even if $\psi_0$ is a non-stack spin-network state. Indeed, for each $b_0$, $\partial_{\l_{b_0}}\cz_n|_{\l_{b_0}\to0}$ removes the faces stacked on $b_0$ but does not affect the result \eqref{logZn0} of $\log(\cz_n)$ except for the $O(1)$ term.

An advantage of our spinfoam-based computation is its avoidance of the sign ambiguity present in conventional Lorentzian path integral derivations of gravitational entropy. Traditional methods typically rely on an analytic continuation of the gravitational action $S_{\rm EH}$ (see e.g., \cite{dche,Colin-Ellerin:2020mva,Dittrich:2024awu,Banihashemi:2024weu}). This procedure generates the entropy from an imaginary part of the action, but it also introduces a fundamental sign ambiguity ($\pm$) due to the singular geometry at the entangling surface $\Fs$ in the replicated manifold. To obtain the correct positive sign, one must make an ad-hoc choice of the path integral contour to circumvent the singularity.

Our approach is free of this problem. As a genuinely Lorentzian path integral, the spinfoam formulation requires no analytic continuation and its integration cycle is non-singular. This ensures the resulting entropy is positive definite and free of $\pm$ ambiguity (this does not rely on the sign of $\l$ because only $\l^2$ enters the computation, see e.g. \eqref{ratio590}). This distinction also points to a more direct physical interpretation: the entropy in our framework, given by Eq. \eqref{arealawSrhoR0}, is not an artifact of a complexified action but a direct consequence of summing over distinct microstates $\mu$ on the surface $\Fs$. This state-sum, $\cz_n=\sum_\mu\kappa_\mu^{1-n}\cz_{\mu,n}$, arises because the von Neumann algebra $\cm$ of the region is not a factor and has a center given by the area operators at the punctures on $\Fs$.

Finally, we emphasize that the entropy formula \eqref{arealawSrhoR0} is valid for all bipartitions of the spatial slice $\Sig$, as far as the interface intersects the graph at the links.

\section{Spinfoam state with fix areas}\label{Spinfoam state with fix areas}

The macroscopic areas $A_b$ are the upper bound of the coarse-grained areas given by the spinfoam state $\psi$. Now we consider an alternative spinfoam state $\psi'$ which fixes a set of coarse-grained areas as in \eqref{stackcoarsegr2}. This state approximates the common eigenstate of the set of area operators $\{\mathbf{A}_\fl\}_{\fl\subset\G}$, each of which associates to the surface $\Fs_\fl=R_i\cap R_j$ transverse to the link $\fl=b\cap\Sig$ and the stacked links at $\fl$. Similar to \eqref{stackcoarsegr2}, $\psi'$ is obtained from $\psi$ by the replacement 
\be
\Theta\left(A_{b}-\alpha_{p_b,\vec{k}}\right)\to \chi_{\delta}\left(A_{b}-\alpha_{p_b,\vec{k}}\right)
\ee
for all boundary faces $b\subset\ck$ intersecting $\Sig$: 
\be
\psi'=\int \prod_{(v,e)}\rmd g_{ve}\, \prod_{f = h,b_0}\o_f \prod_b\o_b', \qquad \o_b'=\sum_{p_b=1}^{\infty}\prod_{i=1}^{p_b}\sum_{k_{i}=1}^{\infty}\zeta^{(b)}_{k_{i}}\chi_\delta\left(A_{b}-\alpha_{p_b,\vec{k}}\right) .
\ee
The state $\psi'$ sums the common eigenstates of $\{\mathbf{A}_\fl\}_{\fl\subset\G}$ with eigenvalues $\alpha_{p_b,\vec{k}}$ satisfying
\be
A_{b}-\delta \leq \alpha_{p_b,\vec{k}}\leq A_{b}+\delta.
\ee
Here $A_{b}\gg1$ is a macroscopic area and $\delta\sim O(1)$.

For any biparation $\Sig=R\cup \bR$ with the interface $\Fs$. The spin profile on $\Fs$ is again denoted by $\mu$. The state $\psi'$ can be written as a restricted sum of $\psi_\mu$ defined in \eqref{psimustate00}: $\psi'=\widetilde{\sum}_\mu\psi_\mu$, where the restricted sum $\widetilde{\sum}_\mu$ respects $\chi_\delta(\cdots)$ instead of $\Theta(\cdots)$. Consequently, when we define $\bm\rho'=|\psi'\rangle\langle\psi'|$ and apply the replica trick, the resulting replica partition function is given by $\cz_n'=\widetilde{\sum}_\mu\cz_\mu$. By the relation
\be
\chi_\delta\left(A_{b}-\alpha_{p_b,\vec{k}}\right)=\Theta\left(A_{b}+\delta-\alpha_{p_b,\vec{k}}\right)-\Theta\left(A_{b}-\delta-\alpha_{p_b,\vec{k}}\right),
\ee
and the expression of $\cz_n$ in terms of $\Theta(\cdots) $, the new replica partition function $\cz_n'$ can be obtained by a sequence of summing and subtracting $\cz_n$'s
\be
&&\cz_n'(A_1,A_2,\cdots,A_m)=\sum_{s_1,\cdots,s_m=\pm1}s_1\cdots s_m\cz_n(A_1+s_1\delta,A_2+s_2\delta,\cdots,A_m+s_m\delta).
%&=&e^{n\Fl'}\prod_{b\cap\Fs\neq \emptyset} \lt(e^{\b(n)(A_b+\delta)}-e^{\b(n)(A_b-\delta)}\rt)\prod_{b\cap\Fs= \emptyset} \lt(e^{n\b_b(A_b+\delta)}-e^{n\b_b(A_b-\delta)}\rt)\bA_b^{-\mathscr{D}_{\rm glue}}(1+O(\delta/\bA_b))
\ee
The right hand side only modifies the $O(1)$ term in $\log(\cz_n)$. The main contribution to $\log(\cz'_n)$ is the same as $\log(\cz_n)$
\be
\log(\cz'_n)= \b(n) a+n\Fl-\mathscr{D}_{\rm glue}(n) \log(a)+O(1),\qquad a=\sum_{b\cap \Fs\neq \emptyset} A_b.
\ee
The entanglement entropy of $\psi'$ is also given by the same expression as in \eqref{arealawSrhoR0}
\be
S(\bm\rho',R)=-\lim_{n\to 1}n^{2}\partial_{n}\left[\frac{1}{n}\log(\cz_n')\right]=\b a+O(1).
\ee

\section{An analog of black hole entropy}\label{An analog of black hole entropy}

The spinfoam state discussed in the last section fixes a set of areas $A_b$, and the entanglement entropy is computed for aribitrary bipartition. In this section, we consider an alternative spinfoam state $\psi_\Fs$ and a preferred biparatition $\Sig=R\cup\bR$. The state $\psi_\Fs$ fixes (approximately) the total area $a$ of the entire interface $\Fs=R\cap\bR$, instead of fixing the individual areas $A_b$ whose sum gives $a$. This scenario may be seen as an analog of a black hole, where the horizon corresponding to $\Fs$ divides the spatial slice into the interior $\bR$ and exterior $R$, although here we still consider $\Sig,\Sig_0$ to be compact to ignore any boundary effect. The entanglement entropy of $\psi_\Fs$ with respect to the preferred bipartition is expected to relate to the black hole entropy.

The state $\psi_\Fs$ is approximately an eigenstate of the area operator associated to $\Fs$. We restrict the state to be a linear combination of the spin-network states satisfying
\be
a-\delta\leq \sum_{b\cap \Fs\neq \emptyset}\a_{p_b,\vec{k}} \leq a+\delta
\ee
where $a\gg1$ and $\delta\sim O(1)$. The macroscopic area of $\Fs$ and $4\pi \g \ell_P a$. $\psi_\Fs$ is obtained from $\psi$ by replacing all the $\Theta(\cdots)$'s of the boundary faces $b$ satisfying $b\cap\Fs\neq \emptyset$ with $\chi_\delta(\cdots)$ restricting the area of the entire surface $\Fs$:
\be
\prod_{b\cap \Fs\neq \emptyset}\Theta\lt(A_b-\a_{p_b,\vec{k}}\rt)\ \to\  \chi_\delta\lt(a-\sum_{b\cap \Fs\neq \emptyset}\a_{p_b,\vec{k}}\rt),
\ee
while all other $\Theta(\cdots)$'s are left unchanged. The resulting state $\psi_\Fs$, has the following expression
\be
\psi_\Fs&=&\int \prod_{(v,e)}\rmd g_{ve}\, \O_\Fs \prod_{f \cap \Fs= \emptyset}\o_f , \label{psistate}\\
\O_\Fs&=&\sum_{\{p_b \in \Z_+\}_{b\cap \Fs\neq \emptyset}} \sum_{\{k_{1},\cdots, k_{p_b}\in\Z_+\}_{b\cap \Fs\neq \emptyset}} \prod_{b\cap \Fs\neq \emptyset}\prod_{i=1}^{p_b}\zeta^{(b)}_{k_{i}}\chi_\delta\left(a-\sum_{b\cap \Fs\neq \emptyset}\alpha_{p_b,\vec{k}}\right),\label{OmegaFs}
\ee
The sum in \eqref{OmegaFs} is the same as $\sum_\mu$, which is now constrained by a single $\chi_\delta(\cdots)$ for the total area of $\Fs$. In \eqref{psistate} and \eqref{OmegaFs}, $v$, $e$, $f$ and $b$ are the quantities in the root complex $\ck$.

We define $\bm\rho_\Fs=|\psi_\Fs\rangle\langle \psi_\Fs|$ and perform the replica trick according to the bipartition. Following the same procedure as before, we obtain the replica partition function
\be
\cz_{\Fs,n}=\cz_{\leq,n}(a+\delta)-\cz_{\leq,n}(a-\delta)\ .\label{substractFs}
\ee
where
\be
\cz_{\leq,n}(a)&=&\int\lt[\rmd g\rt]_n\, \O^{(n)}_{\Fs,\leq}(a)\prod_{f\cap\Fs= \emptyset}\o_f\ ,\label{Zleqna00}\\
\O^{(n)}_{\Fs,\leq}(a)&=&\sum_{\{p_h \in \Z_+\}_{h\cap \Fs\neq \emptyset}} \sum_{\{k_{1},\cdots, k_{p_h}\in\Z_+\}_{h\cap \Fs\neq \emptyset}} \,\prod_{h\cap \Fs\neq \emptyset}\lt[\l_h^{2n p_h}\prod_{i=1}^{p_h}d_{k_i}^{1-n}\tau^{(h)}_{k_{i}}\rt]\Theta\left(a-\sum_{h\cap \Fs\neq \emptyset}\alpha_{p_h,\vec{k}}\right).
\ee
In the expressions of $\cz_{\leq,n}(a)$ and $\O^{(n)}_{\Fs,\leq}(a)$, $f,h$ are the faces in the root complex $\ck_n$ after the replica gluing.

The sum in $\O^{(n)}_{\Fs,\leq}(a)$ can be computed by using \eqref{invLaplace},
\be
\O^{(n)}_{\Fs,\leq}(a)&=&\frac{1}{2\pi i}\int\limits_{T-i\infty}^{T+i\infty}\frac{\rmd s}{s}e^{as}\prod_{h\cap \Fs\neq\emptyset}\lt[\sum_{p_h=1}^\infty \lt(\l^{2n}\sum_{k=1}^\infty d_k^{1-n}\t_k^{(h)}e^{-s\sqrt{k(k+2)}} \rt)^{p_h}\rt]\nonumber\\
&=&\frac{1}{2\pi i}\int\limits_{T-i\infty}^{T+i\infty}\frac{\rmd s}{s}e^{as}\prod_{h\cap \Fs\neq\emptyset}\frac{\l^{2n}\sum_{k=1}^\infty d_k^{1-n}\t_k^{(h)}e^{-s\sqrt{k(k+2)}}}{1-\l^{2n}\sum_{k=1}^\infty d_k^{1-n}\t_k^{(h)}e^{-s\sqrt{k(k+2)}}}\ .\label{ratio59}
\ee 
In this formula, the product of the ratios over different faces $h$ is under a single $s$-integral, in contrast to the situation at \eqref{ratio590}, where each $h$ associates to an $s_h$-integral. We denote by $q$ the number of links intersecting $\Fs$ in $\Sig$. $q$ equals to the number of ratios in the integrand.

In the integrand of \eqref{ratio59}, each $h$ contributes a pole $s_h^{(n)}(g_h)$ in the $s$-plane, which satisfying the same equation as \eqref{polegh} with constant $\l_h=\l$. However, in contrast to the situation at \eqref{polegh}, now all the poles from different $h$ are on the same $s$-plane. All these poles reach the maximal $\re(s_h^{(n)}(g_h))$ when \eqref{g0solution110} are satisfied for all $h$, and they become coincide: $s_h^{(n)}(g_h)=\b(n)$ for all $h\cap\Fs\neq\emptyset$, in other words, Eq.\eqref{g0solution110} satisfied for all $h$ gives an order-$q$ pole to the integral.

For large $a$, we compute $\O^{(n)}_{\Fs,\leq}(a)$ by only taking into account the poles making dominant contributions: 
\be
&&\O^{(n)}_{\Fs,\leq}(a)= \sum_{h\cap \Fs\neq\emptyset} e^{a\,s_h^{(n)}(g_h)}\cg^{(n)}_h(\vec g)+\Fr_\Fs \\
&& \cg^{(n)}_h(\vec g)=\frac{1}{s_h^{(n)}(g_h)\lag \sqrt{k(k+2)}\rag^{(n)}_{h,g}}\,\prod_{\substack{h'\neq h \\ h'\cap \Fs\neq\emptyset}}\frac{F^{(n)}_{h'}\lt(g_{h'},s_h^{(n)}(g_h)\rt)}{1-F^{(n)}_{h'}\lt(g_{h'},s_h^{(n)}(g_h)\rt)}
%\frac{\l^{2n}\sum_{k=1}^\infty d_k^{1-n}\t_k^{({h'})}e^{-s_h(g)\sqrt{k(k+2)}}}{1-\l^{2n}\sum_{k=1}^\infty d_k^{1-n}\t_k^{({h'})}e^{-s_h(g)\sqrt{k(k+2)}}}
\ee
where $\vec g\equiv \{g_{h}\}_{h\cap \Fs\neq\emptyset}$. We have used the short-hand notation
\be
F^{(n)}_{h'}\lt(g_{h'},s^{(n)}_h(g_h)\rt)&=&\l^{2n}\sum_{k=1}^\infty d_k^{1-n}\t_k^{(h')}(g_{h'})e^{-s^{(n)}_h(g_h)\sqrt{k(k+2)}}
% \\
% \lag \sqrt{k(k+2)}\rag^{(n)}_{h,g}&=&\l^{2n}\sum_{k=1}^\infty \sqrt{k(k+2)} d_k^{1-n}\t_k(g_h) e^{- s(g_h)\sqrt{k(k+2)}}.
\ee
This formula is valid for generic $g$, so $s(g_h)\neq s(g_{h'})$ for $h\neq h'$. On $\cc_{\rm int}^{(n)}$ where \eqref{g0solution110} are satisfied for all $h\cap\Fs\neq\emptyset$, $\O^{(n)}_{\Fs,\leq}(a)$ gives the fastest possible growth as $a\to\infty$ (on the space of $g$):
\be
&&\O^{(n)}_{\Fs,\leq}(a) =  a^{q-1}e^{\b(n)a}\cj^{(n)}_\Fs,\qquad \cj^{(n)}_\Fs= \frac{(q-1)!}{\b(n)\lt[\lag \sqrt{k(k+2)}\rag_{\b(n)}\rt]^q}+O\lt(a^{-1}\rt) ,\\
&&\qquad \qquad \lag \sqrt{k(k+2)}\rag_{\b(n)}=\l^{2n}\sum_{k=1}^\infty d_k^{3-n}\sqrt{k(k+2)} e^{- \b(n)\sqrt{k(k+2)}}.
\ee
Then $\cz_{\leq,n}(a)$ is expected to behave as
\be
\cz_{\leq,n}(a)=a^{q-1}e^{\b(n)a} \int  [\rmd g]_n\cdots
\ee
where the integral suppresses in power-law in $a$ as $a\to\infty$.

Indeed, this behavior is demonstrated by the following argument: We define the coordinates $t$ transverse to $\cc_{\rm int}^{(n)}$ in the space of $\{g_{ve}^{I,\pm}\}$, so the integral of $[\rmd g]_n$ becomes the integrals of $u$ and $t$ in a neighborhood, where $u$ are coordinates on $\cc_{\rm int}^{(n)}$, and $\cc_{\rm int}^{(n)}$ is at $t=0$. We expand the following functions in power series in ${t}$ \footnote{The expansion may also be obtain by expanding $F_h^{(n)}(g_h,s)$ in \eqref{ratio59}: $F_h^{(n)}(g_h(t),s)=1+\delta s \partial_s F_h^{(n)}(g_h(0),\beta(n))+\lt[c^{(n)}_{h}\rt]_{\a\b}t^\a t^\b$, where $\delta s=s-\b(n)$ and $\partial_s F_h^{(n)}(g_h(0),\beta(n))=-\langle\sqrt{k(k+2)}\rangle_{\b(n)}$, then performing the contour integral of $\delta s$.}
\be
a\lt[s_h^{(n)}-\b(n)\rt]+\sum_{h'\cap \Fs=\emptyset} A_{h'} \lt[s_{h'}-\b_{h'}\rt]&=&a\lt(\frac{1}{2} \lt[H^{(n)}_h\rt]_{\a\b}t^\a t^\b+O(t^3)\rt),\label{expSh}\\
1-F^{(n)}_{h'}\lt(g_{h'},s_h^{(n)}(g_h)\rt)&=&\lt[c^{(n)}_{h'}\rt]_{\a\b}t^\a t^\b-\lt[c^{(n)}_{h}\rt]_{\a\b}t^\a t^\b+O(t^3),\label{1-Fh}
\ee
where $H_h^{(n)}$ is the Hessian matrix. $c^{(n)}_{h}$ is given by
\be
\lt[c^{(n)}_{h}\rt]_{\a\b}=\frac{1}{2}\frac{\partial}{\partial t^\a}\frac{\partial}{\partial t^\b} F_{h}^{(n)}\lt(g_{h}(t),\beta(n)\rt)\Big|_{t=0},
%\qquad \lt[b^{(n)}_{h}\rt]_{\a\b}=\frac{\partial}{\partial s}F_{h'}^{(n)}\lt(g_{h'}(0),s\rt)\Big|_{s=\b(n)}\frac{\partial}{\partial t^\a}\frac{\partial}{\partial t^\b}s^{(n)}_h\lt(g_h(t)\rt)\Big|_{t=0}
\ee
We apply the scaling of coordinates $t^\a\to a^{-1/2} t^\a$. The scaling removes the overall $a$ in \eqref{expSh} and scales \eqref{1-Fh} with $a^{-1}$. The measure $\rmd^{\Delta(n)} t$ is scaled with $a^{-\Delta(n)/2}$. As a result, $\cz_{\leq,n}(a)$ has the following behavior
\be
\cz_{\leq,n}(a)&=& a^{q-1}a^{-\Delta(n)/2}e^{\b(n)a}e^{\sum_{h\cap \Fs= \emptyset}A_h \b_h}\int\rmd\mu(u)\int \rmd^{\Delta(n)} t\\
&&\sum_{h\cap \Fs\neq \emptyset}\frac{e^{\frac{1}{2} \lt[H^{(n)}_h(u)\rt]_{\a\b}t^\a t^\b+O(a^{-1/2})}}{\prod\limits_{\substack{h'\neq h\\ h'\cap\Fs\neq\emptyset}}\lt(\lt[c^{(n)}_{h'}\rt]_{\a\b}t^\a t^\b-\lt[c^{(n)}_{h}\rt]_{\a\b}t^\a t^\b+O\lt(a^{-1/2}\rt)\rt)}\lt[\Psi_n(u)+O\lt(a^{-1/2}\rt)\rt].
\ee
where $\Psi_n(u)$ is proportional to the Jacobian of measure at $\cc_{\rm int}^{(n)}$. The integrand at $t=0$ is finite and proportional to $\cj^{(n)}_{\Fs}$. As $a\to\infty$, the $t$-integral on a compact neighborhood at $t=0$ is finite and does not scale with $a$. As a result, we obtain
\be
\log(\cz_{\leq,n}(a))=\b(n)a+\lt[q-1-\Delta(n)/2\rt]\log(a)+n\Fl'+O(1)
\ee
where $n\Fl'=\sum_{h\cap \Fs= \emptyset}A_h \b_h$ is proportional to $n$. By \eqref{substractFs}, $\log(\cz_{\Fs,n}(a))$ is given by the same expression as $\log(\cz_{\leq,n}(a))$, where the correction is only at $O(1)$.

The von Neumann entanglement entropy is given by 
\be
S(\bm\rho_\Fs,R)=-\lim_{n\to 1}n^{2}\partial_{n}\left[\frac{1}{n}\log(\cz_{\Fs,n})\right]=\b a+(q-1-\Delta/2)\log(a)+O(1).\label{arealawSrhoR1}
\ee
where $\Delta=\Delta(1)-\partial_n\Delta(1)$. The leading order area law reproduces the black hole entropy $\Ar(\Fs)/(4\ell_P^2)$ when the coupling constant $\l=\l(\g)$ is given by FIG.\ref{lambdaPlot}, by the same discussion as below \eqref{arealawSrhoR0}.
In addition, the logarithmic correction in \eqref{arealawSrhoR1} contains a positive contribution $q$ relating to the discretization of $\Fs$, due to the order-$q$ pole. This dependence of the discreteness is due to the definition of the state $\psi_\Fs$ is based on the fixed boundary root graph $\G$.

\begin{acknowledgements}

The author receives supports from the National Science Foundation through grants PHY-2207763 and PHY-2512890 and from a sponsorship of renewed research stay in Erlangen from the Humboldt Foundation. 

\end{acknowledgements}

\appendix

\section{Solutions of $\sum_{k=1}^\infty (\pm 1)^k d_k^{2}e^{- s_\pm \sqrt{k(k+2)}} = \lambda^{-1}$}\label{Solutions of sum=lambda equations}

\noindent
\emph{Uniqueness of the real solution $s_+=\beta$}:

Let us first analyze the equation $\sum_{k=1}^\infty d_k^{2}e^{- s_+ \sqrt{k(k+2)}} = \lambda^{-1}$. We define the short-hand notation: $E_k = \sqrt{k(k+2)}$ and $g(s) = \sum_{k=1}^\infty d_k^{2}e^{- s E_k}$. The equation for $s_+$ is $g(s_+) = \lambda^{-1}$. Let us examine the properties of $g(s)$ for $s > 0$: $g(s)$ is a strictly decreasing function as $s$ grows, because the derivative $g'(s) = - \sum_{k=1}^\infty d_k^{2} E_k e^{- s E_k}<0$. A strictly decreasing function can only take on any specific value once. There must be a unique positive real value $s_+\equiv \b$ that satisfies $g(\beta) = \lambda^{-1}$.\\

\noindent 
\emph{The bound of other complex solutions $s_+$}:

Now we show that any other complex solution $s_+ = x+iy$ (where $x,y\in\R$ and $y \neq 0$) must satisfy $x < \beta$ strictly for the equation $g(s_+) = \lambda^{-1}$. The equation implies that $\im (g(s_+))=0$ and for the real part
\be
\sum_{k=1}^\infty d_k^{2}e^{-x E_k} \cos(yE_k) = \lambda^{-1}.
\ee
By $\cos(\theta) \le 1$, we obtain
\be
g(\beta)=\lambda^{-1} \leq \sum_{k=1}^\infty d_k^{2}e^{-x E_k} = g(x).
\ee
Since $g(s)$ is a strictly decreasing function, this inequality implies $x \le \beta$.

The equality $x = \beta$ can only happen if $\cos(yE_k) = 1$ for all $k$. This would require $yE_k$ to be a multiple of $2\pi$ for every $k$. This is impossible for any $y \neq 0$ because the ratios of different $E_k$ values (e.g., $E_1/E_2 = \sqrt{3/8}$) are generally irrational. Thus, for any non-real solution, the inequality must be strict, which means $\text{Re}(s_+) = x < \beta$.\\

\noindent
\emph{The equation $\sum_{k=1}^\infty (-1)^k d_k^2 e^{-s_- E_k} = \lambda^{-1}$ and the bound of the solutions $s_-$}:

Let $s_- = x+iy$ be a solution. The real part of the equation is 
\be
\sum_{k=1}^\infty d_k^{2}e^{-x E_k} (-1)^k\cos(yE_k) = \lambda^{-1},
\ee
which again implies $g(\beta)\leq g(x)$ and thus $ x \leq \beta$. The equality $x = \beta$ can only happen if $(-1)^k\cos(yE_k) = 1$ for all $k$, i.e. $yE_k=\pi k$ mod $2\pi$. This is again impossible because the ratios of different $E_k$ values are generally irrational. Therefore, the inequality must again be strictly $\text{Re}(s_-) = x < \beta$.

\bibliography{muxin.bib}

\begin{thebibliography}{46}%
\makeatletter
\providecommand \@ifxundefined [1]{%
 \@ifx{#1\undefined}
}%
\providecommand \@ifnum [1]{%
 \ifnum #1\expandafter \@firstoftwo
 \else \expandafter \@secondoftwo
 \fi
}%
\providecommand \@ifx [1]{%
 \ifx #1\expandafter \@firstoftwo
 \else \expandafter \@secondoftwo
 \fi
}%
\providecommand \natexlab [1]{#1}%
\providecommand \enquote  [1]{``#1''}%
\providecommand \bibnamefont  [1]{#1}%
\providecommand \bibfnamefont [1]{#1}%
\providecommand \citenamefont [1]{#1}%
\providecommand \href@noop [0]{\@secondoftwo}%
\providecommand \href [0]{\begingroup \@sanitize@url \@href}%
\providecommand \@href[1]{\@@startlink{#1}\@@href}%
\providecommand \@@href[1]{\endgroup#1\@@endlink}%
\providecommand \@sanitize@url [0]{\catcode `\\12\catcode `\$12\catcode `\&12\catcode `\#12\catcode `\^12\catcode `\_12\catcode `\%12\relax}%
\providecommand \@@startlink[1]{}%
\providecommand \@@endlink[0]{}%
\providecommand \url  [0]{\begingroup\@sanitize@url \@url }%
\providecommand \@url [1]{\endgroup\@href {#1}{\urlprefix }}%
\providecommand \urlprefix  [0]{URL }%
\providecommand \Eprint [0]{\href }%
\providecommand \doibase [0]{https://doi.org/}%
\providecommand \selectlanguage [0]{\@gobble}%
\providecommand \bibinfo  [0]{\@secondoftwo}%
\providecommand \bibfield  [0]{\@secondoftwo}%
\providecommand \translation [1]{[#1]}%
\providecommand \BibitemOpen [0]{}%
\providecommand \bibitemStop [0]{}%
\providecommand \bibitemNoStop [0]{.\EOS\space}%
\providecommand \EOS [0]{\spacefactor3000\relax}%
\providecommand \BibitemShut  [1]{\csname bibitem#1\endcsname}%
\let\auto@bib@innerbib\@empty
%</preamble>
\bibitem [{\citenamefont {Ryu}\ and\ \citenamefont {Takayanagi}(2006)}]{Ryu:2006bv}%
  \BibitemOpen
  \bibfield  {author} {\bibinfo {author} {\bibfnamefont {S.}~\bibnamefont {Ryu}}\ and\ \bibinfo {author} {\bibfnamefont {T.}~\bibnamefont {Takayanagi}},\ }\bibfield  {title} {\bibinfo {title} {{Holographic derivation of entanglement entropy from AdS/CFT}},\ }\href {https://doi.org/10.1103/PhysRevLett.96.181602} {\bibfield  {journal} {\bibinfo  {journal} {Phys. Rev. Lett.}\ }\textbf {\bibinfo {volume} {96}},\ \bibinfo {pages} {181602} (\bibinfo {year} {2006})},\ \Eprint {https://arxiv.org/abs/hep-th/0603001} {arXiv:hep-th/0603001 [hep-th]} \BibitemShut {NoStop}%
%%CITATION = HEP-TH/0603001;%%
\bibitem [{\citenamefont {Jacobson}(2016)}]{Jacobson:2015hqa}%
  \BibitemOpen
  \bibfield  {author} {\bibinfo {author} {\bibfnamefont {T.}~\bibnamefont {Jacobson}},\ }\bibfield  {title} {\bibinfo {title} {{Entanglement Equilibrium and the Einstein Equation}},\ }\href {https://doi.org/10.1103/PhysRevLett.116.201101} {\bibfield  {journal} {\bibinfo  {journal} {Phys. Rev. Lett.}\ }\textbf {\bibinfo {volume} {116}},\ \bibinfo {pages} {201101} (\bibinfo {year} {2016})},\ \Eprint {https://arxiv.org/abs/1505.04753} {arXiv:1505.04753 [gr-qc]} \BibitemShut {NoStop}%
%%CITATION = ARXIV:1505.04753;%%
\bibitem [{\citenamefont {Thiemann}(2007)}]{thiemann2008modern}%
  \BibitemOpen
  \bibfield  {author} {\bibinfo {author} {\bibfnamefont {T.}~\bibnamefont {Thiemann}},\ }\href@noop {} {\emph {\bibinfo {title} {Modern Canonical Quantum General Relativity}}}\ (\bibinfo  {publisher} {Cambridge University Press},\ \bibinfo {year} {2007})\BibitemShut {NoStop}%
\bibitem [{\citenamefont {Rovelli}\ and\ \citenamefont {Vidotto}(2014)}]{rovelli2014covariant}%
  \BibitemOpen
  \bibfield  {author} {\bibinfo {author} {\bibfnamefont {C.}~\bibnamefont {Rovelli}}\ and\ \bibinfo {author} {\bibfnamefont {F.}~\bibnamefont {Vidotto}},\ }\href@noop {} {\emph {\bibinfo {title} {Covariant Loop Quantum Gravity: An Elementary Introduction to Quantum Gravity and Spinfoam Theory}}},\ Cambridge Monographs on Mathematical Physics\ (\bibinfo  {publisher} {Cambridge University Press},\ \bibinfo {year} {2014})\BibitemShut {NoStop}%
\bibitem [{\citenamefont {Rovelli}\ and\ \citenamefont {Smolin}(1995)}]{rovelli1995discreteness}%
  \BibitemOpen
  \bibfield  {author} {\bibinfo {author} {\bibfnamefont {C.}~\bibnamefont {Rovelli}}\ and\ \bibinfo {author} {\bibfnamefont {L.}~\bibnamefont {Smolin}},\ }\bibfield  {title} {\bibinfo {title} {{Discreteness of area and volume in quantum gravity}},\ }\href {https://doi.org/10.1016/0550-3213(95)00150-Q} {\bibfield  {journal} {\bibinfo  {journal} {Nuclear Physics B}\ }\textbf {\bibinfo {volume} {442}},\ \bibinfo {pages} {593} (\bibinfo {year} {1995})}\BibitemShut {NoStop}%
\bibitem [{\citenamefont {Ashtekar}\ and\ \citenamefont {Lewandowski}(1997)}]{ashtekar1997quantum}%
  \BibitemOpen
  \bibfield  {author} {\bibinfo {author} {\bibfnamefont {A.}~\bibnamefont {Ashtekar}}\ and\ \bibinfo {author} {\bibfnamefont {J.}~\bibnamefont {Lewandowski}},\ }\bibfield  {title} {\bibinfo {title} {{Quantum theory of geometry. 1: Area operators}},\ }\href {https://doi.org/10.1088/0264-9381/14/1A/006} {\bibfield  {journal} {\bibinfo  {journal} {Class.Quant.Grav.}\ }\textbf {\bibinfo {volume} {14}},\ \bibinfo {pages} {A55} (\bibinfo {year} {1997})},\ \Eprint {https://arxiv.org/abs/gr-qc/9602046} {arXiv:gr-qc/9602046 [gr-qc]} \BibitemShut {NoStop}%
%%CITATION = GR-QC/9602046;%%
\bibitem [{\citenamefont {Ashtekar}\ and\ \citenamefont {Lewandowski}(1998)}]{ashtekar1997quantumII}%
  \BibitemOpen
  \bibfield  {author} {\bibinfo {author} {\bibfnamefont {A.}~\bibnamefont {Ashtekar}}\ and\ \bibinfo {author} {\bibfnamefont {J.}~\bibnamefont {Lewandowski}},\ }\bibfield  {title} {\bibinfo {title} {{Quantum theory of geometry. 2. Volume operators}},\ }\href@noop {} {\bibfield  {journal} {\bibinfo  {journal} {Adv.Theor.Math.Phys.}\ }\textbf {\bibinfo {volume} {1}},\ \bibinfo {pages} {388} (\bibinfo {year} {1998})},\ \Eprint {https://arxiv.org/abs/gr-qc/9711031} {arXiv:gr-qc/9711031 [gr-qc]} \BibitemShut {NoStop}%
%%CITATION = GR-QC/9711031;%%
\bibitem [{\citenamefont {Perez}(2013)}]{Perez2012}%
  \BibitemOpen
  \bibfield  {author} {\bibinfo {author} {\bibfnamefont {A.}~\bibnamefont {Perez}},\ }\bibfield  {title} {\bibinfo {title} {{The Spin Foam Approach to Quantum Gravity}},\ }\href {https://doi.org/10.12942/lrr-2013-3} {\bibfield  {journal} {\bibinfo  {journal} {Living Rev.Rel.}\ }\textbf {\bibinfo {volume} {16}},\ \bibinfo {pages} {3} (\bibinfo {year} {2013})},\ \Eprint {https://arxiv.org/abs/1205.2019} {arXiv:1205.2019 [gr-qc]} \BibitemShut {NoStop}%
%%CITATION = ARXIV:1205.2019;%%
\bibitem [{\citenamefont {Han}\ and\ \citenamefont {Krajewski}(2014)}]{hanPI}%
  \BibitemOpen
  \bibfield  {author} {\bibinfo {author} {\bibfnamefont {M.}~\bibnamefont {Han}}\ and\ \bibinfo {author} {\bibfnamefont {T.}~\bibnamefont {Krajewski}},\ }\bibfield  {title} {\bibinfo {title} {{Path integral representation of Lorentzian spinfoam model, asymptotics, and simplicial geometries}},\ }\href {https://doi.org/10.1088/0264-9381/31/1/015009} {\bibfield  {journal} {\bibinfo  {journal} {Class.Quant.Grav.}\ }\textbf {\bibinfo {volume} {31}},\ \bibinfo {pages} {015009} (\bibinfo {year} {2014})},\ \Eprint {https://arxiv.org/abs/1304.5626} {arXiv:1304.5626 [gr-qc]} \BibitemShut {NoStop}%
%%CITATION = ARXIV:1304.5626;%%
\bibitem [{\citenamefont {Conrady}\ and\ \citenamefont {Freidel}(2008)}]{Conrady:2008ea}%
  \BibitemOpen
  \bibfield  {author} {\bibinfo {author} {\bibfnamefont {F.}~\bibnamefont {Conrady}}\ and\ \bibinfo {author} {\bibfnamefont {L.}~\bibnamefont {Freidel}},\ }\bibfield  {title} {\bibinfo {title} {{Path integral representation of spin foam models of 4d gravity}},\ }\href {https://doi.org/10.1088/0264-9381/25/24/245010} {\bibfield  {journal} {\bibinfo  {journal} {Class. Quant. Grav.}\ }\textbf {\bibinfo {volume} {25}},\ \bibinfo {pages} {245010} (\bibinfo {year} {2008})},\ \Eprint {https://arxiv.org/abs/0806.4640} {arXiv:0806.4640 [gr-qc]} \BibitemShut {NoStop}%
\bibitem [{\citenamefont {Rovelli}(1996)}]{Rovelli:1996dv}%
  \BibitemOpen
  \bibfield  {author} {\bibinfo {author} {\bibfnamefont {C.}~\bibnamefont {Rovelli}},\ }\bibfield  {title} {\bibinfo {title} {{Black hole entropy from loop quantum gravity}},\ }\href {https://doi.org/10.1103/PhysRevLett.77.3288} {\bibfield  {journal} {\bibinfo  {journal} {Phys. Rev. Lett.}\ }\textbf {\bibinfo {volume} {77}},\ \bibinfo {pages} {3288} (\bibinfo {year} {1996})},\ \Eprint {https://arxiv.org/abs/gr-qc/9603063} {arXiv:gr-qc/9603063} \BibitemShut {NoStop}%
\bibitem [{\citenamefont {Ashtekar}\ \emph {et~al.}(1998)\citenamefont {Ashtekar}, \citenamefont {Baez}, \citenamefont {Corichi},\ and\ \citenamefont {Krasnov}}]{Ashtekar:1997yu}%
  \BibitemOpen
  \bibfield  {author} {\bibinfo {author} {\bibfnamefont {A.}~\bibnamefont {Ashtekar}}, \bibinfo {author} {\bibfnamefont {J.}~\bibnamefont {Baez}}, \bibinfo {author} {\bibfnamefont {A.}~\bibnamefont {Corichi}},\ and\ \bibinfo {author} {\bibfnamefont {K.}~\bibnamefont {Krasnov}},\ }\bibfield  {title} {\bibinfo {title} {{Quantum geometry and black hole entropy}},\ }\href {https://doi.org/10.1103/PhysRevLett.80.904} {\bibfield  {journal} {\bibinfo  {journal} {Phys. Rev. Lett.}\ }\textbf {\bibinfo {volume} {80}},\ \bibinfo {pages} {904} (\bibinfo {year} {1998})},\ \Eprint {https://arxiv.org/abs/gr-qc/9710007} {arXiv:gr-qc/9710007} \BibitemShut {NoStop}%
\bibitem [{\citenamefont {Domagala}\ and\ \citenamefont {Lewandowski}(2004)}]{Domagala:2004jt}%
  \BibitemOpen
  \bibfield  {author} {\bibinfo {author} {\bibfnamefont {M.}~\bibnamefont {Domagala}}\ and\ \bibinfo {author} {\bibfnamefont {J.}~\bibnamefont {Lewandowski}},\ }\bibfield  {title} {\bibinfo {title} {{Black hole entropy from quantum geometry}},\ }\href {https://doi.org/10.1088/0264-9381/21/22/014} {\bibfield  {journal} {\bibinfo  {journal} {Class. Quant. Grav.}\ }\textbf {\bibinfo {volume} {21}},\ \bibinfo {pages} {5233} (\bibinfo {year} {2004})},\ \Eprint {https://arxiv.org/abs/gr-qc/0407051} {arXiv:gr-qc/0407051} \BibitemShut {NoStop}%
\bibitem [{\citenamefont {Agullo}\ \emph {et~al.}(2010)\citenamefont {Agullo}, \citenamefont {Fernando~Barbero}, \citenamefont {Borja}, \citenamefont {Diaz-Polo},\ and\ \citenamefont {Villasenor}}]{Agullo:2010zz}%
  \BibitemOpen
  \bibfield  {author} {\bibinfo {author} {\bibfnamefont {I.}~\bibnamefont {Agullo}}, \bibinfo {author} {\bibfnamefont {J.}~\bibnamefont {Fernando~Barbero}}, \bibinfo {author} {\bibfnamefont {E.~F.}\ \bibnamefont {Borja}}, \bibinfo {author} {\bibfnamefont {J.}~\bibnamefont {Diaz-Polo}},\ and\ \bibinfo {author} {\bibfnamefont {E.~J.~S.}\ \bibnamefont {Villasenor}},\ }\bibfield  {title} {\bibinfo {title} {{Detailed black hole state counting in loop quantum gravity}},\ }\href {https://doi.org/10.1103/PhysRevD.82.084029} {\bibfield  {journal} {\bibinfo  {journal} {Phys. Rev. D}\ }\textbf {\bibinfo {volume} {82}},\ \bibinfo {pages} {084029} (\bibinfo {year} {2010})},\ \Eprint {https://arxiv.org/abs/1101.3660} {arXiv:1101.3660 [gr-qc]} \BibitemShut {NoStop}%
\bibitem [{\citenamefont {Engle}\ \emph {et~al.}(2010)\citenamefont {Engle}, \citenamefont {Perez},\ and\ \citenamefont {Noui}}]{ENP}%
  \BibitemOpen
  \bibfield  {author} {\bibinfo {author} {\bibfnamefont {J.}~\bibnamefont {Engle}}, \bibinfo {author} {\bibfnamefont {A.}~\bibnamefont {Perez}},\ and\ \bibinfo {author} {\bibfnamefont {K.}~\bibnamefont {Noui}},\ }\bibfield  {title} {\bibinfo {title} {{Black hole entropy and SU(2) Chern-Simons theory}},\ }\href {https://doi.org/10.1103/PhysRevLett.105.031302} {\bibfield  {journal} {\bibinfo  {journal} {Phys.Rev.Lett.}\ }\textbf {\bibinfo {volume} {105}},\ \bibinfo {pages} {031302} (\bibinfo {year} {2010})},\ \Eprint {https://arxiv.org/abs/0905.3168} {arXiv:0905.3168 [gr-qc]} \BibitemShut {NoStop}%
%%CITATION = ARXIV:0905.3168;%%
\bibitem [{\citenamefont {Barbero~G.}\ and\ \citenamefont {Perez}(2017)}]{BarberoG:2015xcq}%
  \BibitemOpen
  \bibfield  {author} {\bibinfo {author} {\bibfnamefont {J.~F.}\ \bibnamefont {Barbero~G.}}\ and\ \bibinfo {author} {\bibfnamefont {A.}~\bibnamefont {Perez}},\ }\bibinfo {title} {{Quantum geometry and black holes.}},\ in\ \href {https://doi.org/10.1142/9789813220003_0008} {\emph {\bibinfo {booktitle} {{Loop Quantum Gravity}: {The First 30 Years}}}},\ \bibinfo {editor} {edited by\ \bibinfo {editor} {\bibfnamefont {A.}~\bibnamefont {Ashtekar}}\ and\ \bibinfo {editor} {\bibfnamefont {J.}~\bibnamefont {Pullin}}}\ (\bibinfo  {publisher} {WSP},\ \bibinfo {year} {2017})\ pp.\ \bibinfo {pages} {241--279},\ \Eprint {https://arxiv.org/abs/1501.02963} {arXiv:1501.02963 [gr-qc]} \BibitemShut {NoStop}%
\bibitem [{\citenamefont {Han}(2025{\natexlab{a}})}]{lqgee1}%
  \BibitemOpen
  \bibfield  {author} {\bibinfo {author} {\bibfnamefont {M.}~\bibnamefont {Han}},\ }\bibfield  {title} {\bibinfo {title} {{Entanglement entropy in Loop Quantum Gravity and geometrical area law}},\ }\href@noop {} {\  (\bibinfo {year} {2025}{\natexlab{a}})}\BibitemShut {NoStop}%
\bibitem [{\citenamefont {Bianchi}\ \emph {et~al.}(2019)\citenamefont {Bianchi}, \citenamefont {Dona},\ and\ \citenamefont {Vilensky}}]{Bianchi:2018fmq}%
  \BibitemOpen
  \bibfield  {author} {\bibinfo {author} {\bibfnamefont {E.}~\bibnamefont {Bianchi}}, \bibinfo {author} {\bibfnamefont {P.}~\bibnamefont {Dona}},\ and\ \bibinfo {author} {\bibfnamefont {I.}~\bibnamefont {Vilensky}},\ }\bibfield  {title} {\bibinfo {title} {{Entanglement entropy of Bell-network states in loop quantum gravity: Analytical and numerical results}},\ }\href {https://doi.org/10.1103/PhysRevD.99.086013} {\bibfield  {journal} {\bibinfo  {journal} {Phys. Rev.}\ }\textbf {\bibinfo {volume} {D99}},\ \bibinfo {pages} {086013} (\bibinfo {year} {2019})},\ \Eprint {https://arxiv.org/abs/1812.10996} {arXiv:1812.10996 [gr-qc]} \BibitemShut {NoStop}%
%%CITATION = ARXIV:1812.10996;%%
\bibitem [{\citenamefont {Pastawski}\ \emph {et~al.}(2015)\citenamefont {Pastawski}, \citenamefont {Yoshida}, \citenamefont {Harlow},\ and\ \citenamefont {Preskill}}]{Pastawski:2015qua}%
  \BibitemOpen
  \bibfield  {author} {\bibinfo {author} {\bibfnamefont {F.}~\bibnamefont {Pastawski}}, \bibinfo {author} {\bibfnamefont {B.}~\bibnamefont {Yoshida}}, \bibinfo {author} {\bibfnamefont {D.}~\bibnamefont {Harlow}},\ and\ \bibinfo {author} {\bibfnamefont {J.}~\bibnamefont {Preskill}},\ }\bibfield  {title} {\bibinfo {title} {{Holographic quantum error-correcting codes: Toy models for the bulk/boundary correspondence}},\ }\href {https://doi.org/10.1007/JHEP06(2015)149} {\bibfield  {journal} {\bibinfo  {journal} {JHEP}\ }\textbf {\bibinfo {volume} {06}},\ \bibinfo {pages} {149}},\ \Eprint {https://arxiv.org/abs/1503.06237} {arXiv:1503.06237 [hep-th]} \BibitemShut {NoStop}%
%%CITATION = ARXIV:1503.06237;%%
\bibitem [{\citenamefont {Hayden}\ \emph {et~al.}(2016)\citenamefont {Hayden}, \citenamefont {Nezami}, \citenamefont {Qi}, \citenamefont {Thomas}, \citenamefont {Walter},\ and\ \citenamefont {Yang}}]{Qi1}%
  \BibitemOpen
  \bibfield  {author} {\bibinfo {author} {\bibfnamefont {P.}~\bibnamefont {Hayden}}, \bibinfo {author} {\bibfnamefont {S.}~\bibnamefont {Nezami}}, \bibinfo {author} {\bibfnamefont {X.-L.}\ \bibnamefont {Qi}}, \bibinfo {author} {\bibfnamefont {N.}~\bibnamefont {Thomas}}, \bibinfo {author} {\bibfnamefont {M.}~\bibnamefont {Walter}},\ and\ \bibinfo {author} {\bibfnamefont {Z.}~\bibnamefont {Yang}},\ }\bibfield  {title} {\bibinfo {title} {{Holographic duality from random tensor networks}},\ }\href {https://doi.org/10.1007/JHEP11(2016)009} {\bibfield  {journal} {\bibinfo  {journal} {JHEP}\ }\textbf {\bibinfo {volume} {11}},\ \bibinfo {pages} {009}},\ \Eprint {https://arxiv.org/abs/1601.01694} {arXiv:1601.01694 [hep-th]} \BibitemShut {NoStop}%
\bibitem [{\citenamefont {Calabrese}\ and\ \citenamefont {Cardy}(2004)}]{Calabrese:2004eu}%
  \BibitemOpen
  \bibfield  {author} {\bibinfo {author} {\bibfnamefont {P.}~\bibnamefont {Calabrese}}\ and\ \bibinfo {author} {\bibfnamefont {J.~L.}\ \bibnamefont {Cardy}},\ }\bibfield  {title} {\bibinfo {title} {{Entanglement entropy and quantum field theory}},\ }\href {https://doi.org/10.1088/1742-5468/2004/06/P06002} {\bibfield  {journal} {\bibinfo  {journal} {J. Stat. Mech.}\ }\textbf {\bibinfo {volume} {0406}},\ \bibinfo {pages} {P06002} (\bibinfo {year} {2004})},\ \Eprint {https://arxiv.org/abs/hep-th/0405152} {arXiv:hep-th/0405152 [hep-th]} \BibitemShut {NoStop}%
%%CITATION = HEP-TH/0405152;%%
\bibitem [{\citenamefont {Lewkowycz}\ and\ \citenamefont {Maldacena}(2013)}]{LM2013}%
  \BibitemOpen
  \bibfield  {author} {\bibinfo {author} {\bibfnamefont {A.}~\bibnamefont {Lewkowycz}}\ and\ \bibinfo {author} {\bibfnamefont {J.}~\bibnamefont {Maldacena}},\ }\bibfield  {title} {\bibinfo {title} {{Generalized gravitational entropy}},\ }\href {https://doi.org/10.1007/JHEP08(2013)090} {\bibfield  {journal} {\bibinfo  {journal} {JHEP}\ }\textbf {\bibinfo {volume} {08}},\ \bibinfo {pages} {090}},\ \Eprint {https://arxiv.org/abs/1304.4926} {arXiv:1304.4926 [hep-th]} \BibitemShut {NoStop}%
%%CITATION = ARXIV:1304.4926;%%
\bibitem [{\citenamefont {Dong}(2016)}]{Dong:2016fnf}%
  \BibitemOpen
  \bibfield  {author} {\bibinfo {author} {\bibfnamefont {X.}~\bibnamefont {Dong}},\ }\bibfield  {title} {\bibinfo {title} {{The Gravity Dual of Renyi Entropy}},\ }\href {https://doi.org/10.1038/ncomms12472} {\bibfield  {journal} {\bibinfo  {journal} {Nature Commun.}\ }\textbf {\bibinfo {volume} {7}},\ \bibinfo {pages} {12472} (\bibinfo {year} {2016})},\ \Eprint {https://arxiv.org/abs/1601.06788} {arXiv:1601.06788 [hep-th]} \BibitemShut {NoStop}%
\bibitem [{\citenamefont {Banihashemi}\ and\ \citenamefont {Jacobson}(2025)}]{Banihashemi:2024weu}%
  \BibitemOpen
  \bibfield  {author} {\bibinfo {author} {\bibfnamefont {B.}~\bibnamefont {Banihashemi}}\ and\ \bibinfo {author} {\bibfnamefont {T.}~\bibnamefont {Jacobson}},\ }\bibfield  {title} {\bibinfo {title} {{The enigmatic gravitational partition function}},\ }\href {https://doi.org/10.1007/s10714-024-03347-0} {\bibfield  {journal} {\bibinfo  {journal} {Gen. Rel. Grav.}\ }\textbf {\bibinfo {volume} {57}},\ \bibinfo {pages} {43} (\bibinfo {year} {2025})},\ \Eprint {https://arxiv.org/abs/2411.00267} {arXiv:2411.00267 [hep-th]} \BibitemShut {NoStop}%
\bibitem [{\citenamefont {Dong}\ \emph {et~al.}(2016)\citenamefont {Dong}, \citenamefont {Lewkowycz},\ and\ \citenamefont {Rangamani}}]{dche}%
  \BibitemOpen
  \bibfield  {author} {\bibinfo {author} {\bibfnamefont {X.}~\bibnamefont {Dong}}, \bibinfo {author} {\bibfnamefont {A.}~\bibnamefont {Lewkowycz}},\ and\ \bibinfo {author} {\bibfnamefont {M.}~\bibnamefont {Rangamani}},\ }\bibfield  {title} {\bibinfo {title} {{Deriving covariant holographic entanglement}},\ }\href {https://doi.org/10.1007/JHEP11(2016)028} {\bibfield  {journal} {\bibinfo  {journal} {JHEP}\ }\textbf {\bibinfo {volume} {11}},\ \bibinfo {pages} {028}},\ \Eprint {https://arxiv.org/abs/1607.07506} {arXiv:1607.07506 [hep-th]} \BibitemShut {NoStop}%
%%CITATION = ARXIV:1607.07506;%%
\bibitem [{\citenamefont {Colin-Ellerin}\ \emph {et~al.}(2021)\citenamefont {Colin-Ellerin}, \citenamefont {Dong}, \citenamefont {Marolf}, \citenamefont {Rangamani},\ and\ \citenamefont {Wang}}]{Colin-Ellerin:2020mva}%
  \BibitemOpen
  \bibfield  {author} {\bibinfo {author} {\bibfnamefont {S.}~\bibnamefont {Colin-Ellerin}}, \bibinfo {author} {\bibfnamefont {X.}~\bibnamefont {Dong}}, \bibinfo {author} {\bibfnamefont {D.}~\bibnamefont {Marolf}}, \bibinfo {author} {\bibfnamefont {M.}~\bibnamefont {Rangamani}},\ and\ \bibinfo {author} {\bibfnamefont {Z.}~\bibnamefont {Wang}},\ }\bibfield  {title} {\bibinfo {title} {{Real-time gravitational replicas: Formalism and a variational principle}},\ }\href {https://doi.org/10.1007/JHEP05(2021)117} {\bibfield  {journal} {\bibinfo  {journal} {JHEP}\ }\textbf {\bibinfo {volume} {05}},\ \bibinfo {pages} {117}},\ \Eprint {https://arxiv.org/abs/2012.00828} {arXiv:2012.00828 [hep-th]} \BibitemShut {NoStop}%
\bibitem [{\citenamefont {Dittrich}\ \emph {et~al.}(2024)\citenamefont {Dittrich}, \citenamefont {Jacobson},\ and\ \citenamefont {Padua-Arg{\"u}elles}}]{Dittrich:2024awu}%
  \BibitemOpen
  \bibfield  {author} {\bibinfo {author} {\bibfnamefont {B.}~\bibnamefont {Dittrich}}, \bibinfo {author} {\bibfnamefont {T.}~\bibnamefont {Jacobson}},\ and\ \bibinfo {author} {\bibfnamefont {J.}~\bibnamefont {Padua-Arg{\"u}elles}},\ }\bibfield  {title} {\bibinfo {title} {{de Sitter horizon entropy from a simplicial Lorentzian path integral}},\ }\href {https://doi.org/10.1103/PhysRevD.110.046006} {\bibfield  {journal} {\bibinfo  {journal} {Phys. Rev. D}\ }\textbf {\bibinfo {volume} {110}},\ \bibinfo {pages} {046006} (\bibinfo {year} {2024})},\ \Eprint {https://arxiv.org/abs/2403.02119} {arXiv:2403.02119 [gr-qc]} \BibitemShut {NoStop}%
\bibitem [{\citenamefont {Bodendorfer}\ and\ \citenamefont {Haneder}(2019)}]{Bodendorfer:2018csn}%
  \BibitemOpen
  \bibfield  {author} {\bibinfo {author} {\bibfnamefont {N.}~\bibnamefont {Bodendorfer}}\ and\ \bibinfo {author} {\bibfnamefont {F.}~\bibnamefont {Haneder}},\ }\bibfield  {title} {\bibinfo {title} {{Coarse graining as a representation change}},\ }\href {https://doi.org/10.1016/j.physletb.2019.03.020} {\bibfield  {journal} {\bibinfo  {journal} {Phys. Lett.}\ }\textbf {\bibinfo {volume} {B792}},\ \bibinfo {pages} {69} (\bibinfo {year} {2019})},\ \Eprint {https://arxiv.org/abs/1811.02792} {arXiv:1811.02792 [gr-qc]} \BibitemShut {NoStop}%
%%CITATION = ARXIV:1811.02792;%%
\bibitem [{\citenamefont {Han}\ and\ \citenamefont {Hung}(2017)}]{Han:2016xmb}%
  \BibitemOpen
  \bibfield  {author} {\bibinfo {author} {\bibfnamefont {M.}~\bibnamefont {Han}}\ and\ \bibinfo {author} {\bibfnamefont {L.-Y.}\ \bibnamefont {Hung}},\ }\bibfield  {title} {\bibinfo {title} {{Loop Quantum Gravity, Exact Holographic Mapping, and Holographic Entanglement Entropy}},\ }\href {https://doi.org/10.1103/PhysRevD.95.024011} {\bibfield  {journal} {\bibinfo  {journal} {Phys. Rev. D}\ }\textbf {\bibinfo {volume} {95}},\ \bibinfo {pages} {024011} (\bibinfo {year} {2017})},\ \Eprint {https://arxiv.org/abs/1610.02134} {arXiv:1610.02134 [hep-th]} \BibitemShut {NoStop}%
\bibitem [{\citenamefont {Han}(2019)}]{Han:2019emp}%
  \BibitemOpen
  \bibfield  {author} {\bibinfo {author} {\bibfnamefont {M.}~\bibnamefont {Han}},\ }\bibfield  {title} {\bibinfo {title} {{Semiclassical Behavior of Spinfoam Amplitude with Small Spins and Entanglement Entropy}},\ }\href {https://doi.org/10.1103/PhysRevD.100.084049} {\bibfield  {journal} {\bibinfo  {journal} {Phys. Rev. D}\ }\textbf {\bibinfo {volume} {100}},\ \bibinfo {pages} {084049} (\bibinfo {year} {2019})},\ \Eprint {https://arxiv.org/abs/1906.05536} {arXiv:1906.05536 [gr-qc]} \BibitemShut {NoStop}%
\bibitem [{\citenamefont {Kaminski}\ \emph {et~al.}(2010)\citenamefont {Kaminski}, \citenamefont {Kisielowski},\ and\ \citenamefont {Lewandowski}}]{KKL}%
  \BibitemOpen
  \bibfield  {author} {\bibinfo {author} {\bibfnamefont {W.}~\bibnamefont {Kaminski}}, \bibinfo {author} {\bibfnamefont {M.}~\bibnamefont {Kisielowski}},\ and\ \bibinfo {author} {\bibfnamefont {J.}~\bibnamefont {Lewandowski}},\ }\bibfield  {title} {\bibinfo {title} {{Spin-Foams for All Loop Quantum Gravity}},\ }\href {https://doi.org/10.1088/0264-9381/29/4/049502, 10.1088/0264-9381/27/9/095006} {\bibfield  {journal} {\bibinfo  {journal} {Class. Quant. Grav.}\ }\textbf {\bibinfo {volume} {27}},\ \bibinfo {pages} {095006} (\bibinfo {year} {2010})},\ \bibinfo {note} {[Erratum: Class. Quant. Grav.29,049502(2012)]},\ \Eprint {https://arxiv.org/abs/0909.0939} {arXiv:0909.0939 [gr-qc]} \BibitemShut {NoStop}%
%%CITATION = ARXIV:0909.0939;%%
\bibitem [{\citenamefont {Ding}\ \emph {et~al.}(2011)\citenamefont {Ding}, \citenamefont {Han},\ and\ \citenamefont {Rovelli}}]{generalize}%
  \BibitemOpen
  \bibfield  {author} {\bibinfo {author} {\bibfnamefont {Y.}~\bibnamefont {Ding}}, \bibinfo {author} {\bibfnamefont {M.}~\bibnamefont {Han}},\ and\ \bibinfo {author} {\bibfnamefont {C.}~\bibnamefont {Rovelli}},\ }\bibfield  {title} {\bibinfo {title} {{Generalized Spinfoams}},\ }\href {https://doi.org/10.1103/PhysRevD.83.124020} {\bibfield  {journal} {\bibinfo  {journal} {Phys.Rev.}\ }\textbf {\bibinfo {volume} {D83}},\ \bibinfo {pages} {124020} (\bibinfo {year} {2011})},\ \Eprint {https://arxiv.org/abs/1011.2149} {arXiv:1011.2149 [gr-qc]} \BibitemShut {NoStop}%
%%CITATION = ARXIV:1011.2149;%%
\bibitem [{\citenamefont {Engle}\ and\ \citenamefont {Pereira}(2009)}]{finite}%
  \BibitemOpen
  \bibfield  {author} {\bibinfo {author} {\bibfnamefont {J.}~\bibnamefont {Engle}}\ and\ \bibinfo {author} {\bibfnamefont {R.}~\bibnamefont {Pereira}},\ }\bibfield  {title} {\bibinfo {title} {{Regularization and finiteness of the Lorentzian LQG vertices}},\ }\href {https://doi.org/10.1103/PhysRevD.79.084034} {\bibfield  {journal} {\bibinfo  {journal} {Phys.Rev.}\ }\textbf {\bibinfo {volume} {D79}},\ \bibinfo {pages} {084034} (\bibinfo {year} {2009})},\ \Eprint {https://arxiv.org/abs/0805.4696} {arXiv:0805.4696 [gr-qc]} \BibitemShut {NoStop}%
%%CITATION = ARXIV:0805.4696;%%
\bibitem [{\citenamefont {Kaminski}(2010)}]{Kaminski:2010qb}%
  \BibitemOpen
  \bibfield  {author} {\bibinfo {author} {\bibfnamefont {W.}~\bibnamefont {Kaminski}},\ }\bibfield  {title} {\bibinfo {title} {{All 3-edge-connected relativistic BC and EPRL spin-networks are integrable}},\ }\href@noop {} {\  (\bibinfo {year} {2010})},\ \Eprint {https://arxiv.org/abs/1010.5384} {arXiv:1010.5384 [gr-qc]} \BibitemShut {NoStop}%
\bibitem [{\citenamefont {Barbero~G.}\ and\ \citenamefont {Villasenor}(2009)}]{BarberoG:2008dwr}%
  \BibitemOpen
  \bibfield  {author} {\bibinfo {author} {\bibfnamefont {J.~F.}\ \bibnamefont {Barbero~G.}}\ and\ \bibinfo {author} {\bibfnamefont {E.~J.~S.}\ \bibnamefont {Villasenor}},\ }\bibfield  {title} {\bibinfo {title} {{On the computation of black hole entropy in loop quantum gravity}},\ }\href {https://doi.org/10.1088/0264-9381/26/3/035017} {\bibfield  {journal} {\bibinfo  {journal} {Class. Quant. Grav.}\ }\textbf {\bibinfo {volume} {26}},\ \bibinfo {pages} {035017} (\bibinfo {year} {2009})},\ \Eprint {https://arxiv.org/abs/0810.1599} {arXiv:0810.1599 [gr-qc]} \BibitemShut {NoStop}%
\bibitem [{\citenamefont {Han}(2025{\natexlab{b}})}]{spinfoamstack1}%
  \BibitemOpen
  \bibfield  {author} {\bibinfo {author} {\bibfnamefont {M.}~\bibnamefont {Han}},\ }\bibfield  {title} {\bibinfo {title} {{Summation and localization of spinfoam amplitudes}},\ }\href@noop {} {\  (\bibinfo {year} {2025}{\natexlab{b}})}\BibitemShut {NoStop}%
\bibitem [{\citenamefont {H\"{o}rmander}(1998)}]{stationaryphase}%
  \BibitemOpen
  \bibfield  {author} {\bibinfo {author} {\bibfnamefont {L.}~\bibnamefont {H\"{o}rmander}},\ }\href {https://doi.org/10.1007/978-3-642-96750-4} {\emph {\bibinfo {title} {The Analysis of Linear Partial Differential Operators I}}},\ \bibinfo {series} {Grundlehren der mathematischen Wissenschaften}, Vol.\ \bibinfo {volume} {256}\ (\bibinfo  {publisher} {Springer Berlin Heidelberg},\ \bibinfo {address} {Berlin, Heidelberg},\ \bibinfo {year} {1998})\BibitemShut {NoStop}%
\bibitem [{\citenamefont {Colafranceschi}\ \emph {et~al.}(2024)\citenamefont {Colafranceschi}, \citenamefont {Dong}, \citenamefont {Marolf},\ and\ \citenamefont {Wang}}]{Colafranceschi:2023moh}%
  \BibitemOpen
  \bibfield  {author} {\bibinfo {author} {\bibfnamefont {E.}~\bibnamefont {Colafranceschi}}, \bibinfo {author} {\bibfnamefont {X.}~\bibnamefont {Dong}}, \bibinfo {author} {\bibfnamefont {D.}~\bibnamefont {Marolf}},\ and\ \bibinfo {author} {\bibfnamefont {Z.}~\bibnamefont {Wang}},\ }\bibfield  {title} {\bibinfo {title} {{Algebras and Hilbert spaces from gravitational path integrals. Understanding Ryu-Takayanagi/HRT as entropy without AdS/CFT}},\ }\href {https://doi.org/10.1007/JHEP10(2024)063} {\bibfield  {journal} {\bibinfo  {journal} {JHEP}\ }\textbf {\bibinfo {volume} {10}},\ \bibinfo {pages} {063}},\ \Eprint {https://arxiv.org/abs/2310.02189} {arXiv:2310.02189 [hep-th]} \BibitemShut {NoStop}%
\bibitem [{\citenamefont {Tobin}(2025)}]{HanTobin}%
  \BibitemOpen
  \bibfield  {author} {\bibinfo {author} {\bibfnamefont {S.}~\bibnamefont {Tobin}},\ }\bibfield  {title} {\bibinfo {title} {{Entanglement entropy in Loop Quantum Gravity and quantum error correction}},\ }\href@noop {} {\  (\bibinfo {year} {2025})}\BibitemShut {NoStop}%
\bibitem [{\citenamefont {Lin}\ and\ \citenamefont {Radicevic}(2020)}]{Lin:2018bud}%
  \BibitemOpen
  \bibfield  {author} {\bibinfo {author} {\bibfnamefont {J.}~\bibnamefont {Lin}}\ and\ \bibinfo {author} {\bibfnamefont {D.}~\bibnamefont {Radicevic}},\ }\bibfield  {title} {\bibinfo {title} {{Comments on defining entanglement entropy}},\ }\href {https://doi.org/10.1016/j.nuclphysb.2020.115118} {\bibfield  {journal} {\bibinfo  {journal} {Nucl. Phys. B}\ }\textbf {\bibinfo {volume} {958}},\ \bibinfo {pages} {115118} (\bibinfo {year} {2020})},\ \Eprint {https://arxiv.org/abs/1808.05939} {arXiv:1808.05939 [hep-th]} \BibitemShut {NoStop}%
\bibitem [{\citenamefont {Han}\ \emph {et~al.}(2023)\citenamefont {Han}, \citenamefont {Liu},\ and\ \citenamefont {Qu}}]{Han:2023cen}%
  \BibitemOpen
  \bibfield  {author} {\bibinfo {author} {\bibfnamefont {M.}~\bibnamefont {Han}}, \bibinfo {author} {\bibfnamefont {H.}~\bibnamefont {Liu}},\ and\ \bibinfo {author} {\bibfnamefont {D.}~\bibnamefont {Qu}},\ }\bibfield  {title} {\bibinfo {title} {{Complex critical points in Lorentzian spinfoam quantum gravity: Four-simplex amplitude and effective dynamics on a double-\ensuremath{\Delta_3} complex}},\ }\href {https://doi.org/10.1103/PhysRevD.108.026010} {\bibfield  {journal} {\bibinfo  {journal} {Phys. Rev. D}\ }\textbf {\bibinfo {volume} {108}},\ \bibinfo {pages} {026010} (\bibinfo {year} {2023})},\ \Eprint {https://arxiv.org/abs/2301.02930} {arXiv:2301.02930 [gr-qc]} \BibitemShut {NoStop}%
\bibitem [{\citenamefont {Magliaro}\ and\ \citenamefont {Perini}(2013)}]{claudio}%
  \BibitemOpen
  \bibfield  {author} {\bibinfo {author} {\bibfnamefont {E.}~\bibnamefont {Magliaro}}\ and\ \bibinfo {author} {\bibfnamefont {C.}~\bibnamefont {Perini}},\ }\bibfield  {title} {\bibinfo {title} {{Regge gravity from spinfoams}},\ }\href {https://doi.org/10.1142/S0218271813500016} {\bibfield  {journal} {\bibinfo  {journal} {Int.J.Mod.Phys.}\ }\textbf {\bibinfo {volume} {D22}},\ \bibinfo {pages} {1} (\bibinfo {year} {2013})},\ \Eprint {https://arxiv.org/abs/1105.0216} {arXiv:1105.0216 [gr-qc]} \BibitemShut {NoStop}%
%%CITATION = ARXIV:1105.0216;%%
\bibitem [{\citenamefont {Magliaro}\ and\ \citenamefont {Perini}(2011)}]{claudio1}%
  \BibitemOpen
  \bibfield  {author} {\bibinfo {author} {\bibfnamefont {E.}~\bibnamefont {Magliaro}}\ and\ \bibinfo {author} {\bibfnamefont {C.}~\bibnamefont {Perini}},\ }\bibfield  {title} {\bibinfo {title} {{Emergence of gravity from spinfoams}},\ }\href {https://doi.org/10.1209/0295-5075/95/30007} {\bibfield  {journal} {\bibinfo  {journal} {Europhys.Lett.}\ }\textbf {\bibinfo {volume} {95}},\ \bibinfo {pages} {30007} (\bibinfo {year} {2011})},\ \Eprint {https://arxiv.org/abs/1108.2258} {arXiv:1108.2258 [gr-qc]} \BibitemShut {NoStop}%
%%CITATION = ARXIV:1108.2258;%%
\bibitem [{\citenamefont {Bianchi}\ and\ \citenamefont {Ding}(2012)}]{propagator3}%
  \BibitemOpen
  \bibfield  {author} {\bibinfo {author} {\bibfnamefont {E.}~\bibnamefont {Bianchi}}\ and\ \bibinfo {author} {\bibfnamefont {Y.}~\bibnamefont {Ding}},\ }\bibfield  {title} {\bibinfo {title} {{Lorentzian spinfoam propagator}},\ }\href {https://doi.org/10.1103/PhysRevD.86.104040} {\bibfield  {journal} {\bibinfo  {journal} {Phys.Rev.}\ }\textbf {\bibinfo {volume} {D86}},\ \bibinfo {pages} {104040} (\bibinfo {year} {2012})},\ \Eprint {https://arxiv.org/abs/1109.6538} {arXiv:1109.6538 [gr-qc]} \BibitemShut {NoStop}%
%%CITATION = ARXIV:1109.6538;%%
\bibitem [{\citenamefont {Bianchi}\ \emph {et~al.}(2009)\citenamefont {Bianchi}, \citenamefont {Magliaro},\ and\ \citenamefont {Perini}}]{propagator2}%
  \BibitemOpen
  \bibfield  {author} {\bibinfo {author} {\bibfnamefont {E.}~\bibnamefont {Bianchi}}, \bibinfo {author} {\bibfnamefont {E.}~\bibnamefont {Magliaro}},\ and\ \bibinfo {author} {\bibfnamefont {C.}~\bibnamefont {Perini}},\ }\bibfield  {title} {\bibinfo {title} {{LQG propagator from the new spin foams}},\ }\href {https://doi.org/10.1016/j.nuclphysb.2009.07.016} {\bibfield  {journal} {\bibinfo  {journal} {Nucl.Phys.}\ }\textbf {\bibinfo {volume} {B822}},\ \bibinfo {pages} {245} (\bibinfo {year} {2009})},\ \Eprint {https://arxiv.org/abs/0905.4082} {arXiv:0905.4082 [gr-qc]} \BibitemShut {NoStop}%
%%CITATION = ARXIV:0905.4082;%%
\bibitem [{\citenamefont {Bianchi}\ \emph {et~al.}(2006)\citenamefont {Bianchi}, \citenamefont {Modesto}, \citenamefont {Rovelli},\ and\ \citenamefont {Speziale}}]{propagator1}%
  \BibitemOpen
  \bibfield  {author} {\bibinfo {author} {\bibfnamefont {E.}~\bibnamefont {Bianchi}}, \bibinfo {author} {\bibfnamefont {L.}~\bibnamefont {Modesto}}, \bibinfo {author} {\bibfnamefont {C.}~\bibnamefont {Rovelli}},\ and\ \bibinfo {author} {\bibfnamefont {S.}~\bibnamefont {Speziale}},\ }\bibfield  {title} {\bibinfo {title} {{Graviton propagator in loop quantum gravity}},\ }\href {https://doi.org/10.1088/0264-9381/23/23/024} {\bibfield  {journal} {\bibinfo  {journal} {Class.Quant.Grav.}\ }\textbf {\bibinfo {volume} {23}},\ \bibinfo {pages} {6989} (\bibinfo {year} {2006})},\ \Eprint {https://arxiv.org/abs/gr-qc/0604044} {arXiv:gr-qc/0604044 [gr-qc]} \BibitemShut {NoStop}%
%%CITATION = GR-QC/0604044;%%
\end{thebibliography}%

\end{document}